\documentclass[english,letter,11pt,twosided]{article}
\usepackage{eurosym}
\usepackage{amsfonts}
\usepackage[left=2.1cm,right=1.9cm, bottom = 2.0cm, top=2.0cm]{geometry}
\usepackage{slantsc}
\usepackage{babel}
\usepackage{array}
\usepackage{amsmath}
\usepackage{amsthm}
\usepackage{amssymb}
\usepackage{graphicx}
\usepackage{enumerate}
\usepackage{setspace}
\usepackage{flafter}
\usepackage{hyperref}
\usepackage{mathtools}
\usepackage{chngcntr}
\usepackage{apptools}
\usepackage{natbib}
\usepackage{bbm}
\usepackage{accents}
\usepackage{tikz}
\usepackage[normalem]{ulem}
\usepackage{url}
\usepackage{epigraph}
\usepackage{subcaption, mwe}
\usepackage[utf8]{inputenc}
\usepackage{utopia}
\usepackage[labelsep=period, labelfont = small]{caption}
\usepackage{hyperref}
\usepackage{hyperref}
\usepackage[blocks]{authblk}

\PassOptionsToPackage{hyphens}{url}

\AtBeginDocument{}

\newtheorem{theorem}{Theorem}
\newtheorem{proposition}{Proposition}
\newtheorem*{result*}{Result}

\newtheorem{lemma}{Lemma}

\newtheorem*{theorem*}{Theorem}

\theoremstyle{definition}

\AtAppendix{\counterwithin{proposition}{section}}
\AtAppendix{\counterwithin{lemma}{section}}
\AtAppendix{\counterwithin{assumption}{section}}

\graphicspath{{figures/}}

\makeatletter
\newcommand{\subalign}[1]{  \vcenter{    \Let@ \restore@math@cr \default@tag
    \baselineskip\fontdimen10 \scriptfont\tw@
    \advance\baselineskip\fontdimen12 \scriptfont\tw@
    \lineskip\thr@@\fontdimen8 \scriptfont\thr@@
    \lineskiplimit\lineskip
    \ialign{\hfil$\m@th\scriptstyle##$&$\m@th\scriptstyle{}##$\crcr
      #1\crcr
    }  }
}
\makeatother
\definecolor{PennBlue}{RGB}{001,031,091}
\definecolor{PennRed}{RGB}{153,0,0}
\hypersetup{
pdfborder = {0 0 0},
    colorlinks,
    citecolor=PennRed,
    filecolor=PennRed,
    linkcolor=PennRed,
    urlcolor=PennRed
}

\newcommand{\br}{\mathbb{R}}

\newcommand{\ba}{\begin{array}}
\newcommand{\ea}{\end{array}}

\input{tcilatex}
\begin{document}

\title{\fontsize{18}{18} \linespread{1.3} 
\selectfont{{How AI Aggregation Affects Knowledge}\thanks{We are grateful to numerous participants at the Applied and Computational Mathematics Seminar at Dartmouth College, the 2025 Annual Network Science in Economics Conference, the Tuck's AI/ML Seminar Series, and the EC'25 Workshop on LLMs and Information Economics.}
}}
\author{ \fontsize{13}{13} \and \fontsize{13}{13} Daron Acemoglu\thanks{%
Massachusetts Institute of Technology, NBER, and CEPR, \href{daron@mit.edu}{%
daron@mit.edu} } \, \, \, \, \, \, \, ~ \fontsize{13}{13} Tianyi Lin\thanks{%
Columbia University, \href{mailto:tl3335@columbia.edu}{tl3335@columbia.edu} }
\, \, \, \, \, \, \, ~ \fontsize{13}{13} Asuman Ozdaglar\thanks{%
Massachusetts Institute of Technology, \href{asuman@mit.edu}{asuman@mit.edu}}
\, \, \, \, \, \, \, ~ \fontsize{13}{13} James Siderius\thanks{%
Tuck School of Business at Dartmouth College, \href{james.siderius@tuck.dartmouth.edu}%
{james.siderius@tuck.dartmouth.edu}} \hfill}
\date{{\normalsize \today
}}
\maketitle

\begin{abstract}
\fontsize{10}{10}\selectfont\baselineskip0.5cm 

Artificial intelligence (AI) changes social learning when aggregated outputs become training data for future predictions. To study this, we extend the DeGroot model by introducing an AI aggregator that trains on population beliefs and feeds synthesized signals back to agents. We define the learning gap as the deviation of long-run beliefs from the efficient benchmark, allowing us to capture how AI aggregation affects learning. Our main result identifies a threshold in the speed of updating: when the aggregator updates too quickly, there is no positive-measure set of training weights that robustly improves learning across a broad class of environments, whereas such weights exist when updating is sufficiently slow. We then compare global and local architectures. Local aggregators trained on proximate or topic-specific data robustly improve learning in all environments. Consequently, replacing specialized local aggregators with a single global aggregator worsens learning in at least one dimension of the state.

\end{abstract}
\sloppy

\bigskip




\thispagestyle{empty} \newpage \setcounter{page}{1}\fontsize{11}{11}%
\selectfont\baselineskip0.61cm 

\setstretch{1.24}

\section{Introduction}

\label{section:introduction} In recent years, generative artificial
intelligence (GenAI) systems have become a leading interface through which
individuals search for, synthesize, and interpret information %
\citep{Cutler-2023-ChatGPT, Xu-2023-ChatGPT, Ayoub-2024-Head}. Unlike
traditional information intermediaries, these systems are trained directly
on large-scale collections of human-generated content and generate
(generally) unified responses to a wide range of queries. However, as GenAI
tools have become more widely adopted, their outputs have started to shape
the content later used for retraining 
\citep{Wang-2023-Survey,
Burtch-2024-Consequences, Burtch-2024-Generative}. This creates a feedback
loop in which AI systems ingest beliefs that they have themselves helped
generate, blurring the distinction between original information and
synthesized knowledge.

A centralized aggregator can in principle improve decision-making by collecting and combining information from many dispersed sources. Yet when training data reflect endogenous belief
formation in socially structured networks, aggregation can reshape not only
collective learning outcomes but also the distribution of epistemic
influence across groups. By combining and synthesizing population beliefs,
aggregation architectures implicitly determine which signals receive greater
weight in shaping AI output. If training data overrepresent certain groups
or viewpoints, the resulting system may amplify those signals even in the
absence of explicit discrimination. Thus, the central concern is not only
predictive performance, but how aggregators, via their training data and
responses, reallocate influence throughout human communities and interact
with social segregation, feedback, and uncertainty about the underlying
environment.

To study these forces, we build on the DeGroot model of belief dynamics
augmented with AI aggregation. The DeGroot model is characterized by a
directed graph, where each edge represents the influence of one agent over
the beliefs of another. This setting is attractive to study the learning
implications of AI aggregation. First, it provides a tractable framework for
the analysis of belief dynamics in the benchmark without AI aggregation.
Second, it formalizes the influence of the training weights of AI models in
a transparent manner --- corresponding to the weights that an AI aggregator
puts on the beliefs of different agents. Third, the influence of an AI
aggregator on each agent can also be similarly incorporated into this
setting, mapping directly to AI adoption. This formalization highlights that
an AI aggregator feeds synthesized signals, based on its training weights,
back into the network, creating feedback loops.

We focus on long-run learning and compare outcomes with and without AI
aggregation. When beliefs converge, we follow the literature and refer to
the common limiting belief as the \emph{consensus}. We evaluate this
consensus against an efficient benchmark --- the posterior mean that would
arise under frictionless aggregation of all private signals. The difference
between these objects, which we term the \emph{learning gap}, measures
mislearning induced by network structure and AI-mediated feedback. Because
consensus is a weighted average of initial signals, the learning gap
reflects not only aggregate efficiency loss but also distortions in the
effective influence weights assigned to heterogeneous agents.

Our first contribution is technical. We provide a closed-form
characterization of the long-run consensus induced by AI-mediated learning.
Building on perturbation methods in \citet{Schweitzer-1968-Perturbation}, we
show that introducing an AI aggregator into a DeGroot network yields a
consensus that can be written explicitly as a function of the original
network and a low-rank modification capturing AI training and feedback. This
representation expresses the learning gap in closed form and makes
transparent how aggregation reshapes influence weights. AI-mediated feedback
effectively alters the social weighting structure through which initial
information propagates.

To sharpen intuition, we specialize our setting to a stylized two-group
structure consisting of a majority island and a minority island. In
practice, these islands can correspond to ideological-distinct communities,
different geographies or demographic groups. Links are more likely within
islands than across islands, capturing the common pattern of \emph{homophily}
or group-level segregation. For example, peers attending a common university
are more likely to communicate and listen to others at that same university %
\citep{Mcpherson-2001-Birds}. This environment allows us to study how
homophily and feedback jointly determine learning outcomes.

When a global AI aggregator updates rapidly, its output closely tracks
current population beliefs. Because those beliefs already reflect
within-group reinforcement, especially within the majority group, the
aggregator trains on endogenously distorted data. Feeding this output back
into the population reinforces the same distortions, creating a recursive
feedback loop between beliefs and training data. In this regime, the impact
of an AI aggregator behaves less like information pooling and more like
amplification of existing social structure. 

We formalize this fragility by assuming that the environment (the true
network topology, the degree of segregation, and/or exact AI adoption
patterns) are not known with precision, so an AI aggregator has to perform
well across a range of ``plausible" environments. We ask whether there exist 
training weights that improve information aggregation in the presence of 
an AI aggregator relative to the benchmark without AI aggregation across 
a range of environments. Our main result establishes that as updating becomes 
faster, such \emph{robust} improvement becomes impossible. Here, robust 
improvement refers to improvement that holds across a class of networks 
and adoption patterns. When feedback is sufficiently strong, there is no 
positive-measure set of training weights that improves learning across 
admissible environments. Intuitively, rapid retraining repeatedly feeds 
AI-shaped beliefs back into the training data, reducing the effective 
diversity of independent information. The system ingests its own outputs. 
This mechanism parallels concerns described as \emph{model collapse}: 
Even with abundant data, learning quality deteriorates when data increasingly 
reflect model-generated content rather than independent signals 
\citep{Shumailov-2023-Curse, Gerstgrasser-2024-Model}. Speed 
couples the impact of a global AI aggregator too tightly with current 
population beliefs, which were themselves shaped by the same AI aggregator.
This feedback destroys robustness.

This fragility has direct implications for fairness and aggregation of
information in society. Because an AI aggregator reshapes effective
influence weights, different training regimes implicitly redistribute
epistemic power across groups. When environments differ in segregation or AI
adoption, the same training design can amplify some group's signals while
attenuating others'. Thus robustness and fairness are structurally linked:
The absence of a universally robust training weight implies that AI-based
aggregation inevitably embeds distributional trade-offs. Unlike standard
fairness notions based on predictive parity or classification error %
\citep{Hardt-2016-Equality, Kleinberg-2017-Inherent}, unfairness in our
framework arises from endogenous reweighting of influence rather than
disparate predictive error. Even when individual updating is symmetric and
no explicit discrimination occurs, the presence of an AI aggregator
systematically shifts whose information drives collective belief. The same
feedback mechanism that generates aggregate fragility also produces
distributional distortions in epistemic influence.

We further characterize asymmetries between majority- and minority-weighted
training. When training disproportionately reflects the majority island,
data imbalance and social segregation reinforce one another: Majority
beliefs already receive excess weight through within-group reinforcement,
and majority-weighted training compounds this distortion. Learning
deteriorates monotonically as homophily increases. By contrast, when 
training places greater weight on minority beliefs, AI can initially counteract 
baseline majority dominance, but its impact is non-monotone: with moderate 
segregation, minority bias protects minority information long enough to 
discipline the consensus, while with high segregation the same minority bias 
is amplified by AI-mediated feedback. Correcting underrepresentation
is therefore not simply a matter of reweighting data; it interacts endogenously 
with network structure and feedback. Even well-intentioned interventions 
can fail when the social environment is imperfectly understood.

Finally, we study an alternative architecture in which information
aggregators are local and topic-specific. Rather than pooling beliefs into a
single global system, the local aggregator model introduces multiple
intermediaries (e.g., local newspapers or community-based websites) trained
on restricted subsets of agents informative about specific topics. Each
local aggregator exerts stronger influence within its constituency than
across groups, and own effects dominate cross effects. This localization
compartmentalizes feedback: Errors in one dimension do not automatically
propagate to others, and informational diversity is preserved even under
rapid updating. As a result, local aggregators robustly improve learning
relative to the benchmark with no such aggregators. However, replacing
specialized local aggregators with a global aggregator necessarily couples
previously separate feedback loops and worsens learning along at least
one dimension.

The key design question is therefore not whether AI aggregates information,
but how broadly it does so. An AI aggregator that pools beliefs across the
entire population broadens the base of information, but also creates
feedback loops, ultimately exacerbating the influence of some groups and
rendering learning fragile. In contrast, architectures that restrict
training to more localized or topic-relevant subsets preserve informational
diversity and compartmentalize feedback, improving robustness even under
rapid updating.

\paragraph{Related Literature.}

Our model has built on the foundational literature on DeGroot learning and
networked information aggregation 
\citep{Degroot-1974-Reaching,
Bala-1998-Learning, Demarzo-2003-Persuasion, Golub-2010-Naive,
Acemoglu-2010-Spread, Acemoglu-2011-Opinion}. These results demonstrate that
decentralized social learning can aggregate dispersed information
effectively under standard conditions. For example, \citet{Golub-2010-Naive}
show that in large networks, beliefs converge arbitrarily close to the truth
so long as influence is sufficiently diffuse. Subsequent works extend these
results to settings with sparse signals \citep{Banerjee-2021-Naive} and
richer belief updating rules \citep{Jadbabaie-2012-Non}. A complementary
strand demonstrates that networked learning can systematically fail. %
\citet{Acemoglu-2010-Spread} show that the presence of agents who remain
anchored to initial beliefs can prevent efficient aggregation, leading to
enduring belief distortions. More recently, \citet{Bohren-2021-Learning}
show that even without stubbornness, misspecified updating rules can
generate systematic long-run errors. Our results align with this second
strand, but identify a distinct mechanism: Mislearning arises not from
individual stubbornness or incorrect inference, but from introducing an
aggregator whose training data are endogenous and shaped by beliefs it
previously influenced.

Our work is related to, but distinct from, models of stubborn or influential
agents 
\citep{Acemoglu-2013-Opinion, Yildiz-2013-Binary,
Ghaderi-2013-Opinion, Hunter-2022-Optimizing, Mostagir-2022-Society}. In
those models, mislearning typically arises because some agents do not fully
update or engage in sustained persuasion, which often leads to persistent
disagreement or polarization rather than full consensus. In contrast, in our
framework beliefs converge to a unique consensus. However, that consensus
can still be distorted, because an AI aggregator endogenously reshapes the
effective weights placed on initial information via feedback loops. As a
result, our paper highlights a specific form of inefficiency due to the
reweighting of information induced by an AI aggregator itself --- rather than
those rooted in stubbornness or disagreement, emphasized in the previous
literature.

Another related literature studies how homophily and network structure shape
opinion dynamics 
\citep{Friedkin-1990-Social, Deffuant-2000-Mixing,
Golub-2012-Homophily, Mostagir-2023-Social, Grabisch-2023-Design}. These
papers show that segregation can distort information aggregation even when
agents update na\"{\i}vely. Our contribution differs in two key aspects.
First, we introduce an explicit aggregator node that collects and
redistributes beliefs, altering the direction and intensity of information
flows. Second, rather than studying segregation in isolation, we specify how
segregation interacts with training imbalance, updating speed, and
aggregation architecture, distinguishing settings where an AI aggregator
mitigates network distortions from those where it amplifies them.

Finally, our paper connects to emerging empirical and computational work on
large language models and their interactions with humans and with one
another 
\citep{Argyle-2023-Out, Park-2022-Social, Park-2023-Generative,
Fu-2023-Improving, Leng-2023-LLM, Xiong-2023-Examining, Chan-2024-Chateval,
Du-2024-Improving, Filippas-2024-Large, Liang-2024-Encouraging,
Papachristou-2025-Network, Chang-2025-LLM}. While this literature documents
emergent behaviors and network effects among LLMs, it is empirical and does
not provide a theory of long-run learning under feedback. Our key
contribution is to offer a theoretical framework that formalizes concerns
often described informally as model/knowledge collapse %
\citep{Shumailov-2024-AI, Dohmatob-2024-Tale, Peterson-2025-AI}: When AI
systems retrain rapidly on data they have themselves influenced, the
effective diversity of information can shrink and learning can fail in
large populations. By connecting this phenomenon to classical results in
social learning, we clarify when and why centralized AI-based information
aggregation improves or undermines collective knowledge.

\paragraph{Paper Outline.}

Section~\ref{sec:models} introduces the social learning model with a single
global AI aggregator. Section~\ref{sec:general} establishes the closed-form
learning gap for general social networks. Section~\ref{sec:speed_results} 
specializes our model to a two-island setup and studies whether an AI aggregator 
can robustly improve learning. Section~\ref{sec:network_results} analyzes 
how segregation and training imbalance interact. Section~\ref{sec:local_aggregators} 
introduces local, topic-specific aggregators, and compares their effects to 
those of a global aggregator. We conclude in Section~\ref{sec:conclusion}. 
Proofs are presented in the appendix sections.

\section{Model}

\label{sec:models}

We study social learning in a population of $n$ agents indexed by 
$i \in \{1, \ldots, n\}$ who seek to learn an unknown scalar state $\theta 
\in \mathbb{R}$. Time is discrete and runs from $t=0$ to infinity. Each 
agent $i$ observes a single private signal $s_i = \theta+\varepsilon_i$, 
where $\{\varepsilon_i\}_{i=1}^n$ are independent, zero-mean noise terms 
with finite variance, at time $t=0$. There are no external signals thereafter. 
Agents update beliefs over time by observing others' beliefs through 
a social network and, when present, by observing the output of an aggregator. 

Because private signals are unbiased and equally informative, we use 
the simple average of all private signals as the efficient benchmark:
\begin{equation*}
\hat{\theta} \equiv \tfrac{1}{n}\sum_{i=1}^n s_i = \tfrac{1}{n}\sum_{i=1}^n p_i(0).
\end{equation*}%
This benchmark corresponds to frictionless aggregation of all private 
information and serves as a reference point for evaluating learning 
outcomes.

\paragraph{Baseline social learning.}

Let $p_{i}(t)$ denote agent $i$'s belief about $\theta $ at time $t$, and
let $p(t)=(p_{1}(t),\ldots ,p_{n}(t))^{\top }$. In the baseline, beliefs evolve according to the benchmark DeGroot learning
rule, which takes the form%
\begin{equation*}
p(t+1)=Tp(t),
\end{equation*}%
where $T\in \mathbb{R}^{n\times n}$ is a row-stochastic matrix describing
the network and accounts for an attention or trust matrix. The entry $T_{ij}$ 
records how much weight agent $i$ places on agent $j$'s current belief.
For example, if agent $i$ forms beliefs by listening to friends, coworkers, 
local media, or members of the same community, then the row $T_{i}$ 
summarizes how these sources are weighted. 

We assume that $T$ is strongly connected and aperiodic. Under these
conditions,~\citet{Golub-2010-Naive} show that beliefs converge to a common
limit: There exists a scalar $p^{\star }$ such that 
\begin{equation*}
\lim_{t\rightarrow \infty }p_{i}(t)=p^{\star },\quad \textnormal{for all } i.
\end{equation*}%
Throughout this paper, we refer to $p^{\star }$ as the \emph{consensus
without aggregators} (to contrast with the consensus \emph{with} aggregators, described below). This consensus reflects the long-run belief generated
by decentralized social learning alone.

\paragraph{Social learning with a global AI aggregator.}

We introduce an AI aggregator, modeled as an information intermediary that
produces a single observable signal based on current population beliefs and
feeds this signal back into the network. At each time $t$, the aggregator
forms a weighted average of agents' beliefs: $m(t)=\sum_{i=1}^{n}\alpha
_{i}p_{i}(t)$, where $\alpha =(\alpha _{1},\dots ,\alpha _{n})$ is a $%
1\times n$ vector of non-negative weights satisfying $\sum_{i=1}^{n}\alpha
_{i}=1$. The training weights $\alpha_i$ capture how strongly the beliefs 
of different agents or groups are represented in the data used to train or 
fine-tune the aggregator. Unequal weights may arise because some groups 
generate more content, are more visible online, receive more engagement, 
are more extensively digitized, or are deliberately reweighted by a platform. 

We initialize the aggregator with an \textit{uninformed} seed,
which is similar to how \cite{Banerjee-2021-Naive} initialize uninformed
agents in their model of na\"{\i}ve learning. This initialization implies
that $a(1)=m(0)$ and $p(1)=Tp(0)$, so that the AI aggregator's output is
shaped by the beliefs of the agents in the population that it places
positive training weight on. Thereafter, this output $a(t)\in \mathbb{R}$
evolves according to 
\begin{equation*}
a(t+1) = \rho a(t)+(1-\rho )m(t),\quad \textnormal{for all } t \geq 1,
\end{equation*}%
where $\rho \in (0,1)$ measures how quickly the aggregator 
refreshes in response to endogenously evolving population beliefs. 
A lower value of $\rho$ places more weight on current population 
beliefs, while a higher value places more weight on the aggregator’s 
past output.

Agents incorporate the output of the AI aggregator into their beliefs with 
varying weights. In particular, once the aggregator is available, population beliefs 
evolve according to 
\begin{equation*}
p_{i}(t+1)=(1-\beta _{i})\sum_{j=1}^{n}T_{ij}p_{j}(t)+\beta _{i}a(t), \quad 
\textnormal{for all } t \geq 1,
\end{equation*}%
where $\beta _{i}\in (0,1)$ measures the extent to which agent $i$ relies on
the aggregator output for all $i$. Under similar regularity conditions to~\citet{Golub-2010-Naive} (see
Proposition~\ref{prop:convergence}), beliefs again converge to a common
limit: There exists a scalar $p^{\star \star }$ such that 
\begin{equation*}
\lim_{t\rightarrow \infty }p_{i}(t)=p^{\star \star },\quad \textnormal{for all } i.
\end{equation*}%
We refer to $p^{\star \star }$ as the \emph{consensus with a global AI
aggregator}.

\paragraph{Learning performance and learning gap.}

We evaluate learning by comparing long-run consensus beliefs to 
the efficient benchmark $\hat{\theta}$ defined above. Accordingly, we define 
the learning gaps without and with AI as
\begin{equation*}
\Delta_0 \equiv |p^\star-\hat{\theta}|,
\qquad
\Delta_1 \equiv |p^{\star\star}-\hat{\theta}|,
\end{equation*}%
where $p^\star$ and $p^{\star\star}$ denote the long-run consensuses 
without and with AI aggregation. The learning gap measures the extent of
mislearning: it is zero if and only if decentralized learning fully aggregates 
private information, and it is positive whenever the consensus is away 
from the efficient benchmark. Throughout the paper, we say AI aggregation 
improves learning when $\Delta_1<\Delta_0$ and worsens learning 
when $\Delta_1>\Delta_0$. \bigskip 

\noindent \emph{Remark} --- For expositional clarity, we focus in this section 
on a scalar state. The analysis extends to a multi-dimensional state, 
with learning occurring componentwise along each dimension. In
Section~\ref{sec:local_aggregators}, we develop this extension and allow
different subsets of agents to be differentially informed about distinct
topics.

\section{General Network Models}

\label{sec:general}

We first establish general results for arbitrary networks. In particular, we
provide sufficient conditions under which beliefs will converge to a common
limit when an aggregator is present. We then derive a closed-form
characterization of the long-run consensus and the associated learning gap
for any network structure. These results serve as the workhorse for the
remainder of the analysis.

\subsection{Convergence of Beliefs}

We begin by deriving conditions under which beliefs converge in the presence
of a global AI aggregator. Recall that $T$ denotes the matrix governing
social learning among agents and let $\Gamma$ denote the augmented transition matrix given by:
\begin{equation*}
\Gamma = \left(\begin{array}{cc}
\rho & (1-\rho) \alpha \\
\beta & \mathnormal{Diag}(1-\beta) T
\end{array}\right), 
\end{equation*}
where $\alpha \in \mathbb{R}^{1\times n}$ is the training weight vector and $\beta \in \mathbb{R}^{n\times 1}$ is the AI adoption vector. 

\begin{proposition}
\label{prop:convergence} Suppose that $T$ is strongly connected and
aperiodic. Then, the augmented transition matrix $\Gamma$ is strongly
connected and aperiodic if: (i) $\rho \in (0,1)$, (ii) $\beta_i < 1$ for all 
$i$, and (iii) $\sum_{i=1}^n \beta_i > 0$.
\end{proposition}

Proposition~\ref{prop:convergence} provides simple sufficient conditions 
for convergence. Indeed, Condition (i) ensures that the AI aggregator does 
not create an absorbing node disconnected from the population: With probability 
$1-\rho >0$, the aggregator's next output depends on current beliefs through 
$\alpha$. Condition (ii) guarantees that agents continue to place positive weight 
on social learning each period, so the strong connectivity of $T$ is inherited by 
the agent-based subgraph in the augmented system. Condition (iii) rules out 
the degenerate case in which no agent ever relies on AI, in which case the 
additional node is irrelevant for learning dynamics.

Under these conditions, $\Gamma $ is a row-stochastic matrix describing a
finite-state Markov chain on $n+1$ nodes that is strongly connected and
aperiodic. By the Perron-Frobenius theorem for primitive stochastic
matrices, $\Gamma $ admits a unique stationary distribution $\pi \in \Delta
^{n+1}$ on the augmented state space, and $\Gamma ^{t}\rightarrow \mathbf{1}%
_{n+1}\pi $ as $t\rightarrow \infty $. Here and throughout, $\mathbf{1}_{k}$
is the $k$-dimensional column vector of ones. Consequently, for any initial
condition $p(0)$, beliefs converge to a common limit: There exists a scalar $%
p^{\star \star }$ such that 
\begin{equation*}
a(t)\rightarrow p^{\star \star }\quad \textnormal{and}\quad
p_{i}(t)\rightarrow p^{\star \star }\quad \text{for all }i.
\end{equation*}%
where $p^{\star \star }$ is the \emph{consensus with the AI aggregator}, as
defined above.

\subsection{Characterization of the Long-Run Consensus}

\label{sec:characterization}

We next provide a closed-form characterization of the consensus with 
a global AI aggregator.

\begin{theorem}
\label{thm:general} Suppose that $\rho \in (0, 1)$ and $\beta_i \in (0, 1)$ 
for all $i$. Then, the consensus with an AI aggregator satisfies 
\begin{equation}
p^{\star \star }=\tfrac{1}{1+z\mathbf{1}_{n}}(\alpha +zT)p(0).
\end{equation}%
where $z=(1-\rho )\alpha (\mathbf{I}_{n}-(\mathbf{I}_{n}-\mathnormal{Diag}%
\,(\beta ))T)^{-1}$ and $\mathbf{I}_{n}$ is a $n\times n$ identity matrix.
\end{theorem}

Theorem~\ref{thm:general} exploits the linear structure of the learning 
dynamics. In the absence of AI aggregation, DeGroot learning converges 
to a weighted average of initial beliefs determined by the stationary distribution 
of $T$. Introducing a global AI aggregator creates an endogenous feedback 
loop: current beliefs influence the aggregator's output through the training 
weights $\alpha \in \br^{1 \times n}$, and this output in turn enters future belief 
updates with intensities $\beta \in \br^{n\times 1}$. Rather than solving directly 
for the stationary distribution of the augmented system, the proof uses perturbation 
arguments for finite Markov chains \citep{Schweitzer-1968-Perturbation}. 
Mathematically, the aggregator induces a low-rank modification of baseline 
DeGroot dynamics, and the resulting closed-form consensus reveals how 
AI-mediated feedback reweights the influence of initial information.

The expression shows that the final consensus can be interpreted as a weighted 
average of agents' initial beliefs, where the weights reflect both direct persistence 
and AI-mediated aggregation through the network. The term $\alpha$ captures 
how much each agent's own prior continues to matter, while the term $zT$ captures 
how the AI aggregates information across the network and redistributes it back to 
agents. The scalar normalization ensures these weights sum to one. Economically, 
the AI aggregator reshapes influence: rather than beliefs diffusing purely through 
the network, the AI reweights and amplifies certain information paths, so that 
an agent's impact on the final consensus depends both on their position in the 
network and on how the aggregator processes and feeds information back into 
the population.

\section{How the Speed of AI Updating Affects Learning}

\label{sec:speed_results} 

In this section, we specialize the analysis to the two-island model and 
ask whether there exist training weights that improve learning not just for 
one fixed environment, but across a range of admissible values of homophily 
and AI reliance. This is our notion of robust improvement.

Specializing the analysis to the two-island model serves two purposes. 
First, it isolates in a minimal way how group-level asymmetries in
representation and adoption interact with feedback to shape learning.
Second, it provides a parsimonious environment in which heterogeneity is
coarse but economically meaningful, allowing us to derive sharp fragility
and mislearning results that would be obscured in fully general networks.
\begin{figure}[!t]
\centering
\includegraphics[width=100mm]{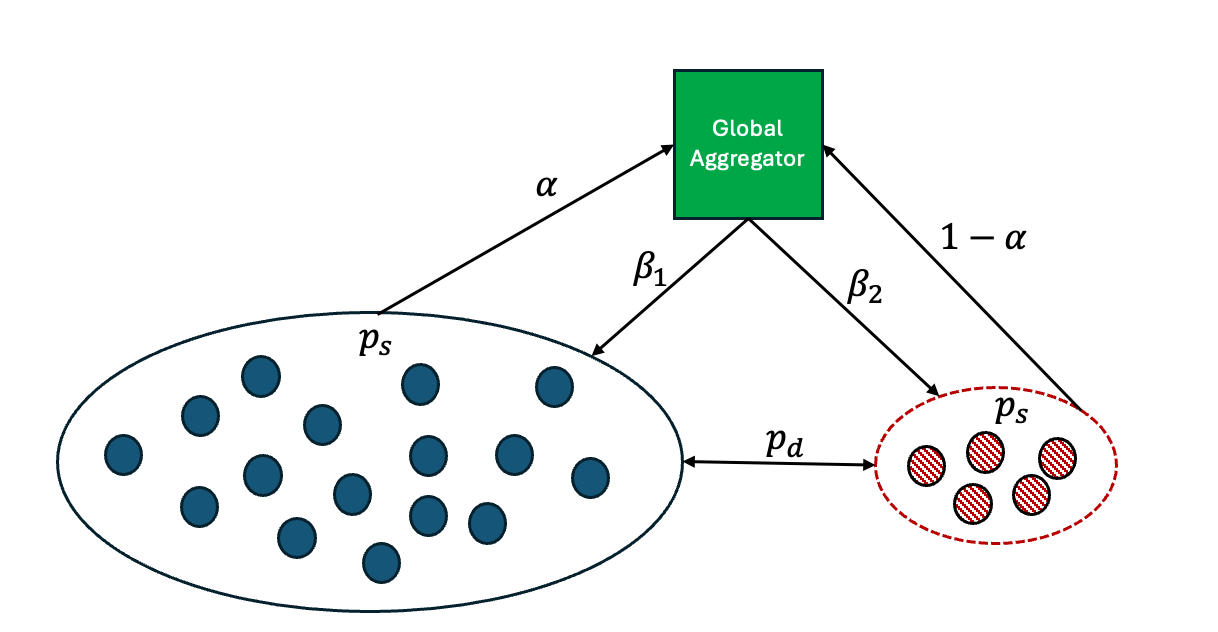}
\caption{Global aggregator architecture.}
\label{fig:global-agg}
\end{figure}

\paragraph{Model.}

Agents are partitioned into two types, which we refer to as \emph{islands}.
Islands may correspond to ideological camps, geographic regions, demographic
groups, or any salient dimension along which social interactions are more
likely within than across groups. Agents of the same type are connected with
probability $p_{s}\in (0,1)$, while agents of different types are connected
with probability $p_{d}<p_{s}$. The ratio $h=p_{s}/p_{d}>1$ captures the
degree of homophily in the social network. Larger values of $h$ correspond
to more segregated communication structures, while $h\rightarrow 1$ recovers
a well-mixed population. There are $n_{1}$ agents on island 1
(\textquotedblleft majority") and $n_{2}=n-n_{1}$ agents on island 2
(\textquotedblleft minority"). We summarize relative group size by $\pi
=n_{1}/n_{2}\in (1,\infty )$.

The two-island model is the simplest network that features within-group
reinforcement, cross-group information flow, and systematic asymmetries in
representation in training data. These features are central to the operation
of the AI aggregator in practice, where training data often overrepresent
some groups and adoption varies across the population. Various qualitative
properties of richer networks --- including echo chambers, amplification of
majority views, and underrepresentation of minority signals --- can be seen
in this two-group structure.

With two islands, the high-dimensional objects $(T,\alpha ,\beta )$ reduce
to a small number of interpretable parameters, as illustrated in Figure~\ref%
{fig:global-agg}. We also let $\alpha \in \lbrack 0,1]$ denote the share of
training weight placed on the majority island, with $1-\alpha $ placed on
the minority island. We also let $\beta _{1},\beta _{2}\in (0,1)$ capture
the reliance on the AI aggregator by the agents in the two islands,
respectively. Then, the expected interaction matrix reduces to the $2\times 2
$ matrix as follows, 
\begin{equation*}
F=%
\begin{pmatrix}
\frac{h\pi }{h\pi +1} & \frac{1}{h\pi +1} \\ 
\frac{\pi }{h+\pi } & \frac{h}{h+\pi }%
\end{pmatrix}%
,
\end{equation*}%
where each entry gives the expected weight an agent places on opinions
originating from each island. The matrix $F$ encapsulates a simple form of
within-group reinforcement in learning (each individual puts more weight
onmembers of its own island) and abstracts from idiosyncratic network
realizations (the structure of connections is symmetric within islands).
This island setup makes explicit the three channels through which the AI
aggregator affects learning: (i) data representation, captured by $\alpha $;
(ii) adoption and reliance, captured by $(\beta _{1},\beta _{2})$; and (iii)
social amplification, governed by homophily $h$ and relative group sizes $%
\pi $. Accordingly, the learning gaps without and with the AI aggregator are
given by $\Delta _{0}(h,\pi )$ and $\Delta _{1}(\rho ,\alpha ,\beta
_{1},\beta _{2},h,\pi )$. Throughout this section, we define 
\begin{equation}
\Delta^{\star} := \Delta_{1}(\rho ,\alpha ,\beta _{1},\beta _{2},h,\pi
)-\Delta _{0}(h,\pi ),
\end{equation}%
which measures how the AI aggregator changes the learning gap relative to
decentralized learning alone. Thus, $\Delta^{\star }<0$ indicates that the
aggregator improves learning, while $\Delta^{\star }>0$ indicates that it
worsens learning.

\paragraph{Fragility of AI aggregation.}

We now study how the speed of updating affects the robustness of information
aggregation with AI. Throughout this subsection, we fix the relative size of the 
two groups $\pi>1$ and consider variation along two dimensions. First, 
the degree of homophily $h$ is assumed to vary over a compact interval 
$[\underaccent{\bar}{h},\bar{h}]$, where $\underaccent{\bar}{h},\bar{h}$ 
are finite and satisfy certain conditions. Second, agents' reliance on the AI
aggregator is allowed to vary across groups, with $(\beta _{1},\beta_{2}) %
\in (0,1)^2$. We define
\begin{equation*}
\Lambda_{\rho} := \{\alpha \in [0,1] \mid \Delta^\star(\rho,\alpha,\beta_1,\beta_2,h,\pi) 
< 0 \textnormal{ for all } h \in [\underaccent{\bar}{h},\bar{h}] \textnormal{ and for all } 
(\beta_1,\beta_2)\in(0,1)^2\}.
\end{equation*}
Thus, $\Lambda_{\rho}$ is the set of training weights that improve learning 
relative to the standard benchmark across this range of environments. We refer 
to $\Lambda_{\rho}$ as the \emph{robust improvement} set. We focus on 
the robust improvement set because it is not reasonable to imagine that AI model 
parameters can be finely tuned exactly to the pattern of homophily and the precise 
usage patterns of different groups in society. With this focus, we require 
the AI aggregator to perform well across a range of environments.

\begin{theorem}
\label{thm:rho} Fix $\pi >1$ and $\underaccent{\bar}{h},\bar{h}$ such that $%
\underaccent{\bar}{h}>2\pi $, $\bar{h}>20\pi $ and $\bar{h}>%
\underaccent{\bar}{h}$.\footnote{%
The conditions $\underaccent{\bar}{h}>2\pi $ and $\bar{h}>20\pi $ are
sufficient bounds that ensure the two-island structure exhibits meaningful
segregation and majority amplification; they are not necessary and are
imposed to simplify the analysis.} Then, there exists a threshold $%
\rho^\star :=\rho^\star(\pi ,\underaccent{\bar}{h},\bar{h}) \in (0,1)$ such
that

\begin{enumerate}
\item if $\rho <\rho^\star$, then the robust improvement set is 
zero-measure: $\mu(\Lambda _{\rho})=0$;

\item if $\rho >\rho^\star$, then the robust improvement set is 
positive-measure: $\mu(\Lambda _{\rho})>0$.
\end{enumerate}
\end{theorem}

Theorem~\ref{thm:rho} highlights that the scope for robust improvement 
depends on updating speed. When updating is sufficiently fast, the robust 
improvement set $\Lambda_{\rho}$ is zero-measure; when updating is 
sufficiently slow, $\Lambda_{\rho}$ is positive-measure. The intuition is 
that fast updating strengthens the feedback loop between current beliefs 
and future training data. Because the current beliefs already reflect homophily 
and within-group reinforcement, an aggregator that closely tracks them 
feeds the same distortions back into the population, and the resulting amplification 
depends sensitively on the realized network and AI-reliance profile. This leaves 
little room for training weights that improve learning robustly across admissible 
environments. By contrast, slow updating weakens this loop: the aggregator 
responds to a smoother history of beliefs rather than the current distorted 
cross-section, so bias is less tightly fed back into training data. In that regime, 
a nontrivial range of training weights can offset homophily across admissible 
environments, implying $\mu(\Lambda_{\rho})>0$. Theorem~\ref{thm:rho} 
therefore identifies a tradeoff between speed and robustness: faster updating 
can make robust improvement harder.

\bigskip\noindent\emph{Remark} --- While Theorem~\ref{thm:rho} focuses on 
robust improvement, this criterion is motivated by the fact that, in practice, 
network structure and patterns of AI reliance are typically not known precisely 
and may vary across settings. By contrast, Appendix~\ref{app:fixed_network} 
studies learning in a fixed and fully specified environment, allowing for a more 
detailed characterization of how the aggregator shapes information aggregation 
when these features are known.

\section{AI-Network Interaction on Learning}

\label{sec:network_results}

In this section, we isolate how segregation and training imbalance interact, 
holding the pattern of AI reliance symmetric across groups. For this reason, 
we now impose $\beta_1=\beta_2=\beta$ and focus on the comparative 
statics of the learning gap with respect to network segregation. We also 
distinguish between two empirically and conceptually relevant training regimes: 
one in which the AI aggregator places substantial weight on the majority island, 
and one in which it places relatively greater weight on the minority island.

\subsection{Strong Majority Bias}

We begin with a regime in which the AI aggregator places substantial weight
on the majority island in its training data. This case captures environments
in which data availability, visibility, or engagement are systematically
skewed toward a dominant group. For example, platforms where majority users
generate disproportionate volumes of content, or an AI aggregator is trained
primarily on data from high-activity populations. In such environments, the
AI aggregator does not merely reflect existing social biases; it risks
amplifying them.

\begin{proposition}
\label{prop:majority} Suppose that $\alpha > \frac{\pi^2}{\pi^2+1}$. Then,
we have $\Delta^\star>0$, and $\Delta_1$ is monotonically increasing in the
degree of homophily $h$.
\end{proposition}

Proposition~\ref{prop:majority} shows that when  $\alpha>\pi^{2}/(\pi^{2}+1)$, 
majority-weighted training worsens learning relative to the standard social dynamics, 
and the learning gap increases monotonically with segregation. In this regime, 
the aggregator places too much weight on majority beliefs relative to the efficient 
benchmark. As segregation rises, majority opinions are reinforced more strongly 
within the dominant island before reaching the minority; feeding these beliefs 
into a majority-weighted aggregator then amplifies the same distortion. Thus, 
segregation and training imbalance reinforce one another: When training is 
sufficiently tilted toward the majority, greater segregation never improves 
learning. From a design perspective, Proposition~\ref{prop:majority} underscores 
that correcting data imbalance is not merely a fairness concern but a robustness
requirement. When training data disproportionately reflect majority groups,
greater segregation unambiguously worsens learning in the presence of a
global aggregator.

\subsection{Minority Bias}

Can biasing the AI aggregator's training weights in favor of the minority
group correct this bias? We next answer this question by considering the
opposite regime, in which the global AI aggregator places greater weight on
the minority island. This captures environments where AI models are
deliberately designed to counteract majority dominance through reweighting
schemes, fairness constraints, or targeted data collection. The effects of 
minority bias are more subtle than those of majority bias. Indeed, minority-weighted 
training can counteract the baseline tendency of segregated networks to 
overweight majority beliefs. However, doing so introduces a new tension: 
correcting one source of bias can lead to overcorrection once feedback 
and social learning are taken into account. As a result, the interaction between 
minority bias and network structure is inherently non-monotone.

\begin{proposition}
\label{prop:minority} There exists $\beta^\star>0$ such that if $\alpha<%
\frac{1}{2}$ and $\beta<\beta^\star$, then the sign of $\Delta^\star$ is
ambiguous and its dependence on $h$ is non-monotone. In particular, there
exist $1<\underaccent{\bar}{h}<\bar{h}<\infty$ such that:

\begin{enumerate}
\item $\Delta^\star>0$ and $\Delta_1$ is decreasing in $h$ over $(1,%
\underaccent{\bar}{h})$;

\item $\Delta^\star<0$ and $\Delta_1$ is non-monotone in $h$ over $(%
\underaccent{\bar}{h},\bar{h})$;

\item $\Delta^\star>0$ and $\Delta_1$ is increasing in $h$ over $(\bar{h}%
,\infty)$.
\end{enumerate}
\end{proposition}

Proposition~\ref{prop:minority} shows that minority-weighted training 
improves learning only at intermediate levels of segregation. When segregation 
is low, placing extra weight on minority signals can over-correct and push 
the long-run consensus away from the efficient benchmark. When segregation 
is moderate, the same tilt offsets majority dominance and improves learning 
relative to the no-AI benchmark. When segregation is high, cross-group interaction 
becomes too weak to discipline the aggregator, so minority-weighted training 
again worsens learning. Thus, the effect of minority reweighting is non-monotone: it is beneficial when it counteracts majority bias, but detrimental when it either over-corrects or when limited cross-group interaction prevents information from being effectively aggregated.

\section{Social Learning with Local Aggregators}

\label{sec:local_aggregators} The analysis so far has focused on a single
global aggregator that is trained on population-wide beliefs and feeds a
unified signal back to all agents. This architecture captures large-scale
systems, such as current large language models, that pool information
broadly. In many environments, however, intermediated information
aggregation can also be more localized and topic-specific. This can be
because of pre-AI intermediaries such as newspapers, professional bodies and
local associations, or because of domain-specific AI models that primarily
train on information from local communities and are thus designed to be
informative about particular issues relevant to these communities (even
though their outputs may diffuse beyond those communities). This section
studies how learning changes when aggregators are local rather than global. 
\begin{figure}[t]
\centering
\includegraphics[width=100mm]{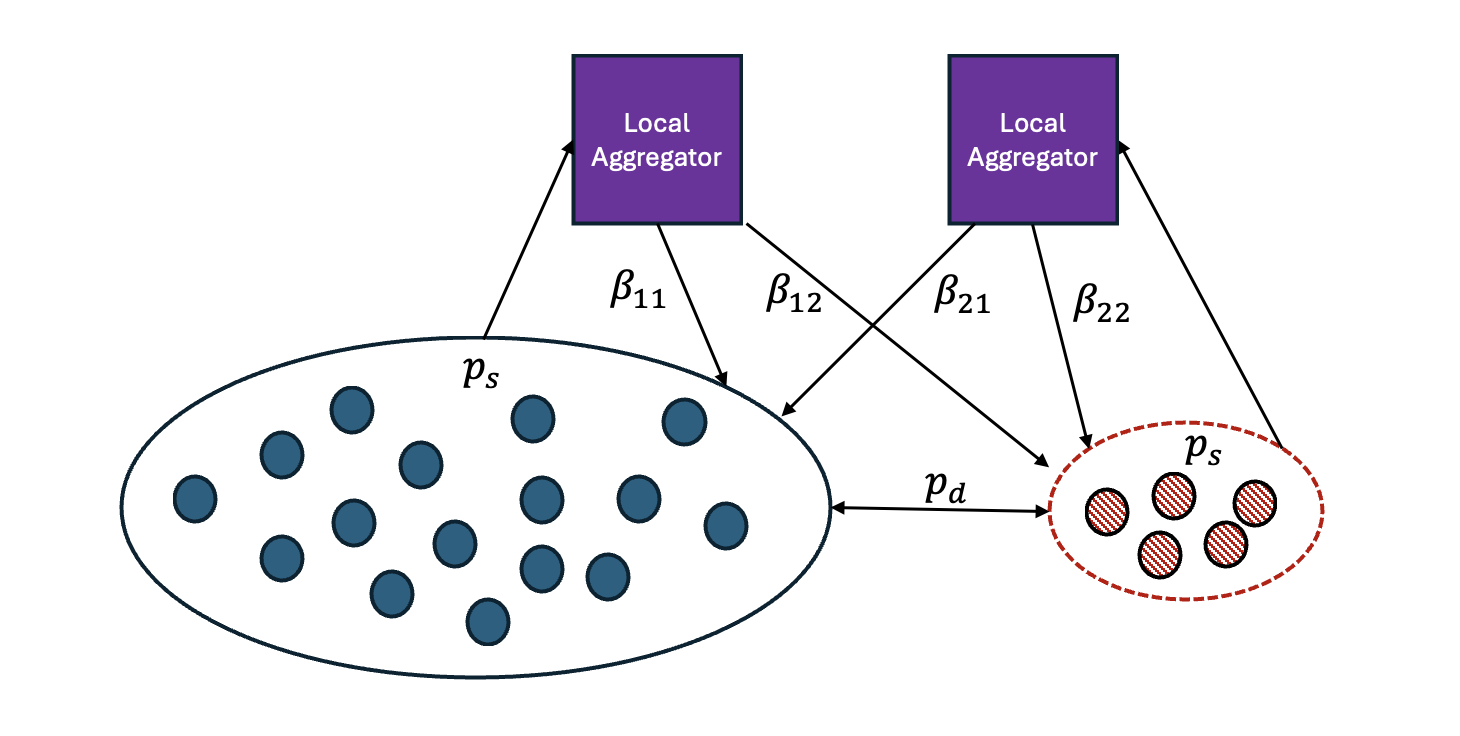}
\caption{Local aggregator architecture.}
\label{fig:local-agg}
\end{figure}

\subsection{Model with Local Aggregators}

\paragraph{Extended environment.}

We extend the baseline environment to a multidimensional state $\theta =
(\theta_1,\theta_2)^\top \in \mathbb{R}^2$, where $\theta_k$ represents the
state of topic $k$. As before, agents are partitioned into two islands $j
\in \{1,2\}$ with relative size $\pi = n_1/n_2 > 1$ and homophily parameter $%
h>1$ governing within- versus cross-island interaction. Let $F$ denote the $%
2 \times 2$ matrix from Section~\ref{sec:speed_results}.

Information is \emph{local} (or topic-specific): island $j$ is the
population that is directly informative about topic $\theta _{j}$. Indeed,
each agent $i$ on island $j$ receives an unbiased private signal about $%
\theta _{j}$: $s_{i,j}=\theta _{j}+\varepsilon _{i,j}$ where $\{\varepsilon
_{i,j}\}_{i=1}^{n}$ are independent, zero mean noise terms with finite
variance, and receives no direct information about the other topic $\theta
_{j^{\prime }}$ with $j^{\prime }\neq j$. The assumption is not that only
one island cares about a topic, but that first-hand signals and specialized
expertise are concentrated locally (e.g., local health systems vs. local
industries/labor markets), making initial information topic-specific.

We normalize initial beliefs so that agents place zero belief on topics
about which they are uninformed, i.e., $p_{i,j^{\prime }}(0)=0$ for $%
j^{\prime }\neq j$. Let $p_{k}(t)\in \mathbb{R}^{2}$ denote the vector of
island-level beliefs about topic $k$ at time $t$, with the no-aggregator
dynamics in the following form of 
\begin{equation*}
p_{k}(t+1)=Fp_{k}(t),\quad \text{\textnormal{for all }}k\in \{1,2\}.
\end{equation*}%
The efficient benchmark aggregates the informative signals topic by topic.
Under the same diffuse prior and equal-variance signal structure as before,
the benchmark is 
\begin{equation*}
\hat{\theta}=(\hat{\theta}_{1},\hat{\theta}_{2})=\left( \tfrac{1}{n_{1}}%
\sum_{i\in \text{Island} \, 1}s_{i,1},\tfrac{1}{n_{2}}\sum_{i\in \text{Island%
}\, 2}s_{i,2}\right) .
\end{equation*}%
There are two local aggregators, indexed by $k\in \{1,2\}$, where local
aggregator $k$ is specialized to topic $\theta _{k}$. Each local aggregator
trains only on beliefs about its topic (see Figure~\ref{fig:local-agg}).
Formally, let $A_{1}=(1\ \ 0)$ and $A_{2}=(0\ \ 1)$ so that $A_{k}p_{k}(t)$
extracts beliefs about topic $k$. Each local aggregator produces an
observable output $a_{k}(t)\in \mathbb{R}$ that updates according to 
\begin{equation*}
a_{k}(t+1)=\rho a_{k}(t)+(1-\rho )A_{k}p_{k}(t),\quad \text{\textnormal{for
all }}k\in \{1,2\},
\end{equation*}%
where $\rho \in (0,1)$ governs the speed of updating. Here, the lower $\rho $
corresponds to faster updating and stronger feedback. Local aggregators
influence agents asymmetrically across islands. Let 
\begin{equation*}
B_k = 
\begin{pmatrix}
\beta_{k1} \\ 
\beta_{k2}%
\end{pmatrix}
\in \mathbb{R}^2, \textnormal{ for all } k \in\{1,2\},
\end{equation*}
so that $B_k$ collects island-by-island reliance on local aggregator $k$. In
particular, 
\begin{equation*}
B_{1}=%
\begin{pmatrix}
\beta _{11} \\ 
\beta _{12}%
\end{pmatrix}%
,\quad B_{2}=%
\begin{pmatrix}
\beta _{21} \\ 
\beta _{22}%
\end{pmatrix}.
\end{equation*}%
Here, $\beta_{kj} \in [0,1)$ denotes the weight placed by island $j$ on local
aggregator $k$. Equivalently, the first index $k$ labels the local
aggregator (topic), and the second index $j$ labels the island.

A key feature of local aggregators is that each of them is primarily trusted
by (and thus has stronger influence on) the population that is informative
about its topic. Because $B_k=(\beta_{k1},\beta_{k2})^\top$ collects
island-by-island reliance on local aggregator $k$, we impose the following
asymmetry: 
\begin{equation}
\beta _{11}>\beta _{12},\quad \beta _{22}>\beta _{21}.
\label{eq:local-dominance}
\end{equation}%
That is, island $1$ relies more on the topic $1$ aggregator than island $2$
does, and island $2$ relies more on the topic $2$ aggregator than island $1$
does. This assumption formalizes the idea that topic-relevant intermediaries
have greater influence within their own communities than across communities,
and rules out the degenerate case in which a local aggregator is relied upon
more heavily by the island that is uninformed about its topic. Given local
aggregator outputs, beliefs about each topic evolve as 
\begin{equation*}
p_{k}(t+1)=(\mathbf{I}_{2}-\mathnormal{Diag}%
\,(B_{k}))Fp_{k}(t)+B_{k}a_{k}(t),\quad \text{for all }k\in \{1,2\},
\end{equation*}
where $\mathnormal{Diag}\,(B_{k})$ is the diagonal matrix with entries given
by $B_{k}$.

Under the same regularity conditions as in Section~\ref{sec:general}, the
augmented system admits a unique consensus for each topic, yielding a
limiting belief vector. By abuse of notation, we define 
\begin{equation*}
p^{\star \star }:=(p_{1}^{\star \star },p_{2}^{\star \star }),
\end{equation*}%
where $p_{k}^{\star \star }$ denotes the consensus belief about topic $k$
under local aggregators.

\paragraph{Performance metric.}

We let the local-aggregation learning gap be the vector 
\begin{equation*}
\Delta _{2}:=(|p_{1}^{\star \star }-\hat{\theta}_{1}|,|p_{2}^{\star \star }-%
\hat{\theta}_{2}|).
\end{equation*}
We compare $\Delta_2$ to the no-aggregator benchmark vector $\Delta_0$ 
(formed by applying the no-aggregator dynamics to each topic) and to the 
global-aggregator learning gap vector $\Delta_1$ (formed by applying the 
global-aggregator dynamics to each topic). Each topic evolves under the
global-aggregator rule applied to $p_k(t)$, with a shared training design
across topics. Accordingly, $\Delta_1$ is computed topic-wise by running the
global-aggregator update on that topic's beliefs. The key question is
whether localization of training and influence improves learning and
mitigates the feedback-driven fragility we identified in the presence of a
global aggregator.

We next compare learning under local aggregators to the no-aggregator
benchmark and to learning under a single global aggregator. To avoid
confusion with Sections~\ref{sec:models}-\ref{sec:network_results}, note
that there $\Delta_0$ and $\Delta_1$ both denote the scalar gap to the
efficient benchmark $\hat{\theta}$ (which equals $\frac{\pi}{\pi+1}$ under
our two-island normalization), whereas $\Delta_0$, $\Delta_1$ and $%
\Delta_2$ here are all vectors of topicwise gaps to the topic truths under
the unit normalization $p_1(0)=(1,0)^\top$ and $p_2(0)=(0,1)^\top$ (hence
the efficient benchmark is given by $(1,1)$). Throughout, we hold fixed the underlying primitives
(i.e., signals, network structure, and agents' updating rules) so that
differences in outcomes arise solely from the architecture of aggregators.
This allows us to isolate the economic forces introduced by scale and
centralization, abstracting from differences in data quality or behavioral
assumptions.

\subsection{Local Aggregators versus the No-Aggregator Benchmark}

We first compare local aggregators to decentralized learning without any
aggregators.

\begin{proposition}
\label{prop:local_effects} Learning is better across all topics under local
aggregators than without any aggregators. That is, $(\Delta_2)_k < (\Delta_0)_k$ 
for each topic $k \in \{1,2\}$.
\end{proposition}

Proposition~\ref{prop:local_effects} demonstrates that local aggregators improve 
learning relative to the no-aggregator benchmark. The reason is that each aggregator 
is topic-specific: aggregator $k$ is trained only on beliefs about $\theta_k$ from 
the subgroup that is informative about that topic, so its input is more relevant 
and less noisy. Its influence is also disciplined, since each local aggregator is 
relied on more heavily by the island that is informative about its topic and 
less heavily by the other island. This allows topic-relevant information to 
spill across groups without generating the system-wide feedback distortions 
of a global aggregator. Unlike the global case in Theorem~\ref{thm:rho}, 
where training reflects an endogenously distorted population-wide mixture 
of beliefs, local aggregators keep feedback in separate channels anchored 
to the informative subgroup, making learning more robust.

Proposition~\ref{prop:local_effects} and Theorem~\ref{thm:rho} emphasize
that the key design issue is not whether aggregator outputs cross groups ---
they do so under both architectures --- but whether training data are
globally pooled and endogenously contaminated or locally anchored to
informative sources. A global aggregator magnifies feedback and this makes
learning fragile, especially under uncertainty or fast updating, while local
aggregators preserve informational discipline by tying each training process
to the agents who observe the relevant state.

\subsection{Limits of A Single Global Aggregator}

We proceed to compare learning under local aggregators and that under a
single global aggregator in a multidimensional setting. By a single global
aggregator, we do \emph{not} mean a scalar intermediary that pools beliefs
across topics and broadcasts one common numerical output. Rather, the model
is \emph{parallel by topic}: for each topic $k$, the aggregator produces a
topic-specific signal/output and the within-topic belief-updating dynamics
are run on that topic's state. The sense in which the aggregator is \emph{%
single} is that it is the same global architecture applied across topics
(e.g., one common set of training weights $\alpha$ and the same adoption
structure, when imposed) so that the induced map is identical across topics
up to the topic's inputs. Consequently, objects such as $\Delta_1$ are
defined and analyzed topicwise by applying the global-aggregator dynamics
separately to each topic, and then comparing the resulting learning gaps
across specifications.

\begin{theorem}
\label{thm:impossibility} Suppose a single global aggregator replaces the
local aggregators. Then there exists at least one topic $k^{\star }\in
\{1,2\}$ for which learning is worse under a global aggregator than under
local aggregators. That is, $(\Delta _{1})_{k^{\star }}>(\Delta
_{2})_{k^{\star }}$.
\end{theorem}

Theorem~\ref{thm:impossibility} formalizes a basic limitation of global 
aggregation in multi-topic environments. Local aggregators are specialized: 
each topic is assigned an aggregator trained on beliefs from the subgroup 
that is informative about that topic, so training remains aligned with the 
relevant source of information even if outputs spill across islands. A single 
global aggregator, by contrast, applies one common training-and-feedback 
design across all topics. This shared design cannot simultaneously match 
different islands’ informational advantages: performing well on topic 1 
requires placing weight on island 1, while performing well on topic 2 requires 
placing weight on island 2. These objectives conflict, so any global design 
that improves learning on one topic necessarily weakens it on another. 
Local aggregators avoid this problem by keeping training channels separate 
and topic-specific.

Theorem~\ref{thm:impossibility} therefore complements the earlier results in
two ways. First, it strengthens the message of Theorem~\ref{thm:rho}:
fragility is not only about updating speed or uncertainty over network
structure, but also about the scope of AI-based aggregation. Second, it clarifies why Proposition~\ref{prop:local_effects} holds: local aggregators improve learning by preserving specialization and anchoring topic-specific aggregation to agents who are most informed about that topic. In short, global AI-based aggregation of information fails
typically both because of feedback-driven amplification and because of
intrinsic multi-topic coupling, whereas localized aggregation avoids both
forces by construction.

\section{Conclusion}

\label{sec:conclusion}

This paper studies how AI aggregation influences social learning. We extend
the DeGroot model of belief dynamics by introducing an AI aggregator as an
endogenous intermediary that both trains on and influences population
beliefs. The DeGroot model provides a tractable framework in which this
training can be formalized --- as training weights attached to the beliefs
of different agents. Our analysis highlights how the network structure (in
particular, the degree of segregation and homophily) interacts with the
training weights and the speed of updating of the global AI aggregator to
shape belief dynamics.

Our first set of results presents an important robustness tradeoff. When a
single global aggregator updates rapidly, feedback between its outputs and
its training data undermines robustness: small misspecifications in training
weights or uncertainty about the social network are amplified rather than
corrected. Beyond a threshold, no training design can robustly improve
learning across plausible environments. This provides a formal account of
feedback-driven failure, often described as model collapse, arising from
endogenous redundancy rather than data scarcity.

We explore the interaction between aggregators and group structure 
in greater detail: majority-weighted training interacts monotonically with
segregation to worsen learning, as network reinforcement and data imbalance
align. Minority-weighted training can initially improve learning by
counteracting majority dominance, but its effects are non-monotone:
increased segregation eventually weakens cross-group discipline and leads to
overcorrection. Bias correction through centralized aggregation of
information therefore depends critically on social structure and feedback.

Finally, we compare global and local aggregators in a multidimensional
setting. Local, topic-specific aggregators anchor training to populations
that are informative about each dimension, compartmentalizing feedback and
preserving informational diversity. This architecture avoids the system-wide
coupling that drives fragility under the global aggregator. Moreover, no
single global aggregator can replicate the performance of specialized local
aggregators across all dimensions, revealing a fundamental limitation of
centralized design.

In summary, our results emphasize that a central design choice in AI is not 
whether information is aggregated, but how broad the information sources 
are for AI models, how quickly these updates take place, and how those updates 
are then fed back into the population. Scale and speed can be beneficial 
only insofar as feedback remains disciplined. Modular, localized architectures 
sacrifice breadth and scale, but preserve valuable specialization, yielding 
more reliable improvements in learning.

There are many interesting areas for future research. First, the framework
here can be extended so that there are multiple global aggregators with
different training weights. Second, a more ambitious generalization would be
to endogenize the reliance of different agents on different global and local
AI aggregation (e.g., by making them more Bayesian in the weights they place on the various aggregators). Third, one could consider hybrid global-local architectures.
Fourth, the overall network structure can be endogenized more generally,
though this is typically challenging in the DeGroot setup. Finally, it would
be interesting to experimentally investigate whether changing the training
weights of AI aggregation along the lines of our analysis will modify the
extent of effects in practice. 

\newpage {
\bibliographystyle{te}
\bibliography{llm.bib}
} \newpage \appendix

\section{Proofs}

\label{app:proofs} 
We present all omitted proofs from the main body.

\subsection{Proofs from Section~\protect\ref{sec:general}}

\noindent\emph{Proof of Proposition~\protect\ref{prop:convergence}.} We show that $\Gamma$ is strongly connected. First, consider any two agents $i$ 
and $j$. Because $T$ is strongly connected and $\beta_i<1$ for all $i$, agent $i$ is 
reached from agent $j$ and agent $j$ is reached from agent $i$ in the augmented 
graph $\Gamma$.

Next, consider the aggregator and an arbitrary agent $j$. Because $\sum_{i=1}^n \alpha_i=1$ 
and $\alpha_i \geq 0$ for all $i$, there exists some agent $i^\ast$ such that 
$\alpha_{i^\star}>0$. Hence the aggregator is reached from agent $i^\star$. 
Because $T$ is strongly connected, agent $i^\ast$ is reached from agent $j$. 
Therefore, the aggregator is reached from agent $j$.

Conversely, because $\sum_{i=1}^n \beta_i>0$, there exists some agent $i^\star$ 
such that $\beta_{i^\star}>0$. Hence agent $i^\star$ is reached from the aggregator. 
Because $T$ is strongly connected, agent $j$ is reached from agent $i^\star$. 
Therefore, agent $j$ is reached from the aggregator. Putting these pieces together 
yields that $\Gamma$ is strongly connected.

We next show that $\Gamma$ is aperiodic. Because $\rho\in(0,1)$, the aggregator has 
a self-loop. In addition, the subgraph induced by agents is aperiodic because 
$T$ is aperiodic and $\beta_i<1$ for each $i$. Putting these pieces together 
yields the desired result. $\qed$

\begin{proposition}
\label{prop:Schweitzer}  Let $T \in \br^{n\times n}$ be a regular Markov transition 
matrix with a unique stationary distribution $s \in \br^{1\times n}$. Let $T^\infty$ 
denote the rank-one matrix with $s$ in every row, and define the fundamental matrix 
$Y \equiv \sum_{k=0}^\infty (T^k-T^\infty)$.  Let $D \in \br^{n \times n}$ be such that 
$\hat{T}=T+D$ is also regular, and let $\hat{s} \in \br^{1\times n}$ denote the unique 
stationary distribution of $\hat{T}$. If $\mathbf{I}_n-DY$ is nonsingular, then 
$\hat{s}-s = sDY(\mathbf{I}_n-DY)^{-1}$. Equivalently, $\hat{s}=s(\mathbf{I}_n-DY)^{-1}$. 
\end{proposition}

\noindent\emph{Proof}. This follows immediately from %
\citet{Schweitzer-1968-Perturbation}. $\qed$

\bigskip\noindent\emph{Proof of Theorem~\protect\ref{thm:general}.} Because $T$ is strongly connected and aperiodic, there is a rank-1 matrix $%
T^{\infty }$ corresponding to the unique left-eigenvector $s$ of eigenvalue
one in every row. In this context, the fundamental matrix of $T$ is defined
by $Y\equiv \sum_{k=0}^{\infty }(T^{k}-T^{\infty })$. We claim the following
form of the consensus as a function of $\rho $, $\alpha $, $\beta $, $s$,
and the fundamental matrix $Y$ of network $T$. We define $D\in \mathbb{R}%
^{n\times n}$, $\hat{w}\in \mathbb{R}$ and $\hat{v}$ (a $1\times n$ vector)
as follows, 
\begin{equation*}
D=\beta\alpha-\mathnormal{Diag}(\beta )T,
\end{equation*}%
and 
\begin{equation*}
\hat{v}=s(\mathbf{I}_{n}-DY)^{-1},\qquad \hat{w}=\tfrac{1}{1-\rho }\hat{v}%
\beta .
\end{equation*}%
Then, the consensus is given by 
\begin{equation*}
\tfrac{1}{1+\hat{w}}(\hat{w}\alpha +\hat{v}T)p(0).
\end{equation*}%
To see this, we have $(\hat{w},\hat{v})\Gamma =(\hat{w},\hat{v})$ if and
only if 
\begin{equation*}
\hat{w}=\rho\hat{w}+\hat{v}\beta ,\quad (1-\rho)\hat{w}\alpha +\hat{v}(%
\mathbf{I}_{n}-\mathnormal{Diag}\,(\beta ))T=\hat{v}.
\end{equation*}
This implies that $(\hat{w},\hat{v})\Gamma =(\hat{w},\hat{v})$ if and only
if $\hat{w}=\frac{1}{1-\rho }\hat{v}\beta$ and $\hat{v}(T+D)=\hat{v} $.
Because $D$ is a perturbation matrix such that $T+D$ is regular, 
Proposition~\ref{prop:Schweitzer} implies $\hat{v}-s = sDY(\mathbf{I}_n-DY)^{-1}$.
Hence, $\hat{v}=s+sDY(\mathbf{I}_{n}-DY)^{-1}=s(\mathbf{I}_{n}-DY)^{-1}$. 
Because $a(1)=\alpha p(0)$ and $p(1)=Tp(0)$, we have 
\begin{equation*}
p^{\star \star} = \tfrac{1}{1+\hat{w}}(\hat{w}a(1)+\hat{v}p(1)) = \tfrac{1}{%
1+\hat{w}}(\hat{w}\alpha +\hat{v}T)p(0).
\end{equation*}%
Finally, we define $z=(1-\rho )\alpha (\mathbf{I}_{n}-(\mathbf{I}_{n}-%
\mathnormal{Diag}(\beta ))T)^{-1}$. Then, we show the consensus is given
by 
\begin{equation*}
p^{\star \star} = \tfrac{1}{1+z\mathbf{1}_{n}}(\alpha +zT)p(0).
\end{equation*}%
As a consequence of our previous argument, the consensus is given by 
\begin{equation}
\tfrac{1}{1+\hat{w}}(\hat{w}\alpha +\hat{v}T)p(0),  \label{eq:main}
\end{equation}%
where $\hat{v}=s(\mathbf{I}_{n}-DY)^{-1}$ and $\hat{w}=\frac{1}{1-\rho }\hat{%
v}\beta $. Note that $D=\beta \alpha -\mathnormal{Diag}(\beta )T$ and $Y=(%
\mathbf{I}_{n}-T+\mathbf{1}_{n}s)^{-1}-\mathbf{1}_{n}s$.  Because $D\mathbf{1}_{n}=0$, 
we have $DY=(\beta \alpha -\mathnormal{Diag}(\beta )T)(\mathbf{I}_{n}-T+\mathbf{1}_{n}s)^{-1}$. 
By applying the Woodbury identity and using the fact that $sT=s$, we have 
\begin{eqnarray*}
\hat{v} &=&s(\mathbf{I}_{n}-(\mathnormal{Diag}\,(\beta )T-\beta \alpha )(%
\mathbf{I}_{n}-T+\mathbf{1}_{n}s+\mathnormal{Diag}\,(\beta )T-\beta \alpha
)^{-1}) \\
&=&s(\mathbf{I}_{n}-T+\mathbf{1}_{n}s)(\mathbf{I}_{n}-T+\mathbf{1}_{n}s+%
\mathnormal{Diag}(\beta )T-\beta \alpha )^{-1} \\
&=&s(\mathbf{I}_{n}-(\mathbf{I}_{n}-\mathnormal{Diag}(\beta ))T+\mathbf{1}%
_{n}s-\beta \alpha )^{-1}.
\end{eqnarray*}%
For simplicity, we define $\Omega =\mathbf{I}_{n}-(\mathbf{I}_{n}-%
\mathnormal{Diag}(\beta ))T$. This matrix is invertible because $\|(\mathbf{I}_{n}-%
\mathnormal{Diag}(\beta ))T\|_\infty = \max_i (1-\beta _{i}) < 1$. 
Then, we can rewrite $\hat{v}$ in the following form of 
\begin{equation*}
\hat{v}=s\left( \Omega +%
\begin{pmatrix}
\mathbf{1}_{n} & -\beta%
\end{pmatrix}%
\begin{pmatrix}
s \\ 
\alpha%
\end{pmatrix}%
\right) ^{-1}.
\end{equation*}%
Applying the Woodbury identity again yields 
\begin{eqnarray*}
\hat{v} &=&s\left( \Omega ^{-1}-\Omega ^{-1}%
\begin{pmatrix}
\mathbf{1}_{n} & -\beta%
\end{pmatrix}%
\left( \mathbf{I}_{2}+%
\begin{pmatrix}
s \\ 
\alpha%
\end{pmatrix}%
\Omega ^{-1}%
\begin{pmatrix}
\mathbf{1}_{n} & -\beta%
\end{pmatrix}%
\right) ^{-1}%
\begin{pmatrix}
s \\ 
\alpha%
\end{pmatrix}%
\Omega ^{-1}\right) \\
&=&s\Omega ^{-1}-%
\begin{pmatrix}
s\Omega ^{-1}\mathbf{1}_{n} & -s\Omega ^{-1}\beta%
\end{pmatrix}%
\begin{pmatrix}
1+s\Omega ^{-1}\mathbf{1}_{n} & -s\Omega ^{-1}\beta \\ 
\alpha \Omega ^{-1}\mathbf{1}_{n} & 1-\alpha \Omega ^{-1}\beta%
\end{pmatrix}%
^{-1}%
\begin{pmatrix}
s\Omega ^{-1} \\ 
\alpha \Omega ^{-1}%
\end{pmatrix}%
.
\end{eqnarray*}%
Using the definition of $\Omega $, we have $\Omega ^{-1}\beta =\mathbf{1}%
_{n} $. Plugging this result into the above equality and using the fact that 
$\alpha \mathbf{1}_{n}=s\mathbf{1}_{n}=1$ yields 
\begin{eqnarray*}
\hat{v} &=&s\Omega ^{-1}-%
\begin{pmatrix}
s\Omega ^{-1}\mathbf{1}_{n} & -1%
\end{pmatrix}%
\begin{pmatrix}
1+s\Omega ^{-1}\mathbf{1}_{n} & -1 \\ 
\alpha \Omega ^{-1}\mathbf{1}_{n} & 0%
\end{pmatrix}%
^{-1}%
\begin{pmatrix}
s\Omega ^{-1} \\ 
\alpha \Omega ^{-1}%
\end{pmatrix}
\\
&=&s\Omega ^{-1}-\tfrac{1}{\alpha \Omega ^{-1}\mathbf{1}_{n}}%
\begin{pmatrix}
s\Omega ^{-1}\mathbf{1}_{n} & -1%
\end{pmatrix}%
\begin{pmatrix}
0 & 1 \\ 
-\alpha \Omega ^{-1}\mathbf{1}_{n} & 1+s\Omega ^{-1}\mathbf{1}_{n}%
\end{pmatrix}%
\begin{pmatrix}
s\Omega ^{-1} \\ 
\alpha \Omega ^{-1}%
\end{pmatrix}
\\
&=&s\Omega ^{-1}-\tfrac{1}{\alpha \Omega ^{-1}\mathbf{1}_{n}}%
\begin{pmatrix}
\alpha \Omega ^{-1}\mathbf{1}_{n} & -1%
\end{pmatrix}%
\begin{pmatrix}
s\Omega ^{-1} \\ 
\alpha \Omega ^{-1}%
\end{pmatrix}
\\
&=&\tfrac{\alpha \Omega ^{-1}}{\alpha \Omega ^{-1}\mathbf{1}_{n}}.
\end{eqnarray*}%
Thus, we have $\hat{w}=\frac{1}{1-\rho }\hat{v}\beta =\frac{1}{(1-\rho
)\alpha \Omega ^{-1}\mathbf{1}_{n}}$. Plugging $(\hat{w},\hat{v})$ into Eq.~%
\eqref{eq:main} yields the consensus $p^{\star \star }$. $\qed$

\subsection{Closed-Form Learning Gaps (Corollary to Theorem~\protect\ref%
{thm:general})}

Using Theorem~\ref{thm:general}, we provide closed-form expressions for the
learning gaps under a global AI aggregator and two local aggregators. For a
global AI aggregator, we have scalar learning gaps $\Delta _{1}$ (with AI
aggregator) and $\Delta _{0}$ (without an aggregator). For two local
aggregators, we have the two-dimensional learning gaps $\Delta _{0}$ (no
aggregator), $\Delta _{1}$ (global aggregator architecture), and $\Delta
_{2} $ (local aggregator architecture).

\bigskip \noindent \textbf{The learning gap with a global aggregator}.
Suppose that $h=p_{s}/p_{d}\in (1,\infty )$ and $\pi =n_{1}/n_{2}\in
(1,\infty )$. Then, we can rewrite $\alpha ,\beta ,F$ as follows, 
\begin{equation*}
\alpha =%
\begin{pmatrix}
\alpha & 1-\alpha%
\end{pmatrix}%
,\quad \beta =%
\begin{pmatrix}
\beta _{1} \\ 
\beta _{2}%
\end{pmatrix}%
,\quad F=%
\begin{pmatrix}
\tfrac{h\pi }{h\pi +1} & \tfrac{1}{h\pi +1} \\ 
\tfrac{\pi }{h+\pi } & \tfrac{h}{h+\pi }%
\end{pmatrix}%
,\quad p(0)=%
\begin{pmatrix}
1 \\ 
0%
\end{pmatrix}%
.
\end{equation*}%
and derive a closed-form characterization of the consensus $p^{\star \star }$
using Theorem~\ref{thm:general} as follows, 
\begin{equation}
p^{\star \star }=\tfrac{1}{1+z\mathbf{1}_{2}}\left( \alpha +z%
\begin{pmatrix}
\tfrac{h\pi }{h\pi +1} \\ 
\tfrac{\pi }{h+\pi }%
\end{pmatrix}%
\right) ,  \label{ineq:global-1}
\end{equation}%
where 
\begin{equation*}
z=(1-\rho )(\alpha \ \ \ 1-\alpha )(\mathbf{I}_{2}-(\mathbf{I}_{2}-%
\mathnormal{Diag}\,(\beta ))F)^{-1}.
\end{equation*}%
First, we claim that 
\begin{equation*}
(\mathbf{I}_{2}-(\mathbf{I}_{2}-\mathnormal{Diag}(\beta ))F)^{-1}=%
\begin{pmatrix}
1 & \tfrac{1-\beta _{1}}{\beta _{1}h\pi +1} \\ 
& 1%
\end{pmatrix}%
\begin{pmatrix}
\tfrac{h\pi +1}{\beta _{1}h\pi +1} &  \\ 
& \tfrac{(h+\pi )(\beta _{1}h\pi +1)}{(\beta _{2}h+\pi )(\beta _{1}h\pi
+1)-(1-\beta _{1})(1-\beta _{2})\pi }%
\end{pmatrix}%
\begin{pmatrix}
1 &  \\ 
\tfrac{(1-\beta _{2})\pi (h\pi +1)}{(h+\pi )(\beta _{1}h\pi +1)} & 1%
\end{pmatrix}%
.
\end{equation*}%
Indeed, we have 
\begin{equation*}
\mathbf{I}_{2}-(\mathbf{I}_{2}-\mathnormal{Diag}(\beta ))F=%
\begin{pmatrix}
\tfrac{\beta _{1}h\pi +1}{h\pi +1} & -\tfrac{1-\beta _{1}}{h\pi +1} \\ 
-\tfrac{(1-\beta _{2})\pi }{h+\pi } & \tfrac{\beta _{2}h+\pi }{h+\pi }%
\end{pmatrix}%
\doteq 
\begin{pmatrix}
a & b \\ 
c & d%
\end{pmatrix}%
,
\end{equation*}%
and obtain the desired result using the one-dimensional version of Schur
complement as follows, 
\begin{equation*}
\begin{pmatrix}
a & b \\ 
c & d%
\end{pmatrix}%
^{-1}=%
\begin{pmatrix}
1 & -\tfrac{b}{a} \\ 
& 1%
\end{pmatrix}%
\begin{pmatrix}
\tfrac{1}{a} &  \\ 
& \tfrac{a}{ad-bc}%
\end{pmatrix}%
\begin{pmatrix}
1 &  \\ 
-\tfrac{c}{a} & 1%
\end{pmatrix}%
.
\end{equation*}%
Then, we have 
\begin{equation*}
\begin{array}{rcl}
z & = & (1-\rho )(\alpha \ \ \ 1-\alpha )(\mathbf{I}_{2}-(\mathbf{I}_{2}-%
\mathnormal{Diag}(\beta ))F)^{-1} \\ 
& = & (1-\rho )(\alpha \ \ \ 1-\alpha )%
\begin{pmatrix}
1 & \tfrac{1-\beta _{1}}{\beta _{1}h\pi +1} \\ 
& 1%
\end{pmatrix}%
\begin{pmatrix}
\tfrac{h\pi +1}{\beta _{1}h\pi +1} &  \\ 
& \tfrac{(h+\pi )(\beta _{1}h\pi +1)}{(\beta _{2}h+\pi )(\beta _{1}h\pi
+1)-(1-\beta _{1})(1-\beta _{2})\pi }%
\end{pmatrix}%
\begin{pmatrix}
1 &  \\ 
\tfrac{(1-\beta _{2})\pi (h\pi +1)}{(h+\pi )(\beta _{1}h\pi +1)} & 1%
\end{pmatrix}
\\ 
& = & (1-\rho )%
\begin{pmatrix}
\alpha & \tfrac{(1-\alpha )\beta _{1}h\pi +(1-\alpha \beta _{1})}{\beta
_{1}h\pi +1}%
\end{pmatrix}%
\begin{pmatrix}
\tfrac{h\pi +1}{\beta _{1}h\pi +1} &  \\ 
& \tfrac{(h+\pi )(\beta _{1}h\pi +1)}{(\beta _{2}h+\pi )(\beta _{1}h\pi
+1)-(1-\beta _{1})(1-\beta _{2})\pi }%
\end{pmatrix}%
\begin{pmatrix}
1 &  \\ 
\tfrac{(1-\beta _{2})\pi (h\pi +1)}{(h+\pi )(\beta _{1}h\pi +1)} & 1%
\end{pmatrix}
\\ 
& = & (1-\rho )%
\begin{pmatrix}
\tfrac{\alpha (h\pi +1)}{\beta _{1}h\pi +1} & \tfrac{(h+\pi )((1-\alpha
)\beta _{1}h\pi +(1-\alpha \beta _{1}))}{(\beta _{2}h+\pi )(\beta _{1}h\pi
+1)-(1-\beta _{1})(1-\beta _{2})\pi }%
\end{pmatrix}%
\begin{pmatrix}
1 &  \\ 
\tfrac{(1-\beta _{2})\pi (h\pi +1)}{(h+\pi )(\beta _{1}h\pi +1)} & 1%
\end{pmatrix}
\\ 
& = & (1-\rho )%
\begin{pmatrix}
\tfrac{(h\pi +1)(\alpha \beta _{2}h+(1-\beta _{2}+\alpha \beta _{2})\pi )}{%
(\beta _{2}h+\pi )(\beta _{1}h\pi +1)-(1-\beta _{1})(1-\beta _{2})\pi } & 
\tfrac{(h+\pi )((1-\alpha )\beta _{1}h\pi +(1-\alpha \beta _{1}))}{(\beta
_{2}h+\pi )(\beta _{1}h\pi +1)-(1-\beta _{1})(1-\beta _{2})\pi }%
\end{pmatrix}%
,%
\end{array}%
\end{equation*}%
which implies 
\begin{equation}
z\mathbf{1}_{2}=\tfrac{(1-\rho )(((1-\alpha )\beta _{1}+\alpha \beta
_{2})h^{2}\pi +(1+(1-\alpha )\beta _{1}-(1-\alpha )\beta _{2})h\pi
^{2}+(1-\alpha \beta _{1}+\alpha \beta _{2})h+(2-\alpha \beta _{1}-(1-\alpha
)\beta _{2})\pi )}{\beta _{1}\beta _{2}h^{2}\pi +\beta _{1}h\pi ^{2}+\beta
_{2}h+(\beta _{1}+\beta _{2}-\beta _{1}\beta _{2})\pi },
\label{ineq:global-2}
\end{equation}%
and 
\begin{equation}
z%
\begin{pmatrix}
\tfrac{h\pi }{h\pi +1} \\ 
\tfrac{\pi }{h+\pi }%
\end{pmatrix}%
=\tfrac{(1-\rho )(\alpha \beta _{2}h^{2}\pi +(1+\beta _{1}-\beta _{2}-\alpha
\beta _{1}+\alpha \beta _{2})h\pi ^{2}+(1-\alpha \beta _{1})\pi )}{\beta
_{1}\beta _{2}h^{2}\pi +\beta _{1}h\pi ^{2}+\beta _{2}h+(\beta _{1}+\beta
_{2}-\beta _{1}\beta _{2})\pi }.  \label{ineq:global-3}
\end{equation}%
Plugging Eq.~\eqref{ineq:global-2} and Eq.~\eqref{ineq:global-3} into Eq.~%
\eqref{ineq:global-1} yields 
\begin{equation*}
{\small 
\begin{array}{l}
p^{\star \star }\equiv \\ 
\tfrac{(\alpha \beta _{1}\beta _{2}+(1-\rho )\alpha \beta _{2})h^{2}\pi
+(\alpha \beta _{1}+(1-\rho )(1+(1-\alpha )(\beta _{1}-\beta _{2})))h\pi
^{2}+\alpha \beta _{2}h+(\alpha (\beta _{1}+\beta _{2}-\beta _{1}\beta
_{2})+(1-\rho )(1-\alpha \beta _{1}))\pi }{(\beta _{1}\beta _{2}+(1-\rho
)(\beta _{1}-\alpha (\beta _{1}-\beta _{2}))h^{2}\pi +(\beta _{1}+(1-\rho
)(1+(1-\alpha )(\beta _{1}-\beta _{2})))h\pi ^{2}+(\beta _{2}+(1-\rho
)(1-\alpha (\beta _{1}-\beta _{2})))h+(\beta _{1}+\beta _{2}-\beta _{1}\beta
_{2}+(1-\rho )(2-\beta _{2}-\alpha (\beta _{1}-\beta _{2})))\pi },%
\end{array}%
}
\end{equation*}%
which implies 
\begin{equation*}
{\footnotesize 
\begin{array}{l}
\Delta _{1}(\rho ,\alpha ,\beta _{1},\beta _{2},h,\pi )\equiv \\ 
\left\vert \tfrac{(\alpha \beta _{1}\beta _{2}+(1-\rho )\alpha \beta
_{2})h^{2}\pi +(\alpha \beta _{1}+(1-\rho )(1+(1-\alpha )(\beta _{1}-\beta
_{2})))h\pi ^{2}+\alpha \beta _{2}h+(\alpha (\beta _{1}+\beta _{2}-\beta
_{1}\beta _{2})+(1-\rho )(1-\alpha \beta _{1}))\pi }{(\beta _{1}\beta
_{2}+(1-\rho )(\beta _{1}-\alpha (\beta _{1}-\beta _{2}))h^{2}\pi +(\beta
_{1}+(1-\rho )(1+(1-\alpha )(\beta _{1}-\beta _{2})))h\pi ^{2}+(\beta
_{2}+(1-\rho )(1-\alpha (\beta _{1}-\beta _{2})))h+(\beta _{1}+\beta
_{2}-\beta _{1}\beta _{2}+(1-\rho )(2-\beta _{2}-\alpha (\beta _{1}-\beta
_{2})))\pi }-\tfrac{\pi }{\pi +1}\right\vert .%
\end{array}%
}
\end{equation*}%
We also have by setting $\beta _{1}=\beta _{2}=\beta $, 
\begin{equation*}
\Delta _{1}(\rho ,\alpha ,\beta ,h,\pi )\equiv \left\vert \tfrac{\alpha
\beta (\beta +1-\rho )h^{2}\pi +(\alpha \beta +1-\rho )h\pi ^{2}+\alpha
\beta h+(\alpha \beta (2-\beta )+(1-\rho )(1-\alpha \beta ))\pi }{\beta
(\beta +1-\rho )h^{2}\pi +(\beta +1-\rho )h\pi ^{2}+(\beta +1-\rho
)h+(2-\beta )(\beta +1-\rho )\pi }-\tfrac{\pi }{\pi +1}\right\vert .
\end{equation*}

\bigskip \noindent \textbf{The learning gap with local aggregators}. Suppose
that $h=p_{s}/p_{d}\in (1,\infty )$ and $\pi =n_{1}/n_{2}\in (1,\infty )$.
Then, we can rewrite $A_{1},A_{2},B_{1},B_{2},F$ as follows, 
\begin{equation*}
A_{1}=%
\begin{pmatrix}
1 & 0%
\end{pmatrix}%
,\ A_{2}=%
\begin{pmatrix}
0 & 1%
\end{pmatrix}%
,\ B_{1}=%
\begin{pmatrix}
\beta _{11} \\ 
\beta _{12}%
\end{pmatrix}%
,\ B_{2}=%
\begin{pmatrix}
\beta _{21} \\ 
\beta _{22}%
\end{pmatrix}%
,\ F=%
\begin{pmatrix}
\tfrac{h\pi }{h\pi +1} & \tfrac{1}{h\pi +1} \\ 
\tfrac{\pi }{h+\pi } & \tfrac{h}{h+\pi }%
\end{pmatrix}%
,\ p_1(0)=%
\begin{pmatrix}
1 \\ 
0%
\end{pmatrix}%
, \ p_2(0)=%
\begin{pmatrix}
0 \\ 
1%
\end{pmatrix}%
.
\end{equation*}
Because information is topic-specific, the initial belief profiles differ
across topics. For topic $1$, island $1$ is the informed population, so we
normalize the initial belief vector as $p_1(0)=(1,0)^\top$ (island 1 starts
with a unit informational advantage and island 2 is uninformed). For topic $2
$, island $2$ is the informed population, so the analogous normalization is $%
p_2(0)=(0,1)^\top$. All subsequent expressions for $p_k^{\star}$ and $%
p_k^{\star\star}$ are linear in $p_k(0)$, and the learning-gap comparisons
depend only on the induced influence weights; thus, without loss of
generality, we work with these unit normalizations. Then, we derive a
closed-form characterization of $p_{1}^{\star \star }$ and $p_{2}^{\star
\star }$ using Theorem~\ref{thm:general} as follows, 
\begin{equation*}
p_{1}^{\star \star }=\tfrac{1}{1+z_{1}\mathbf{1}_{2}^{\top }}\left( 1+z_{1}%
\begin{pmatrix}
\tfrac{h\pi }{h\pi +1} \\ 
\tfrac{\pi }{h+\pi }%
\end{pmatrix}%
\right) ,\quad p_{2}^{\star \star }=\tfrac{1}{1+z_{2}\mathbf{1}_{2}^{\top }}%
\left( 0+z_{2}%
\begin{pmatrix}
\tfrac{1}{h\pi +1} \\ 
\tfrac{h}{h+\pi }%
\end{pmatrix}%
\right) ,
\end{equation*}%
where 
\begin{equation*}
z_{1}=(1-\rho )(1\ \ \ 0)(\mathbf{I}_{2}-(\mathbf{I}_{2}-\mathnormal{Diag}%
\,(B_{1}))F)^{-1},\quad z_{2}=(1-\rho )(0\ \ \ 1)(\mathbf{I}_{2}-(\mathbf{I}%
_{2}-\mathnormal{Diag}\,(B_{2}))F)^{-1}.
\end{equation*}%
By using the same arguments, we have 
\begin{equation*}
\begin{array}{rcl}
p_{1}^{\star \star }(\rho ,\beta _{11},\beta _{12},h,\pi ) & \equiv & \tfrac{%
(1-\rho +\beta _{11})\beta _{12}h^{2}\pi +(1-\rho +\beta _{11})h\pi
^{2}+\beta _{12}h+(1-\rho +\rho \beta _{11}+\beta _{12}-\beta _{11}\beta
_{12})\pi }{(1-\rho +\beta _{11})\beta _{12}h^{2}\pi +(1-\rho +\beta
_{11})h\pi ^{2}+(\beta _{12}+(1-\rho )(1-\beta _{11}+\beta
_{12}))h+(2(1-\rho )+\rho \beta _{11}+\beta _{12}-\beta _{11}\beta _{12})\pi 
}, \\ 
p_{2}^{\star \star }(\rho ,\beta _{21},\beta _{22},h,\pi ) & \equiv & \tfrac{%
(1-\rho+\beta_{22})\beta_{21}h^2\pi+\beta_{21}h\pi^2+(1-\rho+%
\beta_{22})h+((1-\rho)+\beta_{21}+\rho\beta_{22}-\beta_{21}\beta_{22})\pi}{%
(1-\rho +\beta _{22})\beta _{21}h^{2}\pi +(\beta _{21}+(1-\rho )(1+\beta
_{21}-\beta _{22}))h\pi ^{2}+(1-\rho +\beta _{22})h+(2(1-\rho )+\beta
_{21}+\rho \beta _{22}-\beta _{21}\beta _{22})\pi}.%
\end{array}%
\end{equation*}%
As a consequence, we have 
\begin{equation*}
\Delta _{2}(\rho ,\beta _{11},\beta _{12},\beta _{21},\beta _{22},h,\pi )
\equiv |(p_{1}^{\star \star }(\rho ,\beta _{11},\beta _{12},h,\pi
),p_{2}^{\star \star }(\rho ,\beta _{21},\beta _{22},h,\pi ))-(1,1)|.
\end{equation*}%
Similarly, the efficient benchmark without any aggregator is $\left(\frac{%
h\pi^2+\pi}{h\pi^2+h+2\pi},\frac{h+\pi}{h\pi^2+h+2\pi}\right)$. This leads
to 
\begin{equation*}
\Delta_0(h,\pi )\equiv \left\vert \left(\tfrac{h\pi^2+\pi}{h\pi^2+h+2\pi},%
\tfrac{h+\pi}{h\pi^2+h+2\pi}\right)-(1,1)\right\vert.
\end{equation*}
By abuse of notation, we have 
\begin{equation*}
\Delta_1(\rho,\alpha, \beta_1,\beta _2,h,\pi) \equiv \left\vert
(p_1^{\star\star}(\rho, \alpha, \beta_1, \beta_2, h, \pi),
p_2^{\star\star}(\rho, \alpha, \beta_1, \beta_2, h, \pi))-(1,1)\right\vert.
\end{equation*}
where $p_k^{\star\star}$ denotes the topic-$k$ consensus under the
global-aggregator dynamics. Because $p_1(0)=(1,0)^\top$ and $p_2(0)=(0,1)^\top$%
, we have $p_1^{\star\star}(\rho, \alpha, \beta_1, \beta_2, h, \pi) +
p_2^{\star\star}(\rho, \alpha, \beta_1, \beta_2, h, \pi)=1$. This leads to 
\begin{equation*}
\Delta_1(\rho,\alpha, \beta_1,\beta _2,h,\pi) \equiv \left\vert
(p_1^{\star\star}(\rho, \alpha, \beta_1, \beta_2, h, \pi),
1-p_1^{\star\star}(\rho, \alpha, \beta_1, \beta_2, h, \pi))-(1,1)\right\vert,
\end{equation*}
where $p_1^{\star\star}(\rho, \alpha, \beta_1, \beta_2, h, \pi) \in (0, 1)$
is the topic-$1$ consensus.

\subsection{Proofs from Section~\protect\ref{sec:speed_results}}

\noindent \emph{Proof of Theorem~\ref{thm:rho}.} We rewrite the learning gap
with a global aggregator $(\rho ,\alpha ,\beta _{1},\beta _{2})$ as 
\begin{equation*}
\Delta _{1}(\rho ,\alpha ,\beta _{1},\beta _{2},h,\pi )=\left\vert \tfrac{%
\bar{\phi}_{1}(\rho ,\alpha ,\beta _{1},\beta _{2},h,\pi )}{%
\underaccent{\bar}{\phi}_{1}(\rho ,\alpha ,\beta _{1},\beta _{2},h,\pi )}-%
\tfrac{\pi }{\pi +1}\right\vert ,
\end{equation*}%
where $\bar{\phi}_{1}$ and $\underaccent{\bar}{\phi}_{1}$ are defined by 
\begin{eqnarray*}
\bar{\phi}_{1}(\rho ,\alpha ,\beta _{1},\beta _{2},h,\pi ) &=&(\alpha \beta
_{1}\beta _{2}+(1-\rho )\alpha \beta _{2})h^{2}\pi +(\alpha \beta
_{1}+(1-\rho )(1+(1-\alpha )(\beta _{1}-\beta _{2})))h\pi ^{2} \\
&&+\alpha \beta _{2}h+(\alpha (\beta _{1}+\beta _{2}-\beta _{1}\beta
_{2})+(1-\rho )(1-\alpha \beta _{1}))\pi , \\
\underaccent{\bar}{\phi}_{1}(\rho ,\alpha ,\beta _{1},\beta _{2},h,\pi )
&=&(\beta _{1}\beta _{2}+(1-\rho )(\beta _{1}-\alpha (\beta _{1}-\beta
_{2}))h^{2}\pi +(\beta _{1}+(1-\rho )(1+(1-\alpha )(\beta _{1}-\beta
_{2})))h\pi ^{2} \\
&&+(\beta _{2}+(1-\rho )(1-\alpha (\beta _{1}-\beta _{2})))h+(\beta
_{1}+\beta _{2}-\beta _{1}\beta _{2}+(1-\rho )(2-\beta _{2}-\alpha (\beta
_{1}-\beta _{2})))\pi .
\end{eqnarray*}%
The learning gap without a global aggregator is 
\begin{equation*}
\Delta _{0}(h,\pi )=\left\vert \tfrac{h\pi ^{2}+\pi }{h\pi ^{2}+h+2\pi }-%
\tfrac{\pi }{\pi +1}\right\vert .
\end{equation*}%
By definition, we have 
\begin{equation*}
\Lambda _{\rho }=\left\{\alpha \in [0, 1] \mid \Delta _{1}(\rho ,\alpha ,\beta
_{1},\beta _{2},h,\pi )<\Delta _{0}(h,\pi ),\,\forall h\in \lbrack %
\underaccent{\bar}{h},\bar{h}],\forall \beta _{1},\beta _{2}\in (0,1)\right\}
\end{equation*}%
Fixing $\beta _{1},\beta _{2}\in (0,1)$ and $h\in \lbrack %
\underaccent{\bar}{h},\bar{h}]$, we have that $\Delta _{1}(\rho ,\alpha
,\beta _{1},\beta _{2},h,\pi )<\Delta _{0}(h,\pi )$ if and only if 
\begin{equation*}
\tfrac{2\pi }{\pi +1}-\tfrac{h\pi ^{2}+\pi }{h\pi ^{2}+h+2\pi }<\tfrac{\bar{%
\phi}_{1}(\rho ,\alpha ,\beta _{1},\beta _{2},h,\pi )}{\underaccent{\bar}{%
\phi}_{1}(\rho ,\alpha ,\beta _{1},\beta _{2},h,\pi )}<\tfrac{h\pi ^{2}+\pi 
}{h\pi ^{2}+h+2\pi }.
\end{equation*}%
Because $\underaccent{\bar}{\phi}_{1}(\rho ,\alpha ,\beta _{1},\beta
_{2},h,\pi )>0$, we have 
\begin{equation*}
\left( \tfrac{2\pi }{\pi +1}-\tfrac{h\pi ^{2}+\pi }{h\pi ^{2}+h+2\pi }%
\right) \underaccent{\bar}{\phi}_{1}(\rho ,\alpha ,\beta _{1},\beta
_{2},h,\pi )<\bar{\phi}_{1}(\rho ,\alpha ,\beta _{1},\beta _{2},h,\pi
)<\left( \tfrac{h\pi ^{2}+\pi }{h\pi ^{2}+h+2\pi }\right) %
\underaccent{\bar}{\phi}_{1}(\rho ,\alpha ,\beta _{1},\beta _{2},h,\pi ).
\end{equation*}%
This yields two inequalities as follows, 
\begin{equation}
\begin{array}{rcl}
&  & \left( \tfrac{2\pi }{\pi +1}-\tfrac{h\pi ^{2}+\pi }{h\pi ^{2}+h+2\pi }%
\right) \left( \beta _{1}(\beta _{2}+1-\rho )h^{2}\pi +(\beta _{1}+(1-\rho
)(1+\beta _{1}-\beta _{2}))h\pi ^{2}+(\beta _{2}+1-\rho )h\right. \\ 
&  & \left. +(\beta _{1}+\beta _{2}-\beta _{1}\beta _{2}+(1-\rho )(2-\beta
_{2}))\pi \right) -(1-\rho )\left( (\beta _{1}-\beta _{2}+1)h\pi ^{2}+\pi
\right) \\ 
& < & \alpha \left( (1-\rho )(\beta _{1}-\beta _{2})(h^{2}\pi +h\pi
^{2}+h+\pi )\left( \tfrac{2\pi }{\pi +1}-\tfrac{h\pi ^{2}+\pi }{h\pi
^{2}+h+2\pi }\right) \right. \\ 
&  & \left. +\beta _{2}(\beta _{1}+1-\rho )h^{2}\pi +(\beta _{1}-(1-\rho
)(\beta _{1}-\beta _{2}))h\pi ^{2}+\beta _{2}h+(\beta _{1}+\beta _{2}-\beta
_{1}\beta _{2}-(1-\rho )\beta _{1})\pi \right)%
\end{array}
\label{inequality:lower_bound}
\end{equation}%
and 
\begin{equation}
\begin{array}{rcl}
&  & \left( \tfrac{h\pi ^{2}+\pi }{h\pi ^{2}+h+2\pi }\right) \left( \beta
_{1}(\beta _{2}+1-\rho )h^{2}\pi +(\beta _{1}+(1-\rho )(1+\beta _{1}-\beta
_{2}))h\pi ^{2}+(\beta _{2}+1-\rho )h\right. \\ 
&  & \left. +(\beta _{1}+\beta _{2}-\beta _{1}\beta _{2}+(1-\rho )(2-\beta
_{2}))\pi \right) -(1-\rho )\left( (\beta _{1}-\beta _{2}+1)h\pi ^{2}+\pi
\right) \\ 
& > & \alpha \left( (1-\rho )(\beta _{1}-\beta _{2})(h^{2}\pi +h\pi
^{2}+h+\pi )\left( \tfrac{h\pi ^{2}+\pi }{h\pi ^{2}+h+2\pi }\right) \right.
\\ 
&  & \left. +\beta _{2}(\beta _{1}+1-\rho )h^{2}\pi +(\beta _{1}-(1-\rho
)(\beta _{1}-\beta _{2}))h\pi ^{2}+\beta _{2}h+(\beta _{1}+\beta _{2}-\beta
_{1}\beta _{2}-(1-\rho )\beta _{1})\pi \right)%
\end{array}
\label{inequality:upper_bound}
\end{equation}%
The coefficients of $\alpha $ in Eq.~\eqref{inequality:lower_bound} can be rewritten as  
\begin{equation*}
\beta_1\pi(h\pi+1)+\beta_2(h+\pi)+\beta_1\beta_2\pi(h^2-1)+\tfrac{(1-\rho)\pi(h-1)}{(\pi+1)(h\pi ^{2}+h+2\pi)}\left(\beta_1(h\pi+1)E_1+\beta_2(h+\pi)E_2\right), 
\end{equation*}
where $E_1=h\pi^2-h\pi+2h-\pi^2+3\pi>0$ and $E_2=2h\pi^2-h\pi+h+3\pi-1>0$. Similarly, the coefficient $\alpha $ in Eq.~\eqref{inequality:upper_bound} can be rewritten as
\begin{equation*}
\beta_1\beta_2\pi(h^2-1)+\tfrac{[\beta_1\pi(h\pi+1)+\beta_2(h+\pi)][(1-\rho)h^2\pi+h\pi^2+h+(1+\rho)\pi]}{h\pi^2+h+2\pi} > 0.
\end{equation*}
Both coefficients are strictly positive. Thus, we have 
\begin{equation*}
\underaccent{\bar}{\alpha}(\rho ,\beta _{1},\beta _{2},h,\pi )<\alpha <\bar{%
\alpha}(\rho ,\beta _{1},\beta _{2},h,\pi ),
\end{equation*}%
where {\small 
\begin{eqnarray*}
\lefteqn{\bar{\alpha}(\rho ,\beta _{1},\beta _{2},h,\pi )} \\
&=&\tfrac{\left( \tfrac{h\pi ^{2}+\pi }{h\pi ^{2}+h+2\pi }\right) (\beta
_{1}(\beta _{2}+1-\rho )h^{2}\pi +(\beta _{1}+(1-\rho )(1+\beta _{1}-\beta
_{2}))h\pi ^{2}+(\beta _{2}+1-\rho )h+(\beta _{1}+\beta _{2}-\beta _{1}\beta
_{2}+(1-\rho )(2-\beta _{2}))\pi )-(1-\rho )((1+\beta _{1}-\beta _{2})h\pi
^{2}+\pi )}{(1-\rho )(\beta _{1}-\beta _{2})(h^{2}\pi +h\pi^{2}+h+\pi
)\left( \tfrac{h\pi ^{2}+\pi }{h\pi ^{2}+h+2\pi}\right) +\beta _{2}(\beta
_{1}+1-\rho )h^{2}\pi +(\beta_{1}-(1-\rho )(\beta_{1}-\beta _{2}))h\pi
^{2}+\beta _{2}h+(\beta_{1}+\beta _{2}-\beta _{1}\beta _{2}-(1-\rho )\beta
_{1})\pi },
\end{eqnarray*}%
} and {\small 
\begin{eqnarray*}
\lefteqn{\underaccent{\bar}{\alpha}(\rho ,\beta _{1},\beta _{2},h,\pi )} \\
&=&\tfrac{\left( \tfrac{2\pi }{\pi +1}-\tfrac{h\pi ^{2}+\pi }{h\pi
^{2}+h+2\pi }\right) (\beta _{1}(\beta _{2}+1-\rho )h^{2}\pi +(\beta
_{1}+(1-\rho )(1+\beta _{1}-\beta _{2}))h\pi ^{2}+(\beta _{2}+1-\rho
)h+(\beta _{1}+\beta _{2}-\beta _{1}\beta _{2}+(1-\rho )(2-\beta _{2}))\pi
)-(1-\rho )((1+\beta _{1}-\beta _{2})h\pi ^{2}+\pi )}{(1-\rho )(\beta
_{1}-\beta _{2})(h^{2}\pi +h\pi ^{2}+h+\pi )\left( \tfrac{2\pi }{\pi +1}-%
\tfrac{h\pi ^{2}+\pi }{h\pi ^{2}+h+2\pi }\right) +\beta _{2}(\beta
_{1}+1-\rho )h^{2}\pi +(\beta _{1}-(1-\rho )(\beta _{1}-\beta _{2}))h\pi
^{2}+\beta _{2}h+(\beta _{1}+\beta _{2}-\beta _{1}\beta _{2}-(1-\rho )\beta
_{1})\pi }.
\end{eqnarray*}%
} 

\noindent We then show that
\begin{equation} \label{eq:bias_order}
\underaccent{\bar}{\alpha}(\rho,\beta_1,\beta_2,h,\pi) < \bar{\alpha}(\rho,\beta_1,\beta_2,h,\pi), \quad \textnormal{ for all } \rho, \beta_1, \beta_2 \in (0, 1) \textnormal{ and } h, \pi > 1. 
\end{equation}
For simplicity, we define 
\begin{equation*}
M_1 := \tfrac{h\pi^2+\pi}{h\pi^2+h+2\pi}, \quad M_2 := \tfrac{2\pi}{\pi+1}-M_1, 
\end{equation*}
and
\begin{equation*}
\begin{array}{rcl}
D_1 &:= & \beta_1(\beta_2+1-\rho)h^2\pi+(\beta_1+(1-\rho)(1+\beta_1-\beta_2))h\pi^2
+(\beta_2+1-\rho)h \\ 
& & +(\beta_1+\beta_2-\beta_1\beta_2+(1-\rho)(2-\beta_2))\pi,\\
D_2 & := & (1-\rho)((\beta_1-\beta_2+1)h\pi^2+\pi),\\
D_3 & := & (1-\rho)(\beta_1-\beta_2)(h^2\pi+h\pi^2+h+\pi),\\
D_4 & := & \beta_2(\beta_1+1-\rho)h^2\pi+(\rho\beta_1+(1-\rho)\beta_2)h\pi^2+\beta_2 h+
(\beta_2(1-\beta_1)+\rho\beta_1)\pi.
\end{array}
\end{equation*}
Then, we have
\begin{equation*}
\bar{\alpha}(\rho,\beta_1,\beta_2,h,\pi)=\tfrac{M_1D_1-D_2}{M_1D_3+D_4}, \quad
\underaccent{\bar}{\alpha}(\rho,\beta_1,\beta_2,h,\pi)=\tfrac{M_2D_1-D_2}{M_2D_3+D_4}.
\end{equation*}
Because $h, \pi>1$, we have $0<M_1, M_2<1$. In addition, $M_1D_3+D_4>0$ and 
$M_2D_3+D_4>0$ because they are the coefficients of $\alpha $ in 
Eq.~\eqref{inequality:lower_bound} and Eq.~\eqref{inequality:upper_bound}. 
A direct calculation yields
\begin{equation*}
\bar{\alpha}(\rho,\beta_1,\beta_2,h,\pi)-\underaccent{\bar}{\alpha}(\rho,\beta_1,\beta_2,h,\pi)
= \tfrac{(M_1-M_2)(D_1D_4+D_2D_3)}{(M_1D_3+D_4)(M_2D_3+D_4)}.
\end{equation*}
We also have
\begin{equation*}
M_1-M_2 = \tfrac{2\pi(h-1)(\pi-1)}{(\pi+1)(h\pi^2+h+2\pi)} > 0,
\end{equation*}
and
\begin{eqnarray*}
\lefteqn{D_1D_4+D_2D_3 = (\beta_1\beta_2\pi(h^2-1)+\beta_1\pi(h\pi+1)+\beta_2(h+\pi))(\beta_1\beta_2\pi(h^2-1)} \\
& & +\beta_1\bigl(h^2\pi(1-\rho)+h\pi^2+\pi\rho)+\beta_2(h^2\pi(1-\rho)+h+\pi\rho)+(1-\rho)(h^2\pi(1-\rho)+h\pi^2+h+\pi(1+\rho))) > 0.
\end{eqnarray*}
This yields Eq.~\eqref{eq:bias_order}. 

Putting these pieces together yields 
\begin{equation}
\Lambda _{\rho }=[0, 1] \cap \left( \sup_{\beta _{1},\beta _{2}\in (0,1),h\in \lbrack %
\underaccent{\bar}{h},\bar{h}]}\underaccent{\bar}{\alpha}(\rho ,\beta
_{1},\beta _{2},h,\pi ),\inf_{\beta _{1},\beta _{2}\in (0,1),h\in \lbrack %
\underaccent{\bar}{h},\bar{h}]}\bar{\alpha}(\rho ,\beta _{1},\beta
_{2},h,\pi )\right) .  \label{def:Lambda}
\end{equation}%
We introduce and prove two lemmas as follows,

\begin{lemma}
\label{Lemma:rho-1} Suppose that $\pi > 1$ is fixed. Then, we have that $%
\bar{\alpha}(\rho, \beta_1, \beta_2, h, \pi)$ and $\underaccent{\bar}{\alpha}%
(\rho, \beta_1, \beta_2, h, \pi)$ are continuous and strictly increasing in $%
\beta_1$ on the interval $[0, 1]$ for all $\rho, \beta_2 \in (0, 1)$ and $h
> 1$.
\end{lemma}

\noindent\emph{Proof}. We have 
\begin{equation*}
\tfrac{\partial \bar{\alpha}}{\partial \beta_1}(\rho, \beta_1, \beta_2, h,
\pi) = \tfrac{P^{\prime }(\beta_1)Q(\beta_1)-P(\beta_1)Q^{\prime }(\beta_1)}{%
(Q(\beta_1))^2},
\end{equation*}
where 
\begin{equation*}
\begin{array}{rcl}
P(\beta_1) & = & \left(\tfrac{h\pi^2+\pi}{h\pi^2+h+2\pi}\right)\left(%
\beta_1(\beta_2+1-\rho)h^2\pi + (\beta_1+(1-\rho)(1+\beta_1-\beta_2))h\pi^2
+ (\beta_2+1-\rho)h\right. \\ 
&  & \left.+(\beta_1+\beta_2-\beta_1\beta_2+(1-\rho)(2-\beta_2))\pi%
\right)-(1-\rho)\left((\beta_1-\beta_2+1)h\pi^2+\pi\right), \\ 
Q(\beta_1) & = & (1-\rho)(\beta_1-\beta_2)(h^2\pi+h\pi^2+h+\pi)\left(\tfrac{%
h\pi^2+\pi}{h\pi^2+h+2\pi}\right) \\ 
&  & +\beta_2(\beta_1+1-\rho)h^2\pi+(\beta_1-(1-\rho)(\beta_1-\beta_2))h%
\pi^2+\beta_2 h+(\beta_1+\beta_2-\beta_1\beta_2-(1-\rho)\beta_1)\pi.%
\end{array}%
\end{equation*}
It suffices to show that $P^{\prime }(\beta_1)Q(\beta_1)-P(\beta_1)Q^{\prime
}(\beta_1) > 0$ for all $\rho, \beta_2 \in (0, 1)$ and $h, \pi > 1$. Indeed,
we have 
\begin{equation*}
P^{\prime }(\beta_1)Q(\beta_1)-P(\beta_1)Q^{\prime }(\beta_1) = \left(\tfrac{%
\beta_2 \pi^3(1-\rho)(h-1)^2(h+1)^2}{(h\pi^2+h+2\pi)^2}\right) L(\rho,
\beta_2, h, \pi),
\end{equation*}
where $L(\rho, \beta_2, h, \pi)=\pi(1+\beta_2-\rho)h^2+(\pi^2 +
1)h+\pi(1-\beta_2+\rho)$. For all $\rho, \beta_2 \in (0, 1)$ and $h>1$, we
have $L(\rho, \beta_2, h, \pi) > 0$. This yields the desired result.

By abuse of notation, we also have 
\begin{equation*}
\tfrac{\partial \underaccent{\bar}{\alpha}}{\partial \beta_1}(\rho, \beta_1,
\beta_2, h, \pi) = \tfrac{P^{\prime }(\beta_1)Q(\beta_1)-P(\beta_1)Q^{\prime
}(\beta_1)}{(Q(\beta_1))^2},
\end{equation*}
where 
\begin{equation*}
\begin{array}{rcl}
P(\beta_1) & = & \left(\tfrac{2\pi}{\pi+1} - \tfrac{h\pi^2+\pi}{h\pi^2+h+2\pi%
}\right)\left(\beta_1(\beta_2+1-\rho)h^2\pi +
(\beta_1+(1-\rho)(1+\beta_1-\beta_2))h\pi^2 + (\beta_2+1-\rho)h\right. \\ 
&  & \left. +
(\beta_1+\beta_2-\beta_1\beta_2+(1-\rho)(2-\beta_2))\pi\right)-(1-\rho)%
\left((\beta_1-\beta_2+1)h\pi^2+\pi\right), \\ 
Q(\beta_1) & = & (1-\rho)(\beta_1-\beta_2)(h^2\pi+h\pi^2+h+\pi)\left(\tfrac{%
2\pi}{\pi+1} - \tfrac{h\pi^2+\pi}{h\pi^2+h+2\pi}\right) \\ 
&  & +\beta_2(\beta_1+1-\rho)h^2\pi+(\beta_1-(1-\rho)(\beta_1-\beta_2))h%
\pi^2+\beta_2 h+(\beta_1+\beta_2-\beta_1\beta_2-(1-\rho)\beta_1)\pi.%
\end{array}%
\end{equation*}
It suffices to show that $P^{\prime }(\beta_1)Q(\beta_1)-P(\beta_1)Q^{\prime
}(\beta_1) > 0$ for all $\rho, \beta_2 \in (0, 1)$ and $h, \pi > 1$. Indeed,
we have 
\begin{equation*}
P^{\prime }(\beta_1)Q(\beta_1)-P(\beta_1)Q^{\prime }(\beta_1) = \left(\tfrac{%
\pi^2(1-\rho)(h-1)}{(\pi+1)^2(h\pi^2+h+2\pi)^2}\right)L(\rho,\beta_2,h,\pi),
\end{equation*}
where $L(\rho,\beta_2,h,\pi)=L_0(\beta_2,h,\pi)((1-\rho)L_1(\beta_2,h,\pi)+%
\rho L_2(\beta_2,h,\pi))$ with 
\begin{equation*}
L_0(\beta_2,h,\pi) = \beta_2A(h,\pi)+B(h,\pi), \quad L_1(\beta_2,h,\pi) =
\beta_2C(h,\pi)+D(h,\pi), \quad L_2(\beta_2,h,\pi) =
\beta_2C(h,\pi)+E(h,\pi),
\end{equation*}
where 
\begin{equation*}
\begin{array}{rcl}
A(h,\pi) & = & 2h^3\pi^3-h^3\pi^2+h^3\pi+3h^2\pi^2-h^2\pi-2h\pi^3+h\pi^2-h%
\pi-3\pi^2+\pi, \\ 
B(h,\pi) & = & 2h^2\pi^4-2h^2\pi^3+2h^2\pi^2-2h^2\pi+6h\pi^3-6h\pi^2+2h%
\pi-2h+4\pi^2-4\pi, \\ 
C(h,\pi) & = & h^2\pi^2-h^2\pi+2h^2-2h\pi^2+4h\pi-2h+\pi^2-3\pi, \\ 
D(h,\pi) & = & h^2\pi^2-h^2\pi+2h^2+h\pi^3-h\pi^2+5h\pi-h+3\pi^2-\pi, \\ 
E(h,\pi) & = & h\pi^3+h\pi^2+h\pi+h+2\pi^2+2\pi.%
\end{array}%
\end{equation*}
Because $h, \pi > 1$, we have 
\begin{equation*}
\begin{array}{rcl}
A(h,\pi) & = & \pi(h-1)(h+1)(h(2\pi^2-\pi+1)+(3\pi-1)) \ > \ 0, \\ 
B(h,\pi) & = & 2(\pi-1)(h^2\pi(\pi^2+1)+h(3\pi^2+1)+2\pi) \ > \ 0, \\ 
D(h,\pi) & = & h^2(\pi^2-\pi+2)+h(\pi^3-\pi^2+5\pi-1)+\pi(3\pi-1) \ > \ 0,%
\end{array}%
\end{equation*}
and 
\begin{equation*}
\begin{array}{rcl}
D(h,\pi)+C(h,\pi) & = & 2h^2(\pi^2-\pi+2)+h(\pi^3-3\pi^2+9\pi-3)+4\pi(\pi-1)
\ > \ 0, \\ 
E(h,\pi)+C(h,\pi) & = & h^2(\pi^2-\pi+2)+h(\pi^3-\pi^2+5\pi-1)+\pi(3\pi-1) \
> \ 0.%
\end{array}%
\end{equation*}
In addition, we have that $L_0(\beta_2,h,\pi)$, $L_1(\beta_2,h,\pi)$ and $%
L_2(\beta_2,h,\pi)$ are all linear in $\beta_2$ and $\beta_2 \in (0, 1)$.
Putting these pieces together yields that $L_0(\beta_2,h,\pi)>0$, $%
L_1(\beta_2,h,\pi)>0$ and $L_2(\beta_2,h,\pi)>0$ for all $\beta_2 \in (0, 1)$
and $h, \pi > 1$. By definition of $L(\cdot)$, we have $L(\rho,\beta_2,h,%
\pi)>0$ for all $\rho, \beta_2 \in (0, 1)$ and $h, \pi > 1$. This yields the
desired result. $\qed$

\begin{lemma}
\label{Lemma:rho-2} Suppose that $\pi > 1$ is fixed and $\underaccent{%
\bar}{h} > 2\pi$. Then, we have that $\underaccent{\bar}{\alpha}(\rho, 1,
\beta_2, h, \pi)$ is continuous and strictly decreasing in $\beta_2$ on the
interval $[0, 1]$ for all $\rho \in (0, 1)$ and $h \geq \underaccent{\bar}{h}
$.
\end{lemma}

\noindent\emph{Proof}. We have 
\begin{equation*}
\underaccent{\bar}{\alpha}(\rho, 1, \beta_2, h, \pi) = \tfrac{\left(\tfrac{%
2\pi}{\pi+1} - \tfrac{h\pi^2+\pi}{h\pi^2+h+2\pi}\right)\left((\beta_2+1-%
\rho)h^2\pi + (1+(1-\rho)(2-\beta_2))h\pi^2 + (\beta_2+1-\rho)h +
(1+(1-\rho)(2-\beta_2))\pi\right) - (1-\rho)\left((2-\beta_2)h\pi^2 +
\pi\right)}{(1-\rho)(1-\beta_2)(h^2\pi+h\pi^2+h+\pi)\left(\tfrac{2\pi}{\pi+1}
- \tfrac{h\pi^2+\pi}{h\pi^2+h+2\pi}\right)+ \beta_2(2-\rho)h^2\pi +
(1-(1-\rho)(1-\beta_2))h\pi^2 + \beta_2 h + \rho \pi}.
\end{equation*}
This implies 
\begin{equation*}
\tfrac{\partial \underaccent{\bar}{\alpha}(\rho, 1, \beta_2, h, \pi)}{%
\partial \beta_2} = \tfrac{P^{\prime
}(\beta_2)Q(\beta_2)-P(\beta_2)Q^{\prime }(\beta_2)}{(Q(\beta_2))^2},
\end{equation*}
where 
\begin{equation*}
\begin{array}{rcl}
P(\beta_2) & = & \left(\tfrac{2\pi}{\pi+1} - \tfrac{h\pi^2+\pi}{h\pi^2+h+2\pi%
}\right)\left((\beta_2+1-\rho)h^2\pi + (1+(1-\rho)(2-\beta_2))h\pi^2 +
(\beta_2+1-\rho)h\right. \\ 
&  & \left. +
(1+(1-\rho)(2-\beta_2))\pi\right)-(1-\rho)\left((2-\beta_2)h\pi^2+\pi\right),
\\ 
Q(\beta_2) & = & (1-\rho)(1-\beta_2)(h^2\pi+h\pi^2+h+\pi)\left(\tfrac{2\pi}{%
\pi+1} - \tfrac{h\pi^2+\pi}{h\pi^2+h+2\pi}\right) \\ 
&  & +\beta_2(2-\rho)h^2\pi+(1-(1-\rho)(1-\beta_2))h\pi^2+\beta_2 h+\rho\pi.%
\end{array}%
\end{equation*}
It suffices to show that $P^{\prime }(\beta_2)Q(\beta_2)-P(\beta_2)Q^{\prime
}(\beta_2) < 0$ for all $\rho \in (0, 1)$, $h \geq \underaccent{\bar}{h}$
and $\pi > 1$. Indeed, we have 
\begin{equation*}
P^{\prime }(\beta_2)Q(\beta_2)-P(\beta_2)Q^{\prime }(\beta_2) = \left(\tfrac{%
\pi(1-\rho)(h-1)}{(\pi+1)^2(h\pi^2+h+2\pi)^2}\right)L(\rho,h,\pi),
\end{equation*}
where $L(\rho,h,\pi) = L_0(h,\pi)+(1-\rho)L_1(h,\pi)$ with 
\begin{equation*}
\begin{array}{rcl}
L_0(h,\pi) & = & -(h\pi+1)(h(2\pi^2-\pi+1)-\pi^2+3\pi)R(h,\pi), \\ 
L_1(h,\pi) & = & -\pi(h-1)(h(2\pi^2-\pi+1)+3\pi-1)R(h,\pi), \\ 
R(h,\pi) & = & h^3(\pi^3-\pi^2+2\pi)-h^2(3\pi^3-5\pi^2+2\pi-2)-h(2\pi^4-%
\pi^3+5\pi^2-4\pi)-(3\pi^3-\pi^2).%
\end{array}%
\end{equation*}
In what follows, we prove that $R(h, \pi)>0$ for all $h \geq %
\underaccent{\bar}{h}$. Indeed, we have 
\begin{equation*}
\tfrac{\partial^3 R}{\partial h^3}(h, \pi) = 6(\pi^3-\pi^2+2\pi) =
6\pi(\pi^2-\pi+2) > 0.
\end{equation*}
This implies that $\frac{\partial^2 R}{\partial h^2}(h, \pi)$ is strictly
increasing in $h$. Because $h \geq \underaccent{\bar}{h} > 2\pi$ and $\pi>1$,
we have 
\begin{equation*}
\begin{array}{rcl}
\tfrac{\partial^2 R}{\partial h^2}(h, \pi) & > & \tfrac{\partial^2 R}{%
\partial h^2}(2\pi,\pi) = 12\pi^4-18\pi^3+34\pi^2-4\pi+4 \\ 
& = & 12\pi^2(\pi-1)^2+4\pi(\pi^2-1)+2\pi^3+22\pi^2+4 > 0.%
\end{array}%
\end{equation*}
This implies that $\frac{\partial R}{\partial h}(h, \pi)$ is strictly
increasing in $h$ on the interval $[\underaccent{\bar}{h}, +\infty)$. Because $%
h \geq \underaccent{\bar}{h} > 2\pi$ and $\pi>1$, we have 
\begin{equation*}
\begin{array}{rcl}
\tfrac{\partial R}{\partial h}(h, \pi) & > & \tfrac{\partial R}{\partial h}%
(2\pi,\pi) = 12\pi^5-26\pi^4+45\pi^3-13\pi^2+12\pi \\ 
& = & 12\pi^2(\pi-1)^3+\pi^2(\pi^2-1)+9\pi^4+9\pi^3+12\pi > 0.%
\end{array}%
\end{equation*}
This implies that $R(h, \pi)$ is strictly increasing in $h$ on the interval $%
[\underaccent{\bar}{h}, +\infty)$. Because $h \geq \underaccent{\bar}{h} >
2\pi $ and $\pi>1$, we have 
\begin{equation*}
\begin{array}{rcl}
R(h, \pi) & > & R(2\pi,\pi) = 8\pi^6-24\pi^5+38\pi^4-21\pi^3+17\pi^2 \\ 
& = & 8\pi^3(\pi-1)^3+13\pi^3(\pi-1)+\pi^4+17\pi^2 > 0.%
\end{array}%
\end{equation*}
Because $h \geq \underaccent{\bar}{h} > 2\pi$ and $\pi > 1$, we have that $%
h(2\pi^2-\pi+1)+3\pi-1 > 0$ and 
\begin{equation*}
h(2\pi^2-\pi+1)-\pi^2+3\pi \geq 2\pi(2\pi^2-\pi+1)-\pi^2+3\pi =
4\pi^3-3\pi^2+5\pi > 0.
\end{equation*}
Putting these pieces together yields that $L_0(h,\pi), L_1(h,\pi)<0$ for all 
$h \geq \underaccent{\bar}{h}$ and $\pi > 1$. By definition of $L(\cdot)$,
we have $L(\rho,h,\pi)<0$ for all $\rho \in (0, 1)$, $h \geq %
\underaccent{\bar}{h}$ and $\pi > 1$. This yields the desired result. $\qed$

Back to the original proof of Theorem~\ref{thm:rho}, we see from the
definition of $\bar{\alpha}(\cdot )$ that 
\begin{equation*}
\bar{\alpha}(\rho ,0,\beta _{2},h,\pi )=\tfrac{\pi (\rho (h^{2}\pi -\pi
)-(2h^{2}\pi +h\pi ^{2}+h))}{(h+\pi )(\rho (h^{2}\pi -\pi )-(h^{2}\pi +h\pi
^{2}+h+\pi ))},\quad \textnormal{for all } \beta _{2}\in (0,1).
\end{equation*}%
Using Lemma~\ref{Lemma:rho-1} and Lemma~\ref{Lemma:rho-2}, we have 
\begin{equation}
\begin{array}{rcl}
\inf\limits_{\beta _{1},\beta _{2}\in (0,1),h\in \lbrack \underaccent{%
\bar}{h},\bar{h}]}\bar{\alpha}(\rho ,\beta _{1},\beta _{2},h,\pi ) & = & 
\inf\limits_{\beta _{2}\in (0,1),h\in \lbrack \underaccent{\bar}{h},\bar{h}]}%
\bar{\alpha}(\rho ,0,\beta _{2},h,\pi )=\inf\limits_{h\in \lbrack %
\underaccent{\bar}{h},\bar{h}]}\bar{g}(\rho ,h,\pi ), \\ 
\sup\limits_{\beta _{1},\beta _{2}\in (0,1),h\in \lbrack \underaccent{%
\bar}{h},\bar{h}]}\underaccent{\bar}{\alpha}(\rho ,\beta _{1},\beta
_{2},h,\pi ) & = & \sup\limits_{h\in \lbrack \underaccent{\bar}{h},\bar{h}]}%
\underaccent{\bar}{\alpha}(\rho ,1,0,h,\pi )=\sup\limits_{h\in \lbrack %
\underaccent{\bar}{h},\bar{h}]}\underaccent{\bar}{g}(\rho ,h,\pi ),%
\end{array}
\label{def:main}
\end{equation}%
where $\bar{g}$ and $\underaccent{\bar}{g}$ are given by 
\begin{equation*}
\begin{array}{rcl}
\bar{g}(\rho ,h,\pi ) & = & \tfrac{\pi (\rho (h^{2}\pi -\pi )-(2h^{2}\pi
+h\pi ^{2}+h))}{(h+\pi )(\rho (h^{2}\pi -\pi )-(h^{2}\pi +h\pi ^{2}+h+\pi ))}%
, \\ 
\underaccent{\bar}{g}(\rho ,h,\pi ) & = & \tfrac{\left( \tfrac{2\pi }{\pi +1}%
-\tfrac{h\pi ^{2}+\pi }{h\pi ^{2}+h+2\pi }\right) ((1-\rho )h^{2}\pi
+(3-2\rho )h\pi ^{2}+(1-\rho )h+(3-2\rho )\pi )-(1-\rho )(2h\pi ^{2}+\pi )}{%
(1-\rho )(h^{2}\pi +h\pi ^{2}+h+\pi )\left( \tfrac{2\pi }{\pi +1}-\tfrac{%
h\pi ^{2}+\pi }{h\pi ^{2}+h+2\pi }\right) +\rho h\pi ^{2}+\rho \pi }.%
\end{array}%
\end{equation*}

\paragraph{Monotonicity results.}

We prove the monotonicity results as follows,

\begin{itemize}
\item $\inf_{\beta_1, \beta_2 \in (0, 1), h \in [\underaccent{\bar}{h}, \bar{%
h}]} \bar{\alpha}(\rho, \beta_1, \beta_2, h, \pi)$ is increasing in $\rho$
on the interval $(0, 1)$.

\item $\sup_{\beta_1, \beta_2 \in (0, 1), h \in [\underaccent{\bar}{h}, \bar{%
h}]} \underaccent{\bar}{\alpha}(\rho, \beta_1, \beta_2, h, \pi)$ is
decreasing in $\rho$ on the interval $(0, 1)$.
\end{itemize}
Based on Eq.~\eqref{def:main}, it suffices to show that

\begin{itemize}
\item $\bar{g}(\rho, h, \pi)$ is increasing in $\rho$ on the interval $(0,
1) $ for any $h \in [\underaccent{\bar}{h}, \bar{h}]$.

\item $\underaccent{\bar}{g}(\rho, h, \pi)$ is decreasing in $\rho$ on the
interval $(0, 1)$ for any $h \in [\underaccent{\bar}{h}, \bar{h}]$.
\end{itemize}
Indeed, we have 
\begin{equation*}
\tfrac{\partial \bar{g}}{\partial \rho}(\rho,h,\pi) = \tfrac{%
\pi^{3}(h-1)^2(h+1)^2}{(h+\pi)(\rho(h^2\pi-\pi)-(h^2\pi+h\pi^2+h+\pi))^2}.
\end{equation*}
Because $h \geq \underaccent{\bar}{h} > 1$, we have $(h-1)^2>0$. In addition, $%
\rho \in (0, 1)$. Thus, we have 
\begin{equation*}
\rho(h^2\pi-\pi)-(h^2\pi+h\pi^2+h+\pi) < -(h\pi^2+h+2\pi) < 0.
\end{equation*}
Putting these pieces together yields that $\frac{\partial \bar{g}}{\partial
\rho}(\rho,h,\pi)>0$ for any $\rho \in (0, 1)$ and any $h \in [%
\underaccent{\bar}{h}, \bar{h}]$. This implies that $\bar{g}(\rho, h, \pi)$
is increasing in $\rho$ on the interval $(0, 1)$ for any $h \in [%
\underaccent{\bar}{h}, \bar{h}]$.

Proceeding a further step, we have 
\begin{equation*}
\tfrac{\partial\underaccent{\bar}{g}}{\partial\rho}(\rho,h,\pi) = -\tfrac{%
(h-1)R(h,\pi)}{(h\pi+1)(\rho A(h,\pi)-B(h, \pi))^2},
\end{equation*}
where 
\begin{equation*}
\begin{array}{rcl}
R(h,\pi) & = & h^3(2\pi^5-3\pi^4+6\pi^3-3\pi^2+2\pi)-h^2(2\pi^6+4\pi^5-11%
\pi^4+14\pi^3-15\pi^2+4\pi-2) \\ 
&  & -h(6\pi^5+13\pi^4-18\pi^3+13\pi^2-10\pi)-(5\pi^4+10\pi^3-11\pi^2), \\ 
A(h,\pi) & = & (h-1)(h\pi^2-h\pi+2h-\pi^2+3\pi), \\ 
B(h,\pi) & = & (h+\pi)(h\pi^2-h\pi+2h+3\pi-1).%
\end{array}%
\end{equation*}
Because $\rho \in (0, 1)$, we have 
\begin{equation*}
\rho A(h,\pi)-B(h, \pi) < A(h,\pi)-B(h, \pi) = -(\pi+1)(h\pi^2+h+2\pi) < 0.
\end{equation*}
Putting these pieces together yields that the sign of $\frac{\partial%
\underaccent{\bar}{g}}{\partial\rho}(\rho,h,\pi)$ is the same as $-R(h, \pi)$%
.

In what follows, we prove that $R(h, \pi)>0$ for any $h \in [%
\underaccent{\bar}{h}, \bar{h}]$. Indeed, we have 
\begin{equation*}
\tfrac{\partial^3 R}{\partial h^3}(h, \pi) =
6(2\pi^5-3\pi^4+6\pi^3-3\pi^2+2\pi) = 6\pi(\pi^2-\pi+2)(2\pi^2-\pi+1) > 0.
\end{equation*}
This implies that $\frac{\partial^2 R}{\partial h^2}(h, \pi)$ is strictly
increasing in $h$. Because $h \geq \underaccent{\bar}{h} > 2\pi$ and $\pi>1$,
we have 
\begin{equation*}
\begin{array}{rcl}
\tfrac{\partial^2 R}{\partial h^2}(h, \pi) & > & \tfrac{\partial^2 R}{%
\partial h^2}(2\pi,\pi) = 20\pi^6-44\pi^5+94\pi^4-64\pi^3+54\pi^2-8\pi+4 \\ 
& = & 20\pi^3(\pi-1)^3+16\pi^3(\pi^2-1)+34\pi^3(\pi-1)+8\pi(\pi-1)+6\pi^3+46%
\pi^2+4 > 0.%
\end{array}%
\end{equation*}
This implies that $\frac{\partial R}{\partial h}(h, \pi)$ is strictly
increasing in $h$ on the interval $[\underaccent{\bar}{h}, +\infty)$. Because $%
h \geq \underaccent{\bar}{h} > 2\pi$ and $\pi>1$, we have 
\begin{equation*}
\begin{array}{rcl}
\tfrac{\partial R}{\partial h}(h, \pi) & > & \tfrac{\partial R}{\partial h}%
(2\pi,\pi) = 16\pi^7-52\pi^6+110\pi^5-105\pi^4+102\pi^3-29\pi^2+18\pi, \\ 
& = & 16\pi^3(\pi-1)^4+12\pi^2(\pi^2-1)^2+14\pi^3(\pi-1)^2+41\pi^2(\pi-1)+11%
\pi^4+31\pi^3+18\pi > 0.%
\end{array}%
\end{equation*}
This implies that $R(h, \pi)$ is strictly increasing in $h$ on the interval $%
[\underaccent{\bar}{h}, +\infty)$. Because $h \geq \underaccent{\bar}{h} >
2\pi $, we have 
\begin{equation*}
R(h, \pi) > R(2\pi,\pi) =
\pi^2(8\pi^6-40\pi^5+80\pi^4-106\pi^3+107\pi^2-52\pi+39).
\end{equation*}
Because $\pi>1$, we let $t=\pi-1>0$ for simplicity. Then, we have 
\begin{equation*}
8\pi^6-40\pi^5+80\pi^4-106\pi^3+107\pi^2-52\pi+39 =
(8t^6+36-26t^3)+t(8t^4+12-11t).
\end{equation*}
For the first term, we have 
\begin{equation*}
8t^6+36-26t^3 \geq 24\sqrt{2}t^3-26t^3 > 0.
\end{equation*}
For the second term, we have 
\begin{equation*}
8t^4+12-11t \geq 4(8\cdot4\cdot4\cdot4)^{1/4}t-11t = (16\sqrt[4]{2}-11)t > 0.
\end{equation*}
Putting these pieces together yields that $R(h, \pi)>0$ for any $h \in [%
\underaccent{\bar}{h}, \bar{h}]$. Thus, we have that $\frac{\partial%
\underaccent{\bar}{g}}{\partial\rho}(\rho,h,\pi)<0$ for any $\rho \in (0, 1)$
and any $h \in [\underaccent{\bar}{h}, \bar{h}]$. This implies that $%
\underaccent{\bar}{g}(\rho, h, \pi)$ is decreasing in $\rho$ on the interval 
$(0, 1)$ for any $h \in [\underaccent{\bar}{h}, \bar{h}]$.

\paragraph{Boundary results.}

We prove the boundary results as follows,

\begin{itemize}
\item The following statement holds true, 
\begin{equation*}
\sup_{\beta _{1},\beta _{2}\in (0,1),h\in \lbrack \underaccent{\bar}{h},\bar{%
h}]}\underaccent{\bar}{\alpha}(\rho ,\beta _{1},\beta _{2},h,\pi )\geq
\inf_{\beta _{1},\beta _{2}\in (0,1),h\in \lbrack \underaccent{\bar}{h},\bar{%
h}]}\bar{\alpha}(\rho ,\beta _{1},\beta _{2},h,\pi ),\quad \textnormal{for
all } \rho \in (0, \tfrac{1}{2}].
\end{equation*}

\item Suppose that $\epsilon \in \left(0, \frac{1}{2}\left(\frac{%
\underaccent{\bar}{h}\pi^2+\pi}{\underaccent{\bar}{h}\pi^2+%
\underaccent{\bar}{h}+2\pi}-\tfrac{\pi}{\pi+1}\right)\right)$ is fixed.
Then, there exists $\delta \in (0, \frac{1}{2})$ such that, for all $\rho
\in (1-\delta, 1)$, the following statement holds true, 
\begin{equation*}
\begin{array}{rcl}
\inf_{\beta_1, \beta_2 \in (0, 1), h \in [\underaccent{\bar}{h}, \bar{h}]} 
\bar{\alpha}(\rho, \beta_1, \beta_2, h, \pi) & \geq & \tfrac{%
\underaccent{\bar}{h}\pi^2+\pi}{\underaccent{\bar}{h}\pi^2+%
\underaccent{\bar}{h}+2\pi} - \epsilon, \\ 
\sup_{\beta_1, \beta_2 \in (0, 1), h \in [\underaccent{\bar}{h}, \bar{h}]} %
\underaccent{\bar}{\alpha}(\rho, \beta_1, \beta_2, h, \pi) & \leq & \left(%
\tfrac{2\pi}{\pi+1} - \tfrac{\underaccent{\bar}{h}\pi^2+\pi}{%
\underaccent{\bar}{h}\pi^2+\underaccent{\bar}{h}+2\pi}\right) + \epsilon.%
\end{array}%
\end{equation*}
\end{itemize}

Based on Eq.~\eqref{def:main}, it suffices to show that

\begin{itemize}
\item The following statement holds true, 
\begin{equation*}
\sup\limits_{h\in \lbrack \underaccent{\bar}{h},\bar{h}]}\underaccent{%
\bar}{g}(\rho ,h,\pi )\geq \inf\limits_{h\in \lbrack \underaccent{\bar}{h},%
\bar{h}]}\bar{g}(\rho ,h,\pi ),\quad \textnormal{for all } \rho \in (0,%
\tfrac{1}{2}].
\end{equation*}

\item Suppose that $\epsilon \in \left(0, \frac{1}{2}\left(\frac{%
\underaccent{\bar}{h}\pi^2+\pi}{\underaccent{\bar}{h}\pi^2+%
\underaccent{\bar}{h}+2\pi}-\tfrac{\pi}{\pi+1}\right)\right)$ is fixed.
Then, there exists $\delta \in (0, \frac{1}{2})$ such that, for all $\rho
\in (1-\delta, 1)$, the following statement holds true, 
\begin{equation*}
\inf\limits_{h \in [\underaccent{\bar}{h}, \bar{h}]} \bar{g}(\rho, h, \pi)
\geq \tfrac{\underaccent{\bar}{h}\pi^2+\pi}{\underaccent{\bar}{h}\pi^2+%
\underaccent{\bar}{h}+2\pi} - \epsilon, \quad \sup\limits_{h \in [%
\underaccent{\bar}{h}, \bar{h}]} \underaccent{\bar}{g}(\rho, h, \pi) \leq
\left(\tfrac{2\pi}{\pi+1} - \tfrac{\underaccent{\bar}{h}\pi^2+\pi}{%
\underaccent{\bar}{h}\pi^2+\underaccent{\bar}{h}+2\pi}\right) + \epsilon.
\end{equation*}
\end{itemize}

Indeed, we have 
\begin{equation*}
\begin{array}{rcl}
\bar{g}(\rho ,h,\pi ) & = & \tfrac{\pi (\rho (h^{2}\pi -\pi )-(2h^{2}\pi
+h\pi ^{2}+h))}{(h+\pi )(\rho (h^{2}\pi -\pi )-(h^{2}\pi +h\pi ^{2}+h+\pi ))}%
, \\ 
\underaccent{\bar}{g}(\rho ,h,\pi ) & = & \tfrac{\left( \tfrac{2\pi }{\pi +1}%
-\tfrac{h\pi ^{2}+\pi }{h\pi ^{2}+h+2\pi }\right) ((1-\rho )h^{2}\pi
+(3-2\rho )h\pi ^{2}+(1-\rho )h+(3-2\rho )\pi )-(1-\rho )(2h\pi ^{2}+\pi )}{%
(1-\rho )(h^{2}\pi +h\pi ^{2}+h+\pi )\left( \tfrac{2\pi }{\pi +1}-\tfrac{%
h\pi ^{2}+\pi }{h\pi ^{2}+h+2\pi }\right) +\rho h\pi ^{2}+\rho \pi }.%
\end{array}%
\end{equation*}%
For the first boundary result, it suffices to show 
\begin{equation}
\underaccent{\bar}{g}(\rho ,\bar{h},\pi )>\bar{g}(\rho ,\bar{h},\pi ),\quad
\textnormal{for all } \rho \in (0,\tfrac{1}{2}].  \label{boundary:main_1}
\end{equation}%
Because $\pi >1$, $\rho \in (0,\frac{1}{2}]$ and $\bar{h}>20\pi >1$, we have 
\begin{equation*}
\rho (\bar{h}^{2}\pi -\pi )-(2\bar{h}^{2}\pi +\bar{h}\pi ^{2}+\bar{h}%
)<0,\quad \rho (\bar{h}^{2}\pi -\pi )-(\bar{h}^{2}\pi +\bar{h}\pi ^{2}+\bar{h%
}+\pi )<0.
\end{equation*}%
Thus, we rewrite 
\begin{equation*}
\bar{g}(\rho ,\bar{h},\pi )=\tfrac{\pi ((2\bar{h}^{2}\pi +\bar{h}\pi ^{2}+%
\bar{h})-\rho (\bar{h}^{2}\pi -\pi ))}{(\bar{h}+\pi )((\bar{h}^{2}\pi +\bar{h%
}\pi ^{2}+\bar{h}+\pi )-\rho (\bar{h}^{2}\pi -\pi ))}=\tfrac{\pi ((2-\rho )%
\bar{h}^{2}\pi +(\pi ^{2}+1)\bar{h}+\rho \pi )}{(\bar{h}+\pi )((1-\rho )\bar{%
h}^{2}\pi +(\pi ^{2}+1)\bar{h}+(1+\rho )\pi )}.
\end{equation*}%
Note that $(1-\rho )\bar{h}^{2}\pi +(\pi ^{2}+1)\bar{h}+(1+\rho )\pi
>(1-\rho )\bar{h}^{2}\pi >0$ and $\bar{h}+\pi >\bar{h}>0$. Thus, we have 
\begin{equation*}
\bar{g}(\rho ,\bar{h},\pi )<\tfrac{(2-\rho )\bar{h}^{2}\pi +(\pi ^{2}+1)\bar{%
h}+\rho \pi }{(1-\rho )\bar{h}^{3}}=\tfrac{(2-\rho )\pi }{(1-\rho )\bar{h}}+%
\tfrac{\pi ^{2}+1}{(1-\rho )\bar{h}^{2}}+\tfrac{\rho \pi }{(1-\rho )\bar{h}%
^{3}}.
\end{equation*}%
Because $\rho \in (0,\frac{1}{2}]$, we have 
\begin{equation*}
\tfrac{2-\rho }{1-\rho }=1+\tfrac{1}{1-\rho }\leq 3,\quad \tfrac{1}{1-\rho }%
\leq 2,\quad \tfrac{\rho }{1-\rho }\leq 1.
\end{equation*}%
Putting these pieces together yields 
\begin{equation*}
\bar{g}(\rho ,\bar{h},\pi )\leq \tfrac{3\pi }{\bar{h}}+\tfrac{2(\pi ^{2}+1)}{%
\bar{h}^{2}}+\tfrac{\pi }{\bar{h}^{3}},\quad \textnormal{for all }\rho \in (0,\tfrac{1}{2}].
\end{equation*}%
Because $\bar{h}>20\pi $, we have 
\begin{equation*}
\tfrac{3\pi }{\bar{h}}<\tfrac{3}{20},\quad \tfrac{2(\pi ^{2}+1)}{\bar{h}^{2}}%
<\tfrac{2(\pi ^{2}+1)}{400\pi ^{2}}<\tfrac{4\pi ^{2}}{400\pi ^{2}}=\tfrac{1}{%
100},\quad \tfrac{\pi }{\bar{h}^{3}}<\tfrac{\pi }{8000\pi ^{3}}=\tfrac{1}{%
8000\pi ^{2}}<\tfrac{1}{8000}.
\end{equation*}%
Putting these pieces together yields 
\begin{equation}
\bar{g}(\rho ,\bar{h},\pi )<\tfrac{3}{20}+\tfrac{1}{100}+\tfrac{1}{8000}%
<0.17,\quad \textnormal{for all } \rho \in (0,\tfrac{1}{2}].
\label{boundary:upper_threshold_1}
\end{equation}%
Because $\pi >1$ and $\bar{h}>2\pi $, we have 
\begin{equation*}
\tfrac{1}{2}\leq \tfrac{2\pi }{\pi +1}-\tfrac{\bar{h}\pi ^{2}+\pi }{\bar{h}%
\pi ^{2}+\bar{h}+2\pi }<1.
\end{equation*}%
Then, we have 
\begin{equation*}
\begin{array}{rcl}
&  & \left( \tfrac{2\pi }{\pi +1}-\tfrac{\bar{h}\pi ^{2}+\pi }{\bar{h}\pi
^{2}+\bar{h}+2\pi }\right) ((1-\rho )\bar{h}^{2}\pi +(3-2\rho )\bar{h}\pi
^{2}+(1-\rho )\bar{h}+(3-2\rho )\pi )-(1-\rho )(2\bar{h}\pi ^{2}+\pi ) \\ 
& > & \tfrac{1}{2}(1-\rho )\bar{h}^{2}\pi -(1-\rho )(2\bar{h}\pi ^{2}+\pi )\
=\ (1-\rho )(\tfrac{1}{2}\bar{h}^{2}\pi -2\bar{h}\pi ^{2}-\pi )%
\end{array}%
\end{equation*}%
Because $\rho \in (0,\frac{1}{2}]$, we have 
\begin{equation*}
\left( \tfrac{2\pi }{\pi +1}-\tfrac{\bar{h}\pi ^{2}+\pi }{\bar{h}\pi ^{2}+%
\bar{h}+2\pi }\right) ((1-\rho )\bar{h}^{2}\pi +(3-2\rho )\bar{h}\pi
^{2}+(1-\rho )\bar{h}+(3-2\rho )\pi )-(1-\rho )(2\bar{h}\pi ^{2}+\pi )>%
\tfrac{1}{4}\bar{h}^{2}\pi -\bar{h}\pi ^{2}-\tfrac{1}{2}\pi .
\end{equation*}%
We also have 
\begin{equation*}
\begin{array}{rcl}
&  & (1-\rho )(\bar{h}^{2}\pi +\bar{h}\pi ^{2}+\bar{h}+\pi )\left( \tfrac{%
2\pi }{\pi +1}-\tfrac{\bar{h}\pi ^{2}+\pi }{\bar{h}\pi ^{2}+\bar{h}+2\pi }%
\right) +\rho \bar{h}\pi ^{2}+\rho \pi \\ 
& < & (1-\rho )(\bar{h}^{2}\pi +\bar{h}\pi ^{2}+\bar{h}+\pi )+\tfrac{1}{2}(%
\bar{h}\pi ^{2}+\pi )\ <\ (\bar{h}^{2}\pi +\bar{h}\pi ^{2}+\bar{h}+\pi )+%
\tfrac{1}{2}(\bar{h}\pi ^{2}+\pi ) \\ 
& = & \bar{h}^{2}\pi +\tfrac{3}{2}\bar{h}\pi ^{2}+\bar{h}+\tfrac{3}{2}\pi .%
\end{array}%
\end{equation*}%
Putting these pieces together yields 
\begin{equation*}
\underaccent{\bar}{g}(\rho ,\bar{h},\pi )>\tfrac{\frac{1}{4}\bar{h}^{2}\pi -%
\bar{h}\pi ^{2}-\frac{1}{2}\pi }{\bar{h}^{2}\pi +\frac{3}{2}\bar{h}\pi ^{2}+%
\bar{h}+\frac{3}{2}\pi }=\tfrac{\frac{1}{4}-\frac{\pi }{\bar{h}}-\frac{1}{2%
\bar{h}^{2}}}{1+\frac{3\pi }{2\bar{h}}+\frac{1}{\bar{h}\pi }+\frac{3}{2\bar{h%
}^{2}}}.
\end{equation*}%
Because $\bar{h}>20\pi $ and $\pi >1$, we have 
\begin{equation*}
\tfrac{1}{\bar{h}}<\tfrac{1}{20\pi }<\tfrac{1}{20},\quad \tfrac{1}{2\bar{h}%
^{2}}<\tfrac{1}{800\pi ^{2}}<\tfrac{1}{800},\quad \tfrac{3\pi }{2\bar{h}}<%
\tfrac{3}{40},\quad \tfrac{1}{\bar{h}\pi }<\tfrac{1}{20\pi ^{2}}<\tfrac{1}{20%
},\quad \tfrac{3}{2\bar{h}^{2}}\leq \tfrac{3}{800\pi ^{2}}<\tfrac{3}{800}.
\end{equation*}%
This implies 
\begin{equation*}
\tfrac{1}{4}-\tfrac{\pi }{\bar{h}}-\tfrac{1}{2\bar{h}^{2}}>\tfrac{1}{4}-%
\tfrac{1}{20}-\tfrac{1}{800}=0.19875,
\end{equation*}%
and 
\begin{equation*}
1+\tfrac{3\pi }{2\bar{h}}+\tfrac{1}{\bar{h}\pi }+\tfrac{3}{2\bar{h}^{2}}<1+%
\tfrac{3}{40}+\tfrac{1}{20}+\tfrac{3}{800}=1.12875.
\end{equation*}%
Putting these pieces together yields 
\begin{equation}
\underaccent{\bar}{g}(\rho ,\bar{h},\pi )>\tfrac{0.19875}{1.12875}%
>0.17,\quad \textnormal{for all } \rho \in (0,\tfrac{1}{2}].
\label{boundary:lower_threshold_1}
\end{equation}%
Combining Eq.~\eqref{boundary:upper_threshold_1} and Eq.~%
\eqref{boundary:lower_threshold_1} yields the desired result in Eq.~%
\eqref{boundary:main_1}.

For the second boundary result, we have 
\begin{equation}
\inf_{h\in \lbrack \underaccent{\bar}{h},\bar{h}]}\left\{ \tfrac{h\pi
^{2}+\pi }{h\pi ^{2}+h+2\pi }\right\} =\tfrac{\underaccent{\bar}{h}\pi
^{2}+\pi }{\underaccent{\bar}{h}\pi ^{2}+\underaccent{\bar}{h}+2\pi },\quad
\sup_{h\in \lbrack \underaccent{\bar}{h},\bar{h}]}\left\{ \tfrac{2\pi }{\pi
+1}-\tfrac{h\pi ^{2}+\pi }{h\pi ^{2}+h+2\pi }\right\} =\tfrac{2\pi }{\pi +1}-%
\tfrac{\underaccent{\bar}{h}\pi ^{2}+\pi }{\underaccent{\bar}{h}\pi ^{2}+%
\underaccent{\bar}{h}+2\pi }.  \label{boundary:simple_threshold}
\end{equation}%
We rewrite 
\begin{equation*}
\bar{g}(\rho ,h,\pi )=\tfrac{h\pi ^{2}+\pi }{h\pi ^{2}+h+2\pi }-(1-\rho
)\left( \tfrac{\pi ^{3}(h^{2}-1)^{2}}{(h+\pi )(h\pi ^{2}+h+2\pi )(h\pi
^{2}+h+2\pi +(1-\rho )\pi (h^{2}-1))}\right) .
\end{equation*}%
Because $\rho \in (0,1)$ and $h,\pi >1$, we have $h\pi ^{2}+h+2\pi +(1-\rho
)\pi (h^{2}-1)>h\pi ^{2}+h+2\pi >0$. This implies 
\begin{equation*}
\bar{g}(\rho ,h,\pi )>\tfrac{h\pi ^{2}+\pi }{h\pi ^{2}+h+2\pi }-(1-\rho
)\left( \tfrac{\pi ^{3}(h^{2}-1)^{2}}{(h+\pi )(h\pi ^{2}+h+2\pi )^{2}}%
\right) .
\end{equation*}%
Because $\underaccent{\bar}{h},\bar{h}$ are finite, we have $M_{1}:=\max_{h\in
\lbrack \underaccent{\bar}{h},\bar{h}]}\left\{ \tfrac{\pi ^{3}(h^{2}-1)^{2}}{%
(h+\pi )(h\pi ^{2}+h+2\pi )^{2}}\right\} $ is finite. This implies 
\begin{equation*}
\bar{g}(\rho ,h,\pi )>\tfrac{h\pi ^{2}+\pi }{h\pi ^{2}+h+2\pi }-(1-\rho
)M_{1},\quad \textnormal{for all } h\in \lbrack \underaccent{\bar}{h},\bar{h}].
\end{equation*}%
Choosing $\delta _{1}:=\min \left\{ \frac{1}{2},\frac{\epsilon }{M_{1}}%
\right\} \in (0,\frac{1}{2})$. If $\rho \in (1-\delta _{1},1)$, we have 
\begin{equation*}
\bar{g}(\rho ,h,\pi )>\tfrac{h\pi ^{2}+\pi }{h\pi ^{2}+h+2\pi }-\epsilon
,\quad \textnormal{for all } h\in \lbrack \underaccent{\bar}{h},\bar{h}].
\end{equation*}%
Taking the infimum over the interval $[\underaccent{\bar}{h},\bar{h}]$ and
using Eq.~\eqref{boundary:simple_threshold} yields 
\begin{equation}
\inf_{h\in \lbrack \underaccent{\bar}{h},\bar{h}]}\bar{g}(\rho ,h,\pi )\geq 
\tfrac{\underaccent{\bar}{h}\pi ^{2}+\pi }{\underaccent{\bar}{h}\pi ^{2}+%
\underaccent{\bar}{h}+2\pi }-\epsilon \geq 0.  \label{boundary:upper_threshold_2}
\end{equation}%
By abuse of notation, we rewrite 
\begin{equation*}
\underaccent{\bar}{g}(\rho ,h,\pi )=\left( \tfrac{2\pi }{\pi +1}-\tfrac{h\pi
^{2}+\pi }{h\pi ^{2}+h+2\pi }\right) +(1-\rho )\left( \tfrac{P(h)-\left( 
\frac{2\pi }{\pi +1}-\frac{h\pi ^{2}+\pi }{h\pi ^{2}+h+2\pi }\right) Q(h)}{%
h\pi ^{2}+\pi +(1-\rho )Q(h)}\right) .
\end{equation*}%
where 
\begin{equation*}
\begin{array}{rcl}
P(h) & = & \left( \tfrac{2\pi }{\pi +1}-\tfrac{h\pi ^{2}+\pi }{h\pi
^{2}+h+2\pi }\right) (h^{2}\pi +2h\pi ^{2}+h+2\pi )-(2h\pi ^{2}+\pi ), \\ 
Q(h) & = & \left( \tfrac{2\pi }{\pi +1}-\tfrac{h\pi ^{2}+\pi }{h\pi
^{2}+h+2\pi }\right) (h^{2}\pi +h\pi ^{2}+h+\pi )-(h\pi ^{2}+\pi ).%
\end{array}%
\end{equation*}%
Because $\pi >1$ and $\underaccent{\bar}{h}>2\pi $, we have 
\begin{equation*}
\tfrac{1}{2}\leq \tfrac{2\pi }{\pi +1}-\tfrac{h\pi ^{2}+\pi }{h\pi
^{2}+h+2\pi }<1,\quad \textnormal{for all } h\in \lbrack \underaccent{\bar}{h},\bar{h}].
\end{equation*}%
This implies 
\begin{equation*}
Q(h)\geq \tfrac{1}{2}(h^{2}\pi -h\pi ^{2}+h-\pi )=\tfrac{1}{2}(h-\pi )(h\pi
+1)>0,\quad \textnormal{for all } h\in \lbrack \underaccent{\bar}{h},\bar{h}].
\end{equation*}%
Because $\rho \in (0,1)$, we have $h\pi ^{2}+\pi +(1-\rho )Q(h)>h\pi ^{2}+\pi
>\pi >1$ for all $h\in \lbrack \underaccent{\bar}{h},\bar{h}]$. Thus, we
have 
\begin{equation*}
\underaccent{\bar}{g}(\rho ,h,\pi )-\left( \tfrac{2\pi }{\pi +1}-\tfrac{h\pi
^{2}+\pi }{h\pi ^{2}+h+2\pi }\right) \leq (1-\rho )\left\vert \tfrac{%
P(h)-\left( \frac{2\pi }{\pi +1}-\frac{h\pi ^{2}+\pi }{h\pi ^{2}+h+2\pi }%
\right) Q(h)}{h\pi ^{2}+\pi +(1-\rho )Q(h)}\right\vert <(1-\rho )\left\vert
P(h)-\left( \tfrac{2\pi }{\pi +1}-\tfrac{h\pi ^{2}+\pi }{h\pi ^{2}+h+2\pi }%
\right) Q(h)\right\vert .
\end{equation*}%
Because $\underaccent{\bar}{h},\bar{h}$ are finite, we have $M_{2}:=\max_{h\in
\lbrack \underaccent{\bar}{h},\bar{h}]}\left\{ \left\vert P(h)-\left( \frac{%
2\pi }{\pi +1}-\frac{h\pi ^{2}+\pi }{h\pi ^{2}+h+2\pi }\right)
Q(h)\right\vert \right\} $ is finite. This implies 
\begin{equation*}
\underaccent{\bar}{g}(\rho ,h,\pi )<\left( \tfrac{2\pi }{\pi +1}-\tfrac{h\pi
^{2}+\pi }{h\pi ^{2}+h+2\pi }\right) +(1-\rho )M_{2},\quad \textnormal{for
all } h\in \lbrack \underaccent{\bar}{h},\bar{h}].
\end{equation*}%
Choosing $\delta _{2}:=\min \left\{ \frac{1}{2},\frac{\epsilon }{M_{2}}%
\right\} \in (0,\frac{1}{2})$. If $\rho \in (1-\delta _{2},1)$, we have 
\begin{equation*}
\underaccent{\bar}{g}(\rho ,h,\pi )<\left( \tfrac{2\pi }{\pi +1}-\tfrac{h\pi
^{2}+\pi }{h\pi ^{2}+h+2\pi }\right) +\epsilon ,\quad \text{for all }h\in
\lbrack \underaccent{\bar}{h},\bar{h}].
\end{equation*}%
Taking the supremum over the interval $[\underaccent{\bar}{h},\bar{h}]$ and
using Eq.~\eqref{boundary:simple_threshold} yields 
\begin{equation}
\sup_{h\in \lbrack \underaccent{\bar}{h},\bar{h}]}\underaccent{\bar}{g}(\rho
,h,\pi )\leq \left( \tfrac{2\pi }{\pi +1}-\tfrac{\underaccent{\bar}{h}\pi
^{2}+\pi }{\underaccent{\bar}{h}\pi ^{2}+\underaccent{\bar}{h}+2\pi }\right)
+\epsilon \leq 1. \label{boundary:lower_threshold_2}
\end{equation}%
Combining Eq.~\eqref{boundary:upper_threshold_2} and Eq.~%
\eqref{boundary:lower_threshold_2} and choosing $\delta :=\min \{\delta
_{1},\delta _{2}\}\in (0,\frac{1}{2})$ yields the desired result.

By definition of $\Lambda_\rho$ (see Eq.~\eqref{def:Lambda}), we have 
\begin{equation*}
\Lambda_\rho = [0, 1] \cap \left(\sup_{\beta_1, \beta_2 \in (0, 1), h \in [%
\underaccent{\bar}{h}, \bar{h}]} \underaccent{\bar}{\alpha}(\rho, \beta_1,
\beta_2, h, \pi), \inf_{\beta_1, \beta_2 \in (0, 1), h \in [%
\underaccent{\bar}{h}, \bar{h}]} \bar{\alpha}(\rho, \beta_1, \beta_2, h,
\pi)\right).
\end{equation*}
The monotonicity results guarantee that $\Lambda_{\rho_1} \subseteq
\Lambda_{\rho_2}$ if $0< \rho_1 \leq \rho_2 < 1$. The boundary results
guarantee that $\Lambda_\rho = \emptyset$ for all $\rho \in (0, \tfrac{1}{2}%
] $ and there exists $\delta \in (0, \frac{1}{2})$ such that $\Lambda_\rho
\neq \emptyset$ for all $\rho \in (1-\delta, 1)$. Putting these pieces
together yields that $\mu(\Lambda_\rho)$ as a function of $\rho$ on the
interval $(0, 1)$ is nondecreasing and satisfies that $\mu(\Lambda_\rho)=0$
for all $\rho \in (0, \tfrac{1}{2}]$ and $\mu(\Lambda_\rho)>0$ for all $\rho
\in (1-\delta, 1)$.

We define $\rho^\star = \sup\{\rho \in (0, 1): \mu(\Lambda_\rho)=0\}$. Then,
the previous results guarantee that $\frac{1}{2} \leq \rho^\star \leq
1-\delta$ and the following statement holds true,

\begin{enumerate}
\item If $\rho < \rho^\star$, then $\mu(\Lambda_\rho) = 0$.

\item If $\rho > \rho^\star$, then $\mu(\Lambda_\rho) > 0$.
\end{enumerate}

This completes the proof. $\qed$

\subsection{Proofs from Section~\protect\ref{sec:network_results}}

\noindent\emph{Proof of Proposition~\ref{prop:majority}.} We have 
\begin{equation*}
\Delta^\star \equiv \Delta_1(\rho, \alpha, \beta, \beta, h, \pi) -
\Delta_0(h, \pi) = |\bar{\Delta}_1(\rho, \alpha, \beta, h, \pi)| -
\Delta_0(h, \pi),
\end{equation*}
where 
\begin{equation*}
\bar{\Delta}_1(\rho, \alpha, \beta, h, \pi) = \tfrac{(1-\rho)(h\pi^2+\pi)+%
\alpha(\beta(1+\beta-\rho)h^2\pi+\beta h\pi^2+\beta h+\beta(1-\beta+\rho)\pi)%
}{(1-\rho)(h\pi^2+2\pi+h)+(\beta(1+\beta-\rho)h^2\pi+\beta h\pi^2+\beta
h+\beta(1-\beta+\rho)\pi)} - \tfrac{\pi}{\pi+1}.
\end{equation*}
Fixing $\rho \in (0, 1)$ and $\pi > 1$, we have 
\begin{equation}  \label{def:gap-h}
\tfrac{\partial\bar{\Delta}_1}{\partial h}(\rho, \alpha, \beta, h, \pi) = 
\tfrac{\pi(1-\rho)(\beta(\alpha(\pi^2+1)-\pi^2)h^2+2\beta\pi(2\alpha-1)h+%
\beta(\alpha(\pi^2+1)-\pi^2)+\pi^2-1)}{(1+\beta-\rho)(\beta h^2\pi + h\pi^2
+ h + (2-\beta)\pi)^2},
\end{equation}
and 
\begin{equation}  \label{def:gap-beta}
\tfrac{\partial\bar{\Delta}_1}{\partial \beta}(\rho, \alpha, \beta, h, \pi)
= -\tfrac{(1-\rho)(h\pi^2+\pi-\alpha(h\pi^2+2\pi+h))((1+2\beta-\rho)h^2\pi +
h\pi^2 + h + (1-2\beta+\rho)\pi)}{(1+\beta-\rho)^2(\beta
h^2\pi+h\pi^2+h+(2-\beta)\pi)^2},
\end{equation}
and 
\begin{equation}  \label{def:gap-alpha}
\tfrac{\partial\bar{\Delta}_1}{\partial \alpha}(\rho, \alpha, \beta, h, \pi)
= \tfrac{\beta(1+\beta-\rho)h^2\pi+\beta h\pi^2+\beta
h+\beta(1-\beta+\rho)\pi}{(1+\beta-\rho)(\beta h^2\pi+h\pi^2+h+(2-\beta)\pi)}%
.
\end{equation}
As a consequence, we have that $\frac{\partial\bar{\Delta}_1}{\partial \beta}%
(\rho, \alpha, \beta, h, \pi) < 0$ for any $\alpha < \frac{h\pi^2+\pi}{%
h\pi^2+2\pi+h}$ and $\frac{\partial\bar{\Delta}_1}{\partial \alpha}(\rho,
\alpha, \beta, h, \pi) > 0$.

Notice that if $\alpha > \frac{\pi^2}{\pi^2+1}$, then $\Delta^\star > 0$ and 
$\Delta_1$ is monotonically increasing in the homophily, $h$. Indeed, we
have $\alpha > \frac{\pi^2}{\pi^2+1} > \frac{h\pi^2+\pi}{h\pi^2+2\pi+h}$ 
for all $h>1$. This implies
\begin{equation*}
\bar{\Delta}_1(\rho, \alpha, \beta, h, \pi) > \min\left\{\alpha, \tfrac{%
h\pi^2+\pi}{h\pi^2+2\pi+h}\right\} - \tfrac{\pi}{\pi+1} = \tfrac{h\pi^2+\pi}{%
h\pi^2+2\pi+h} - \tfrac{\pi}{\pi+1} = \Delta_0(h, \pi) > 0.
\end{equation*}
Thus, we have that $\Delta_1(\rho, \alpha, \beta, h, \pi) = |\bar{%
\Delta}_1(\rho, \alpha, \beta, h, \pi)| = \bar{\Delta}_1(\rho, \alpha,
\beta, h, \pi)$ and $\Delta^\star = |\bar{\Delta}_1(\rho, \alpha, \beta, h,
\pi)| - \Delta_0(h, \pi) > 0$. In addition, we have $\alpha(\pi^2+1)-\pi^2>0$%
. Using Eq.~\eqref{def:gap-h}, we have 
\begin{equation*}
\tfrac{\partial\Delta_1}{\partial h}(\rho, \alpha, \beta, h, \pi) = \tfrac{%
\partial\bar{\Delta}_1}{\partial h}(\rho, \alpha, \beta, h, \pi) > 0.
\end{equation*}
This implies that $\Delta_1$ is monotonically increasing in the homophily, $%
h $. $\qed$

\bigskip\noindent\emph{Proof of Proposition~\ref{prop:minority}.} We show
that, if $\alpha < \frac{1}{2}$ and $\beta < \beta^\star$, then $\textnormal{%
sign}(\Delta^\star)$ is ambiguous and $\Delta_1$ is non-monotone in $h$. In
particular, there exist $1 < \underaccent{\bar}{h} < \bar{h} < \infty$ such
that

\begin{enumerate}
\item $\Delta^\star > 0$ and $\Delta_1$ is decreasing over $h \in (1, %
\underaccent{\bar}{h})$;

\item $\Delta^\star < 0$ and $\Delta_1$ is non-monotone over $h \in (%
\underaccent{\bar}{h}, \bar{h})$;

\item $\Delta^\star > 0$ and $\Delta_1$ is increasing over $h \in (\bar{h},
\infty)$.
\end{enumerate}

We introduce and prove two lemmas as follows,

\begin{lemma}
\label{Lemma:homophily-1} Fixing $\rho, \beta \in (0, 1)$, $\alpha \in (0, 
\frac{1}{2})$ and $\pi > 1$. For each $h>1$, let $\beta_1^\star(h) \in (0,1)$
denote the (unique) threshold such that $-\bar{\Delta}_1(\rho,\alpha,%
\beta,h,\pi)-\Delta_0(h,\pi) \geq 0$ if and only if $\beta \in
[\beta_1^\star(h),1]$. Then, we define 
\begin{equation*}
\beta_1^\star := \sup_{h>1} \beta_1^\star(h) \in (0,1].
\end{equation*}
For any $\beta < \beta_1^\star$, there exists $1 < \underaccent{\bar}{h}_1 < 
\bar{h}_1 < \infty$ such that 
\begin{equation*}
-\bar{\Delta}_1(\rho, \alpha, \beta, h, \pi)-\Delta_0(h, \pi) \left\{%
\begin{array}{ll}
\geq 0, & \textnormal{if } 1 \leq h \leq \underaccent{\bar}{h}_1 \textnormal{%
\ or } h \geq \bar{h}_1, \\ 
< 0, & \textnormal{otherwise}.%
\end{array}%
\right.
\end{equation*}
\end{lemma}

\noindent\emph{Proof}. We have 
\begin{equation*}
-\bar{\Delta}_1(\rho, \alpha, \beta, h, \pi)-\Delta_0(h, \pi) = \tfrac{P(h)}{%
(1+\beta-\rho)(\pi+1)(h\pi^2 + 2\pi + h)(\beta h^2\pi+h\pi^2+h+(2-\beta)\pi)}%
,
\end{equation*}
where $P(h)=A_3(\rho,\alpha,\beta,\pi)h^3+A_2(\rho,\alpha,\beta,\pi)h^2+A_1(%
\rho,\alpha,\beta,\pi)h+A_0(\rho,\alpha,\beta,\pi)$ and 
\begin{equation*}
\begin{array}{rcl}
A_3(\rho,\alpha,\beta,\pi) & = & -\beta\pi(1+\beta-\rho)(\alpha(\pi^3+\pi^2+%
\pi+1)-(\pi^3-\pi^2+2\pi)), \\ 
A_2(\rho,\alpha,\beta,\pi) & = & \beta(1+\beta-\rho)(3\pi^3-\pi^2)+\beta(%
\pi^3+\pi)(\pi^2-\pi+2)-2(1-\rho)(\pi^3+\pi)(\pi-1) \\ 
&  & -\alpha\beta(\pi+1)(\pi^4+(4-2\rho)\pi^2+2\beta\pi^2+1), \\ 
A_1(\rho,\alpha,\beta,\pi) & = & \alpha\beta^2(\pi^4+\pi^3+\pi^2+\pi)-%
\beta^2(\pi^4-\pi^3+2\pi)+2(1-\rho)(\pi-1)^3 \\ 
&  & +\beta(\pi^3+\pi)(\pi(\rho+4)-(\rho+2)-\alpha(\pi+1)(\rho+3))+\beta(%
\rho+1)(\pi^2+\pi), \\ 
A_0(\rho,\alpha,\beta,\pi) & = & \pi^2(\beta^2((2\alpha-3)\pi+(2%
\alpha+1))+(1+\rho)\beta((3-2\alpha)\pi-(2\alpha+1))+4(1-\rho)(\pi-1)).%
\end{array}%
\end{equation*}
Clearly, the sign of $-\bar{\Delta}_1(\rho, \alpha, \beta, h,
\pi)-\Delta_0(h, \pi)$ is the same as the sign of $P(h)$.

Because $\alpha < \frac{1}{2}$, we have $A_3(\rho,\alpha,\beta,\pi) > 0$.
Indeed, we have $\alpha(\pi^3+\pi^2+\pi+1)-(\pi^3-\pi^2+2\pi) < 0$. This
implies that $\lim_{h \to +\infty} P(h) = +\infty$ and hence $P(h) > 0$ for
all sufficiently large $h$. In addition, because $\alpha < \frac{\pi}{\pi+1}$,
we have 
\begin{equation*}
-\bar{\Delta}_1(\rho, \alpha, \beta, 1, \pi)-\Delta_0(1, \pi) = \tfrac{\pi}{%
\pi+1} - \tfrac{(1-\rho)(\pi^2+\pi)+\alpha(\beta(1+\beta-\rho)\pi+\beta
\pi^2+\beta +\beta(1-\beta+\rho)\pi)}{(1-\rho)(\pi^2+2\pi+1)+(\beta(1+\beta-%
\rho)\pi+\beta \pi^2+\beta+\beta(1-\beta+\rho)\pi)} > \tfrac{\pi}{\pi+1} -
\max\left\{\alpha, \tfrac{\pi}{\pi+1}\right\} = 0,
\end{equation*}
which implies $P(1) > 0$. By definition, the function $P(\cdot)$ is a cubic
polynomial with a strictly positive leading coefficient. Suppose that $P(h)
< 0$ for some $h > 1$. Then, the continuity of $P(\cdot)$ guarantees that
there exists $1 < \underaccent{\bar}{h}_1 < \bar{h}_1 < \infty$ such that 
\begin{equation*}
P(h) \left\{%
\begin{array}{ll}
\geq 0, & \textnormal{if } 1 \leq h \leq \underaccent{\bar}{h}_1 \textnormal{%
\ or } h \geq \bar{h}_1, \\ 
< 0, & \textnormal{otherwise}.%
\end{array}%
\right.
\end{equation*}
This together with the fact that the sign of $-\bar{\Delta}_1(\rho, \alpha,
\beta, h, \pi)-\Delta_0(h, \pi)$ is the same as the sign of $P(h)$ yields
the desired result.

In what follows, we show that $P(h) < 0$ for some $h > 1$ whenever $\alpha < 
\frac{1}{2}$ and $\beta < \beta_1^\star$. Indeed, we show why $\beta_1^\star$
exists and is unique. Because $\alpha < \frac{h\pi^2+\pi}{h\pi^2+2\pi+h}$, we
have $\frac{\partial\bar{\Delta}_1}{\partial \beta}(\rho, \alpha, \beta, h,
\pi) < 0$, implying that $-\bar{\Delta}_1(\rho, \alpha, \beta, h,
\pi)-\Delta_0(h, \pi)$ as a function of $\beta$ is increasing over $[0, 1]$.
Fixing any $h > 1$, we have 
\begin{equation*}
-\bar{\Delta}_1(\rho, \alpha, 0, h, \pi)-\Delta_0(h, \pi) = \tfrac{2\pi}{%
\pi+1} - \tfrac{2(h\pi^2+\pi)}{h\pi^2+2\pi+h} < 0.
\end{equation*}
In addition, we have 
\begin{equation*}
-\bar{\Delta}_1(\rho, \alpha, 1, h, \pi)-\Delta_0(h, \pi) = \tfrac{2\pi}{%
\pi+1} - \tfrac{h\pi^2+\pi}{h\pi^2+2\pi+h} - \tfrac{(1-\rho)(h\pi^2+\pi)+%
\alpha((2-\rho)h^2\pi+h\pi^2+h+\rho\pi)}{(1-\rho)(h\pi^2+2\pi+h)+((2-%
\rho)h^2\pi+h\pi^2+h+\rho\pi)}.
\end{equation*}
This implies that $-\bar{\Delta}_1(\rho, \alpha, 1, h, \pi)-\Delta_0(h, \pi)$
as a function of $\alpha$ is strictly decreasing over $[0, 1]$. Because $%
\alpha < \frac{1}{2}$, we have 
\begin{equation*}
-\bar{\Delta}_1(\rho, \alpha, 1, h, \pi)-\Delta_0(h, \pi) > -\bar{\Delta}%
_1(\rho, \tfrac{1}{2}, 1, h, \pi)-\Delta_0(h, \pi) = \tfrac{(\pi-1)R(\rho)}{%
2(2-\rho)(\pi+1)(h\pi^2+2\pi+h)(h^2\pi+h\pi^2+h+\pi)}
\end{equation*}
where $R(\rho)=R_0(h, \pi)+R_1(h, \pi)\rho$ and 
\begin{equation*}
\begin{array}{rcl}
R_1(h, \pi) & = & -h^3\pi^3+2h^3\pi^2-h^3\pi+4h^2\pi^3-4h^2\pi^2+4h^2\pi-3h%
\pi^3+6h\pi^2-3h\pi-4\pi^2, \\ 
R_0(h, \pi) & = & 2h^3\pi^3-4h^3\pi^2+2h^3\pi+h^2\pi^4-6h^2\pi^3+10h^2%
\pi^2-6h^2\pi+h^2+8h\pi^3-8h\pi^2+8h\pi+8\pi^2.%
\end{array}%
\end{equation*}
Then, we have 
\begin{equation*}
\begin{array}{rcl}
R(1) & = & h^3\pi^3-2h^3\pi^2+h^3\pi+h^2\pi^4-2h^2\pi^3+6h^2\pi^2-2h^2%
\pi+h^2+5h\pi^3-2h\pi^2+5h\pi+4\pi^2 \\ 
& = & (h+\pi)(h\pi+1)(h(\pi-1)^2+4\pi) \ > \ 0.%
\end{array}%
\end{equation*}
Proceeding to $R(0) = R_0(h, \pi)$. For simplicity, we let $x=\pi-1>0$ and $%
y=h-1>0$. Then, we have 
\begin{equation*}
R_0(h,\pi) = x^4
y^2+2x^3y^3+2x^4y+4x^3y^2+2x^2y^3+x^4+10x^3y+4x^2y^2+8x^3+18x^2y+24x^2+16xy+32x+8y+16 > 0.
\end{equation*}
Because $R(\rho)$ is linear in $\rho$ and $R(0), R(1) > 0$, we have that $%
R(\rho)>0$ for $\forall \rho \in (0,1)$. This implies that $-\bar{\Delta}%
_1(\rho, \alpha, 1, h, \pi)-\Delta_0(h, \pi) > 0$. Putting these pieces
together yields that there exists $\beta_1^\star(h) \in (0, 1)$ such that $-%
\bar{\Delta}_1(\rho, \alpha, \beta, h, \pi)-\Delta_0(h, \pi) \geq 0$ if and
only if $\beta \in [\beta_1^\star(h),1]$. For each fixed $h>1$, the
preceding argument implies that there exists a unique $\beta_1^\star(h) \in
(0,1)$ such that $-\bar{\Delta}_1(\rho,\alpha,\beta,h,\pi)-\Delta_0(h,\pi)
\geq 0$ if and only if $\beta \in [\beta_1^\star(h),1]$. We define $%
\beta_1^\star := \sup_{h>1} \beta_1^\star(h) \in (0,1]$. In what follows, we
show that $P(h)<0$ for some $h>1$ whenever $\alpha<\frac{1}{2}$ and $%
\beta<\beta_1^\star$. Indeed, if $\beta<\beta_1^\star$, then by the
definition of supremum there exists some $h_\beta>1$ such that $%
\beta<\beta_1^\star(h_\beta)$. This implies $-\bar{\Delta}%
_1(\rho,\alpha,\beta,h_\beta,\pi)-\Delta_0(h_\beta,\pi)<0$. Thus, we have $%
P(h_\beta)<0$, as desired. $\qed$

\begin{lemma}
\label{Lemma:homophily-2} Fixing $\rho, \beta \in (0, 1)$, $\alpha \in (0, 
\frac{1}{2})$ and $\pi > 1$. For each $h>1$, let $\beta_2^\star(h) \in (0,1)$
denote the (unique) threshold such that $-\bar{\Delta}_1(\rho,\alpha,%
\beta,h,\pi) \leq 0$ if and only if $\beta \in [\beta_2^\star(h),1]$. Then,
we define 
\begin{equation*}
\beta_2^\star := \sup_{h>1} \beta_2^\star(h) \in (0,1].
\end{equation*}
For any $\beta < \min\{\beta_2^\star, \frac{\pi-1}{2\pi - 2\alpha(\pi+1)}\}$%
, there exists $h_0 > 1$ and $1 < \underaccent{\bar}{h}_2 < \bar{h}_2 <
\infty$ such that 
\begin{equation*}
\tfrac{\partial\bar{\Delta}_1}{\partial h}(\rho, \alpha, \beta, h, \pi)
\left\{%
\begin{array}{ll}
\geq 0, & \textnormal{if } 1 \leq h \leq h_0, \\ 
< 0, & \textnormal{otherwise}.%
\end{array}%
\right.
\end{equation*}
and 
\begin{equation*}
\bar{\Delta}_1(\rho, \alpha, \beta, h, \pi) \left\{%
\begin{array}{ll}
\leq 0, & \textnormal{if } 1 \leq h \leq \underaccent{\bar}{h}_2 \textnormal{%
\ or } h \geq \bar{h}_2, \\ 
> 0, & \textnormal{otherwise}.%
\end{array}%
\right.
\end{equation*}
\end{lemma}

\noindent\emph{Proof}. We have 
\begin{equation*}
\tfrac{\partial\bar{\Delta}_1}{\partial h}(\rho, \alpha, \beta, h, \pi) = 
\tfrac{\pi(1-\rho)Q(h)}{(1+\beta-\rho)(\beta h^2\pi + h\pi^2 + h +
(2-\beta)\pi)^2},
\end{equation*}
where $Q(h)=\beta(\alpha(\pi^2+1)-\pi^2)h^2+2\beta\pi(2\alpha-1)h+\beta(%
\alpha(\pi^2+1)-\pi^2)+\pi^2-1$. Because $\alpha < \frac{1}{2}$, we have $%
\alpha(1+\pi^2)-\pi^2 < 0$. This implies that $\lim_{h \rightarrow +\infty}
Q(h) = -\infty$. Because $\beta < \frac{\pi-1}{2\pi - 2\alpha(\pi+1)}$, we
have 
\begin{equation*}
Q(1) = (\pi+1)\left(2\beta\left(\alpha(\pi+1)-\pi\right)+\pi-1\right)> 0.
\end{equation*}
Note that the function $Q(\cdot)$ is a quadratic polynomial with a strictly
negative leading coefficient. Thus, we have that there exists $h_0 > 1$ such
that 
\begin{equation*}
Q(h) \left\{%
\begin{array}{ll}
\geq 0, & \textnormal{if } 1 \leq h \leq h_0, \\ 
< 0, & \textnormal{otherwise}.%
\end{array}%
\right.
\end{equation*}
This together with the fact that the sign of $\frac{\partial\bar{\Delta}_1}{%
\partial h}(\rho, \alpha, \beta, h, \pi)$ is the same as the sign of $Q(h)$
yields the desired result.

Because $\alpha < \frac{1}{2}$, we have $\bar{\Delta}_1(\alpha, \beta, 1, \pi)
< 0$. As proved before, there exists $h_0 > 1$ such that 
\begin{equation*}
\tfrac{\partial\bar{\Delta}_1}{\partial h}(\rho, \alpha, \beta, h, \pi)
\left\{%
\begin{array}{ll}
\geq 0, & \textnormal{if } 1 \leq h \leq h_0, \\ 
< 0, & \textnormal{otherwise}.%
\end{array}%
\right.
\end{equation*}
It suffices to show that $\bar{\Delta}_1(\rho, \alpha, \beta, h, \pi) > 0$
for some $h > 1$ whenever $\alpha < \frac{1}{2}$ and $\beta < \beta_2^\star$%
. Indeed, we show why $\beta_2^\star$ exists and is unique. Because $\alpha < 
\frac{h\pi^2+\pi}{h\pi^2+2\pi+h}$, we have $\frac{\partial\bar{\Delta}_1}{%
\partial \beta}(\rho, \alpha, \beta, h, \pi) < 0$, implying that $\bar{\Delta%
}_1(\rho, \alpha, \beta, h, \pi)$ as a function of $\beta$ is decreasing
over $[0, 1]$. Fixing any $h > 1$, we have 
\begin{equation*}
\bar{\Delta}_1(\rho, \alpha, 0, h, \pi) = \tfrac{h\pi^2+\pi}{h\pi^2+2\pi+h}
- \tfrac{\pi}{\pi+1} > 0.
\end{equation*}
In addition, we have 
\begin{equation*}
\bar{\Delta}_1(\rho, \alpha, 1, h, \pi) = \tfrac{(1-\rho)(h\pi^2+\pi)+%
\alpha((2-\rho)h^2\pi+h\pi^2+h+\rho\pi)}{(1-\rho)(h\pi^2+2\pi+h)+((2-%
\rho)h^2\pi+h\pi^2+h+\rho\pi)} - \tfrac{\pi}{\pi+1}.
\end{equation*}
This implies that $\bar{\Delta}_1(\rho, \alpha, 1, h, \pi)$ as a function of 
$\alpha$ is strictly increasing over $[0, 1]$. Because $\alpha < \frac{1}{2}$,
we have 
\begin{equation*}
\bar{\Delta}_1(\rho, \alpha, 1, h, \pi) < \bar{\Delta}_1(\rho, \tfrac{1}{2},
1, h, \pi) = \tfrac{(\pi-1)V(\rho)}{2(2-\rho)(\pi+1)(h+\pi)(h\pi+1)},
\end{equation*}
where $V(\rho)=V_0(h, \pi)+V_1(h, \pi)\rho$ and 
\begin{equation*}
V_1(h, \pi) = \pi(h-1)^2, \quad V_0(h, \pi) = -(h\pi^2+h+2\pi+2h\pi(h-1)).
\end{equation*}
Then, we have 
\begin{equation*}
V(1) = - (h\pi^2+h+\pi+h^2\pi) < 0, \quad V(0) = -(h\pi^2+h+2\pi+2h\pi(h-1))
< 0.
\end{equation*}
Because $V(\rho)$ is linear in $\rho$ and $V(0), V(1) < 0$, we have that $%
V(\rho)<0$ for all $\rho \in (0,1)$. This implies that $\bar{\Delta}_1(\rho,
\alpha, 1, h, \pi) < 0$. Putting these pieces together yields that there
exists $\beta_2^\star(h) \in (0, 1)$ such that $\bar{\Delta}_1(\rho, \alpha,
\beta, h, \pi) \leq 0$ if and only if $\beta \in [\beta_2^\star(h),1]$. For
each fixed $h>1$, the preceding argument implies that there exists a unique $%
\beta_2^\star(h) \in (0,1)$ such that $\bar{\Delta}_1(\rho,\alpha,\beta,h,%
\pi) \leq 0$ if and only if $\beta \in [\beta_2^\star(h),1]$. We define $%
\beta_2^\star := \sup_{h>1} \beta_2^\star(h) \in (0,1]$. In what follows, we
show that $\bar{\Delta}_1(\rho,\alpha,\beta,h,\pi)>0$ for some $h>1$
whenever $\alpha<\frac{1}{2}$ and $\beta<\beta_2^\star$. Indeed, if $%
\beta<\beta_2^\star$, then by the definition of supremum there exists some $%
h_\beta>1$ such that $\beta<\beta_2^\star(h_\beta)$. This implies $\bar{%
\Delta}_1(\rho,\alpha,\beta,h_\beta,\pi)>0$, as desired. $\qed$

Back to the original claim of Proposition~\ref{prop:minority}. We set $%
\beta^\star = \min\{\beta_1^\star, \beta_2^\star, \frac{\pi-1}{2\pi -
2\alpha(\pi+1)}\} \in (0, 1)$. By Lemma~\ref{Lemma:homophily-1}, we have
that there exists $1 < \underaccent{\bar}{h}_1 < \bar{h}_1 < \infty$ such
that 
\begin{equation*}
-\bar{\Delta}_1(\rho, \alpha, \beta, h, \pi)-\Delta_0(h, \pi) \left\{%
\begin{array}{ll}
\geq 0, & \textnormal{if } 1 \leq h \leq \underaccent{\bar}{h}_1 \textnormal{%
\ or } h \geq \bar{h}_1, \\ 
< 0, & \textnormal{otherwise}.%
\end{array}%
\right.
\end{equation*}
If $1 \leq h \leq \underaccent{\bar}{h}_1$ or $h \geq \bar{h}_1$, we have
that $\Delta_1(\rho, \alpha, \beta, h, \pi) = -\bar{\Delta}_1(\rho, \alpha,
\beta, h, \pi) \geq \Delta_0(h, \pi)$ because $\Delta_0(h, \pi) \geq 0$.
Otherwise, we consider: $\bar{\Delta}_1(\rho, \alpha, \beta, h, \pi) \geq 0$
or $\bar{\Delta}_1(\rho, \alpha, \beta, h, \pi) < 0$. For the former case,
we have 
\begin{equation*}
\Delta_1(\rho, \alpha, \beta, h, \pi) - \Delta_0(h, \pi) = \bar{\Delta}%
_1(\rho, \alpha, \beta, h, \pi) - \Delta_0(h, \pi) < \max\left\{\alpha, 
\tfrac{h\pi^2+\pi}{h\pi^2+2\pi+h}\right\}-\tfrac{h\pi^2+\pi}{h\pi^2+2\pi+h}
= 0.
\end{equation*}
For the latter case, we have 
\begin{equation*}
\Delta_1(\rho, \alpha, \beta, h, \pi) - \Delta_0(h, \pi) = -\bar{\Delta}%
_1(\rho, \alpha, \beta, h, \pi) - \Delta_0(h, \pi) < 0.
\end{equation*}
Putting these pieces together yields 
\begin{equation}  \label{result:homophily-gap}
\Delta^\star \equiv \Delta_1(\rho, \alpha, \beta, h, \pi)-\Delta_0(h, \pi)
\left\{%
\begin{array}{ll}
> 0, & \textnormal{if } 1 < h < \underaccent{\bar}{h}_1 \textnormal{\ or } h
> \bar{h}_1, \\ 
< 0, & \textnormal{if } \underaccent{\bar}{h}_1 < h < \bar{h}_1.%
\end{array}%
\right.
\end{equation}
By Lemma~\ref{Lemma:homophily-2}, we have that there exists $h_0 > 1$ and $1
< \underaccent{\bar}{h}_2 < \bar{h}_2 < \infty$ such that 
\begin{equation*}
\tfrac{\partial\bar{\Delta}_1}{\partial h}(\rho, \alpha, \beta, h, \pi)
\left\{%
\begin{array}{ll}
> 0, & \textnormal{if } 1 < h < h_0, \\ 
< 0, & \textnormal{if } h > h_0.%
\end{array}%
\right.
\end{equation*}
and 
\begin{equation*}
\bar{\Delta}_1(\rho, \alpha, \beta, h, \pi) \left\{%
\begin{array}{ll}
< 0, & \textnormal{if } 1 < h < \underaccent{\bar}{h}_2 \textnormal{\ or } h
> \bar{h}_2, \\ 
> 0, & \textnormal{if } \underaccent{\bar}{h}_2 < h < \bar{h}_2.%
\end{array}%
\right.
\end{equation*}
Because $\Delta_1(\rho, \alpha, \beta, h, \pi)=|\bar{\Delta}_1(\rho, \alpha,
\beta, h, \pi)|$, we have

\begin{enumerate}
\item $\Delta_1$ is decreasing if $1 < h < \min\{h_0, \underaccent{\bar}{h}%
_2\}$;

\item $\Delta_1$ is non-monotone if $\min\{h_0, \underaccent{\bar}{h}_2\} <
h < \max\{h_0, \bar{h}_2\}$;

\item $\Delta_1$ is increasing if $h > \max\{h_0, \bar{h}_2\}$.
\end{enumerate}

In addition, we have $\bar{\Delta}_1(\rho, \alpha, \beta, h, \pi) <
-\Delta_0(h, \pi) < 0$ if $1 \leq h \leq \underaccent{\bar}{h}_1$ or $h \geq 
\bar{h}_1$. This implies that $\underaccent{\bar}{h}_1 \leq %
\underaccent{\bar}{h}_2$ and $\bar{h}_1 \geq \bar{h}_2$. Putting these
pieces together with Eq.~\eqref{result:homophily-gap} yields the desired
result with $\underaccent{\bar}{h} = \min\{h_0, \underaccent{\bar}{h}_1\}$
and $\bar{h} = \max\{h_0, \bar{h}_1\}$. This completes the proof. $\qed$

\subsection{Proofs from Section~\protect\ref{sec:local_aggregators}}

\noindent\emph{Proof of Proposition~\ref{prop:local_effects}.} It suffices
to show that 
\begin{eqnarray*}
\vert p_1^{\star\star}(\rho, \beta_{11}, \beta_{12}, h, \pi)-1\vert & < &
\vert \tfrac{h\pi^2+\pi}{h\pi^2+h+2\pi}-1\vert \ = \ \tfrac{h+\pi}{%
h\pi^2+h+2\pi}, \\
\vert p_2^{\star\star}(\rho, \beta_{21}, \beta_{22}, h, \pi)-1\vert & < &
\vert \tfrac{h+\pi}{h\pi^2+h+2\pi}-1\vert \ = \ \tfrac{h\pi^2+\pi}{%
h\pi^2+h+2\pi}.
\end{eqnarray*}
Using the definition of $p_1^{\star\star}(\rho, \beta_{11}, \beta_{12}, h,
\pi)$, we have 
\begin{equation*}
\vert p_1^{\star\star}(\rho, \beta_{11}, \beta_{12}, h, \pi) - 1\vert = 
\tfrac{(1-\rho)(1-\beta_{11}+\beta_{12})h + (1-\rho)\pi}{(1-\rho+\beta_{11})%
\beta_{12} h^2\pi + (1-\rho+\beta_{11})h\pi^2 +
(\beta_{12}+(1-\rho)(1-\beta_{11}+\beta_{12}))h +
(2(1-\rho)+\rho\beta_{11}+\beta_{12}-\beta_{11}\beta_{12})\pi}.
\end{equation*}
Because 
\begin{equation*}
(1-\rho+\beta_{11})\beta_{12} h^2\pi \geq 0, \quad 1-\rho+\beta_{11} \geq 1-\rho,
\quad \beta_{12} \geq 0, \quad \rho\beta_{11}+\beta_{12}-\beta_{11}\beta_{12} \geq 0,
\end{equation*}
we have 
\begin{equation*}
\vert p_1^{\star\star}(\rho, \beta_{11}, \beta_{12}, h, \pi) - 1\vert \leq 
\tfrac{(1-\rho)(1-\beta_{11}+\beta_{12})h + (1-\rho)\pi}{(1-\rho)h\pi^2 +
(1-\rho)(1-\beta_{11}+\beta_{12})h + 2(1-\rho)\pi}.
\end{equation*}
Because $1 - \rho > 0$, we have 
\begin{equation*}
\vert p_1^{\star\star}(\rho, \beta_{11}, \beta_{12}, h, \pi) - 1\vert \leq 
\tfrac{(1-\beta_{11}+\beta_{12})h + \pi}{h\pi^2 + (1-\beta_{11}+\beta_{12})h
+ 2\pi}.
\end{equation*}
Then, we have 
\begin{equation*}
\tfrac{(1-\beta_{11}+\beta_{12})h + \pi}{h\pi^2 + (1-\beta_{11}+\beta_{12})h
+ 2\pi} - \tfrac{h + \pi}{h\pi^2 + h + 2\pi} = - \tfrac{(\beta_{11}-%
\beta_{12})h\pi(h\pi+1)}{(h\pi^2 + (1-\beta_{11}+\beta_{12})h + 2\pi)(h\pi^2
+ h + 2\pi)} \overset{\beta_{11} > \beta_{12}}{<} 0.
\end{equation*}
Putting these pieces together yields 
\begin{equation}  \label{inequality:main-1}
\vert p_1^{\star\star}(\rho, \beta_{11}, \beta_{12}, h, \pi) - 1\vert < 
\tfrac{h + \pi}{h\pi^2 + h + 2\pi}.
\end{equation}
Using the definition of $p_2^{\star\star}(\rho, \beta_{21}, \beta_{22}, h,
\pi)$, we have 
\begin{equation*}
\vert p_2^{\star\star}(\rho, \beta_{21}, \beta_{22}, h, \pi)-1 \vert = 
\tfrac{(1-\rho)(1+\beta_{21}-\beta_{22})h\pi^2 + (1-\rho)\pi}{%
(1-\rho+\beta_{22})\beta_{21} h^2\pi +
(\beta_{21}+(1-\rho)(1+\beta_{21}-\beta_{22}))h\pi^2 + (1-\rho+\beta_{22})h
+ (2(1-\rho)+\beta_{21}+\rho\beta_{22}-\beta_{21}\beta_{22})\pi}.
\end{equation*}
Because 
\begin{equation*}
(1-\rho+\beta_{22})\beta_{21} h^2\pi \geq 0, \quad \beta_{21} \geq 0, \quad
1-\rho+\beta_{22} \geq 1-\rho, \quad \beta_{21}+\rho\beta_{22}-\beta_{21}\beta_{22} \geq 0,
\end{equation*}
we have 
\begin{equation*}
\vert p_2^{\star\star}(\rho, \beta_{21}, \beta_{22}, h, \pi) - 1\vert \leq \tfrac{%
(1-\rho)(1+\beta_{21}-\beta_{22})h\pi^2 + (1-\rho)\pi}{(1-\rho)(1+%
\beta_{21}-\beta_{22})h\pi^2 + (1-\rho)h + 2(1-\rho)\pi}.
\end{equation*}
Because $1 - \rho > 0$, we have 
\begin{equation*}
\vert p_2^{\star\star}(\rho, \beta_{21}, \beta_{22}, h, \pi) - 1\vert \leq \tfrac{%
(1+\beta_{21}-\beta_{22})h\pi^2 + \pi}{(1+\beta_{21}-\beta_{22})h\pi^2 + h +
2\pi}.
\end{equation*}
Then, we have 
\begin{equation*}
\tfrac{(1+\beta_{21}-\beta_{22})h\pi^2 + \pi}{(1+\beta_{21}-\beta_{22})h%
\pi^2 + h + 2\pi} - \tfrac{h\pi^2 + \pi}{h\pi^2 + h + 2\pi} = \tfrac{%
(\beta_{21}-\beta_{22})h\pi^2(h+\pi)}{((1+\beta_{21}-\beta_{22})h\pi^2 + h +
2\pi)(h\pi^2 + h + 2\pi)} \overset{\beta_{21} < \beta_{22}}{<} 0.
\end{equation*}
Putting these pieces together yields 
\begin{equation}  \label{inequality:main-2}
\vert p_2^{\star\star}(\rho, \beta_{21}, \beta_{22}, h, \pi) - 1\vert < \tfrac{%
h\pi^2 + \pi}{h\pi^2 + h + 2\pi}.
\end{equation}
This completes the proof. $\qed$

\bigskip\noindent\emph{Proof of Theorem~\ref{thm:impossibility}.} 
From Proposition~\ref{prop:local_effects}, we have 
\begin{equation*}
\vert p_1^{\star\star}(\rho, \beta_{11}, \beta_{12}, h, \pi) - 1\vert < 
\tfrac{h+\pi}{h\pi^2+h+2\pi}, \quad \vert p_2^{\star\star}(\rho, \beta_{21},
\beta_{22}, h, \pi) - 1\vert < \tfrac{h\pi^2+\pi}{h\pi^2+h+2\pi}.
\end{equation*}
Thus, we have $\vert p_1^{\star\star}(\rho, \beta_{11}, \beta_{12}, h, \pi)
- 1\vert + \vert p_2^{\star\star}(\rho, \beta_{21}, \beta_{22}, h, \pi) -1\vert
< 1$. Because $p_1^{\star\star}(\rho, \alpha, \beta_1, \beta_2, h, \pi) \in
(0, 1)$, we have 
\begin{equation*}
\vert p_1^{\star\star}(\rho, \alpha, \beta_1, \beta_2, h, \pi) - 1\vert +
\vert p_1^{\star\star}(\rho, \alpha, \beta_1, \beta_2, h, \pi) \vert = 1.
\end{equation*}
Suppose, toward a contradiction, that $(\Delta_1)_k \leq (\Delta_2)_k$ %
for all $k \in \{1,2\}$. Then, we have
\begin{equation*}
(\Delta_1)_1+(\Delta_1)_2 \le (\Delta_2)_1+(\Delta_2)_2 = %
\vert p_1^{\star\star}(\rho, \beta_{11}, \beta_{12}, h, \pi) - 1\vert + %
\vert p_2^{\star\star}(\rho, \beta_{21}, \beta_{22}, h, \pi) -1\vert < 1. 
\end{equation*}
However, we have
\begin{equation*}
(\Delta_1)_1+(\Delta_1)_2 = \vert p_1^{\star\star}(\rho, \alpha, \beta_1, \beta_2, h, \pi) - 1\vert +
\vert p_1^{\star\star}(\rho, \alpha, \beta_1, \beta_2, h, \pi) \vert = 1.
\end{equation*}
This yields the contradiction. Thus, there exists at least one topic 
$k^\star \in \{1,2\}$ such that $(\Delta_1)_{k^\star}>(\Delta_2)_{k^\star}$.
This completes the proof. $\qed$ \newpage %

\section{Fixed Two-Island Environment}

\label{app:fixed_network}

In this appendix subsection, we study a different question from that in 
Theorem~\ref{thm:rho}. We remain within the stylized two-island environment 
analyzed in the main text but treat its parameters $(h,\pi,\beta_1,\beta_2)$ 
as fixed and known. The training weights can therefore be calibrated to this 
particular environment. The objective is to characterize when a global aggregator 
improves learning pointwise in this fixed two-island environment, rather than 
whether a single training design is robustly beneficial across a range of 
admissible environments. 

\begin{proposition}
\label{prop:known} 
Fix a two-island environment with parameters $(h,\pi,\beta_1,\beta_2)$. Then 
there exist $\underaccent{\bar}{\alpha}(\rho)<\overline{\alpha}(\rho) \in (0,1)$ 
such that:
\begin{equation*}
\Delta^\star(\rho,\alpha,\beta_1,\beta_2,h,\pi) \begin{cases}
\leq 0 & \textnormal{if } \alpha \in \left[\max\{0, \underaccent{\bar}{\alpha}(\rho)\}, \bar{\alpha}(\rho)\right], \\ 
> 0 & \textnormal{if } \alpha \in [0, \max\{0, \underaccent{\bar}{\alpha}(\rho)\}) \cup (\bar{\alpha}(\rho), 1]. 
\end{cases}
\end{equation*}
\end{proposition}

\bigskip\noindent\emph{Proof of Proposition~\ref{prop:known}.} As in the
proof of Theorem~\ref{thm:rho}, we have 
\begin{equation*}
\Delta_1(\rho, \alpha, \beta_1, \beta_2, h, \pi) - \Delta_0(h, \pi) \leq 0,
\end{equation*}
if and only if 
\begin{equation*}
\underaccent{\bar}{\alpha}(\rho, \beta_1, \beta_2, h, \pi) \leq \alpha \leq %
\bar{\alpha}(\rho, \beta_1, \beta_2, h, \pi),
\end{equation*}
where {\small 
\begin{eqnarray*}
\lefteqn{\bar{\alpha}(\rho, \beta_1, \beta_2, h, \pi)} \\
& = & \tfrac{\left(\tfrac{h\pi^2+\pi}{h\pi^2+h+2\pi}\right)(\beta_1(%
\beta_2+1-\rho)h^2\pi + (\beta_1+(1-\rho)(1+\beta_1-\beta_2))h\pi^2 +
(\beta_2+1-\rho)h + (\beta_1+\beta_2-\beta_1\beta_2+(1-\rho)(2-\beta_2))\pi)
- (1-\rho)((\beta_1-\beta_2+1)h\pi^2 + \pi)}{(1-\rho)(\beta_1-\beta_2)(h^2%
\pi+h\pi^2+h+\pi)\left(\tfrac{h\pi^2+\pi}{h\pi^2+h+2\pi}\right) +
\beta_2(\beta_1+1-\rho)h^2\pi + (\beta_1-(1-\rho)(\beta_1-\beta_2))h\pi^2 +
\beta_2 h + (\beta_1+\beta_2-\beta_1\beta_2-(1-\rho)\beta_1)\pi},
\end{eqnarray*}
} and {\small 
\begin{eqnarray*}
\lefteqn{\underaccent{\bar}{\alpha}(\rho, \beta_1, \beta_2, h, \pi)} \\
& = & \tfrac{\left(\tfrac{2\pi}{\pi+1}-\tfrac{h\pi^2+\pi}{h\pi^2+h+2\pi}%
\right)(\beta_1(\beta_2+1-\rho)h^2\pi +
(\beta_1+(1-\rho)(1+\beta_1-\beta_2))h\pi^2 + (\beta_2+1-\rho)h +
(\beta_1+\beta_2-\beta_1\beta_2+(1-\rho)(2-\beta_2))\pi) -
(1-\rho)((\beta_1-\beta_2+1)h\pi^2 + \pi)}{(1-\rho)(\beta_1-\beta_2)(h^2%
\pi+h\pi^2+h+\pi)\left(\tfrac{2\pi}{\pi+1} - \tfrac{h\pi^2+\pi}{h\pi^2+h+2\pi%
}\right)+ \beta_2(\beta_1+1-\rho)h^2\pi +
(\beta_1-(1-\rho)(\beta_1-\beta_2))h\pi^2 + \beta_2 h +
(\beta_1+\beta_2-\beta_1\beta_2-(1-\rho)\beta_1)\pi}. 
\end{eqnarray*}
} 

\noindent In what follows, we show that 
\begin{equation}\label{eq:bound}
\bar{\alpha}(\rho, \beta_1, \beta_2, h, \pi) \in (0, 1), \quad \textnormal{for all } \rho, \beta_1, \beta_2 \in (0, 1) \textnormal{ and } h, \pi > 1.  
\end{equation}
Indeed, we let 
$N_{\bar{\alpha}}$ and $D_{\bar{\alpha}}$ denote the numerator and denominator 
of $\bar{\alpha}(\rho,\beta_1,\beta_2,h,\pi)$, respectively. A direct rearrangement 
yields
\begin{equation*}
\begin{array}{rcl}
N_{\bar{\alpha}} & = & \tfrac{\pi(\beta_1\pi((h\pi+1)^2+(1-\rho)h\pi(h^2-1)) %
+\beta_2((h+\pi)(h\pi+1)+(1-\rho)\pi(h^2-1))+\beta_1\beta_2\pi(h^2-1)(h\pi+1))}{h\pi^2+h+2\pi}, \\
D_{\bar{\alpha}} & = & \beta_1\beta_2\pi(h^2-1)+\tfrac{(\beta_1\pi(h\pi+1)+\beta_2(h+\pi))(h^2\pi(1-\rho)+h\pi^2+h+\pi(1+\rho))}{h\pi^2+h+2\pi}.
\end{array}
\end{equation*}
Because $\rho,\beta_1,\beta_2\in(0,1)$ and $h,\pi>1$, we have
\begin{equation*}
1-\rho>0, \quad h\pi^2+h+2\pi>0,\quad h^2-1>0,\quad h\pi+1>0,\quad h+\pi>0,
\end{equation*}
This implies that $N_{\bar{\alpha}}>0$ and $D_{\bar{\alpha}}>0$. Thus, we have $\bar{\alpha}(\rho,\beta_1,\beta_2,h,\pi)>0$. 

We also have
\begin{equation*}
D_{\bar{\alpha}}-N_{\bar{\alpha}} = \tfrac{\beta_1\pi((1-\rho)\pi(h^2-1)+h^2\pi+h\pi^2+h+\pi)+\beta_1\beta_2\pi(h^2-1)(h+\pi)+\beta_2(h\pi(1-\rho)(h^2-1)+(h+\pi)^2)}{h\pi^2+h+2\pi}. 
\end{equation*}
Because $\rho,\beta_1,\beta_2\in(0,1)$ and $h,\pi>1$, we have
\begin{equation*}
(1-\rho)\pi(h^2-1)+h^2\pi+h\pi^2+h+\pi>0, \quad \pi(h^2-1)(h+\pi)>0, \quad h\pi(1-\rho)(h^2-1)+(h+\pi)^2>0. 
\end{equation*}
This implies that $D_{\bar{\alpha}}-N_{\bar{\alpha}}>0$. Because $N_{\bar{\alpha}}>0$ and $D_{\bar{\alpha}}>0$, we have $\bar{\alpha}(\rho,\beta_1,\beta_2,h,\pi)<1$. Putting these 
pieces together yields Eq.~\eqref{eq:bound}. 

Because $\underaccent{\bar}{\alpha}(\rho,\beta_1,\beta_2,h,\pi) < \bar{\alpha}(\rho,\beta_1,\beta_2,h,\pi)\in(0,1)$ for all $\rho, \beta_1, \beta_2 \in (0, 1)$ and $h, \pi>1$ (see Eq.~\eqref{eq:bias_order}), 
the interval $[\max\{0,\underaccent{\bar}{\alpha}(\rho,\beta_1,\beta_2,h,\pi)\},%
\bar{\alpha}(\rho,\beta_1,\beta_2,h,\pi)]$ is nonempty. $\qed$

Proposition~\ref{prop:known} shows that improvement requires correction, 
not simply more weight on minority signals. In the two-island environment, 
the no-AI benchmark overweights majority information because beliefs 
circulate disproportionately within the larger group. Lowering $\alpha$ helps 
only if it offsets this distortion by the right amount: if $\alpha$ is too high, 
the aggregator reinforces majority dominance, while if $\alpha$ is too low, it 
over-corrects toward the minority island. The beneficial set is therefore 
an interior interval rather than a monotone region. This is a pointwise result 
for a fixed, known environment $(h,\pi,\beta_1,\beta_2)$; it does not imply 
that the same training weights improve learning robustly across nearby 
environments.
\end{document}